\documentclass[a4paper]{article}
\pdfoutput=1

\usepackage[T1]{fontenc}
\usepackage[utf8]{inputenc}
\usepackage{lmodern}

\usepackage[title]{appendix}

\usepackage[final]{microtype}

\usepackage[english]{babel}

\usepackage{amsmath,amssymb}

\usepackage[amsmath,amsthm,thmmarks,hyperref]{ntheorem}

\usepackage{mathrsfs}

\usepackage{pifont}

\usepackage{authblk}

\usepackage[a4paper,top=2.5cm,right=2.0cm,bottom=2.0cm,left=2.0cm,centering]{geometry}

\usepackage[final]{graphicx}

\usepackage{subcaption}

\usepackage[numbers,sort]{natbib}

\usepackage{hyperref}
\hypersetup{%
pdfauthor={Georg Gottlob, Enrico Malizia},%
pdftitle={Achieving New Upper Bounds for the Hypergraph Duality Problem through Logic},%
pdfsubject={Computational Complexity of the Hypergraph Duality Problem},%
pdfkeywords={hypergraphs, hypergraph transversals, hypergraph duality, DNF duality, complexity, first order logic, logical analysis of algorithms, TC0, first order logic with counting quantifiers, algorithms, limited nondeterminism, graph theory}%
}

\usepackage{doi}

\usepackage{algorithm}
\usepackage[noend]{algpseudocode}

\usepackage{xspace}

\usepackage{enumitem}

\usepackage[capitalise,noabbrev]{cleveref}

\makeatletter
\newcounter{algorithmicH}% New algorithmic-like hyperref counter
\let\oldalgorithmic\algorithmic
\renewcommand{\algorithmic}{%
  \stepcounter{algorithmicH}% Step counter
  \oldalgorithmic}% Do what was always done with algorithmic environment
\renewcommand{\theHALG@line}{ALG@line.\thealgorithmicH.\arabic{ALG@line}}
\makeatother

\newcommand{\nbdash}{\nobreakdash}

\newcommand{\cmark}{\ding{51}}
\newcommand{\xmark}{\ding{55}}

\newcommand{\Wlog}{w.l.o.g.\ }
\newcommand{\Wrt}{w.r.t.\ }

\newcommand{\cnf}{\mathit{CNF}}
\newcommand{\dnf}{\mathit{DNF}}
\newcommand{\DM}{\updownarrow}
\newcommand{\Hyp}{\mathscr{H}}

\newcommand{\tuple}[1]{\langle #1 \rangle}
\newcommand{\assign}[1]{\tuple{#1}}
\newcommand{\emptyassign}{\sigma_\varepsilon}
\newcommand{\addassign}{+}

\newcommand{\In}{\mathit{In}}
\newcommand{\Ex}{\mathit{Ex}}
\newcommand{\reverse}{\overline}

\renewcommand{\emptyset}{\varnothing}
\newcommand{\wt}{\widetilde}

\newcommand{\wh}{\widehat}
\newcommand{\compl}[1]{{\overline{#1}}}
\newcommand{\G}{\mathcal{G}}
\renewcommand{\H}{\mathcal{H}}
\newcommand{\F}{\mathcal{F}}
\newcommand{\Tr}{\mathit{tr}}
\newcommand{\NullID}{\ensuremath{\mathit{NIL}}\xspace}
\newcommand{\C}{\mathcal{C}}
\newcommand{\A}{\mathcal{A}}
\newcommand{\I}{\mathcal{I}}
\newcommand{\InputSize}{N}
\newcommand{\NodesTree}{N}

\newcommand{\Tree}{\mathcal{T}}

\newcommand{\Com}{\mathit{Com}}
\newcommand{\Sep}{\mathit{Sep}}

\newcommand{\PathToNode}{\mathscr{N}}
\newcommand{\ExcNode}[1]{\ensuremath{-#1}}
\newcommand{\IncNodeCrit}[2]{\ensuremath{(#1,#2)}}

\newcommand{\SetOfLogSeqs}{\mathscr{S}\ensuremath{^{\log}}}

\newcommand{\LabelsOfHyp}{\mathscr{L}}

\newcommand{\dblsubeq}{\sqsubseteq}
\newcommand{\dblsub}{\sqsubset}

\newcommand{\Freq}{\mathit{Freq}}
\newcommand{\Infreq}{\mathit{Infreq}}

\newcommand{\DUAL}{{\mbox{\sc Dual}}\xspace}
\newcommand{\coDUAL}{\ensuremath{\overline{\mbox{\sc Dual}}}\xspace}
\newcommand{\DualProbHyp}{\DUAL}
\newcommand{\NonDualProbHyp}{\coDUAL}
\newcommand{\Dual}{\DUAL}
\newcommand{\coDual}{\coDUAL}

\newcommand{\valtrue}{\textnormal{\texttt{true}}\xspace}
\newcommand{\valfalse}{\textnormal{\texttt{false}}\xspace}
\newcommand{\NDAlg}{ND-NotDual}
\newcommand{\ComputeNT}{ComputeNT}
\newcommand{\CheckIP}{Check-Intersection\-Property}
\newcommand{\CheckSimpleIP}{Check-Simple-And-Intersection}
\newcommand{\CheckConsistencySet}{Check-Consistency}

\newcommand{\CheckDoubleWitnessAug}{Check-Aug-DoubleWitness}

\newcommand{\Vertex}{\mathit{Vertex}}
\newcommand{\Hyper}{\mathit{Hyp}}
\newcommand{\Edge}{\mathit{EdgeOf}}
\newcommand{\Incidence}{\mathit{In}}
\newcommand{\Sone}{S_1}
\newcommand{\Stwo}{S_2}

\newcommand{\Plus}{\mathit{PLUS}}
\newcommand{\Succ}{\mathit{SUCC}}
\newcommand{\Prec}{<}
\newcommand{\Bit}{\mathit{BIT}}

\newcommand{\Simple}{\mathit{simple}}

\newcommand{\IntersectionProperty}{\mathit{intersection\mhyphen{}property}}

\newcommand{\NotConsistency}{\mathit{inconsistent}}
\newcommand{\CongruentGuess}{\mathit{congruentGuess}}
\newcommand{\ConsistentGuess}{\mathit{consistentGuess}}

\newcommand{\EGuess}{\mathit{E\mhyphen{}guess}}
\newcommand{\IGuess}{\mathit{I\mhyphen{}guess}}

\newcommand{\ComHAfterGuess}{\mathit{com}}

\newcommand{\Half}{\mathit{half}}
\newcommand{\CountCompEdges}{\mathit{count\mhyphen{}com}}
\newcommand{\CountCompEdgesIncl}{\mathit{count\mhyphen{}com\mhyphen{}inc}}
\newcommand{\ALHCompEdges}{\mathit{freq}}
\newcommand{\EModified}{\mathit{E}\mhyphen{}aug}
\newcommand{\IModified}{\mathit{I}\mhyphen{}aug}
\newcommand{\ComHModified}{\mathit{com\mhyphen{}aug}}
\newcommand{\SepGModified}{\mathit{sep\mhyphen{}aug}}
\newcommand{\CheckGuessAugDoubleWitness}{\mathit{CheckGuessAugDoubleWitness}}

\newcommand{\Alg}{Det-Dual}
\newcommand{\DetNewTrAlg}{New-Tr}

\newcommand{\FreeSet}{\mathit{Free}}
\newcommand{\IncAlg}{\mathit{Incl}}
\newcommand{\ExcAlg}{\mathit{Excl}}
\newcommand{\FreeAlg}{\mathit{Free}}
\newcommand{\ComAlg}{\Com_\H}
\newcommand{\SepAlg}{\Sep_\G}

\newcommand{\TC}[1]{\textnormal{\mbox{{TC}$^{#1}$}}\xspace}
\newcommand{\PP}{\textnormal{{PTIME}}\xspace}
\newcommand{\GC}[2]{\textnormal{\textrm{{GC}\ensuremath{(}#1{},\linebreak[0] #2{}\ensuremath{)}}}\xspace}

\newcommand{\LogSpace}{\textnormal{{LOGSPACE}}\xspace}
\newcommand{\DSpace}{\textnormal{{DSPACE}}\xspace}

\newcommand{\FO}{FO\xspace}
\newcommand{\FOM}{FOM\xspace}
\newcommand{\FOC}{FO(COUNT)\xspace}

\newcommand{\pol}{{\rm pol}}
\newcommand{\SUPER}{\mbox{\rm $[\![$LOGSPACE$_\pol]\!]^{\log}$}}
\newcommand{\dspace}[1]{{\mbox{\rm\sc DSPACE[}#1{\rm]}}}
\mathchardef\mhyphen="2D

\theoremstyle{plain}
\newtheorem{theorem}{Theorem}[section]
\newtheorem{lemma}[theorem]{Lemma}

\newtheorem{corollary}[theorem]{Corollary}

\theoremstyle{definition}

\title{Achieving New Upper Bounds for the Hypergraph Duality Problem through Logic\footnote{This is an extended version of a paper that will shortly appear in SIAM Journal on Computing.}}
\date{}

\author[1]{Georg Gottlob}
\author[2]{Enrico Malizia}
\affil[1]{%
Department of Computer Science\\
University of Oxford, UK%\\
}
\affil[2]{%
Department of Computer Science\\
University of Exeter, UK%\\
}

\begin{document}
\maketitle

\begin{abstract}
  The hypergraph duality problem \DualProbHyp is defined as follows: given two simple hypergraphs $\G$ and $\H$, decide whether $\H$ consists precisely of all minimal transversals of $\G$ (in which case we say that $\G$ is the dual of $\H$, or, equivalently, the transversal hypergraph of $\H$).
  This problem is equivalent to deciding whether two given non-redundant monotone DNFs/CNFs are dual.
  It is known that $\coDual$, the complementary problem to \DualProbHyp, is in \GC{$\log^2 n$}{\PP}, where \GC{$f(n)$}{$\C$} denotes the complexity class of all problems that after a nondeterministic guess of $O(f(n))$ bits can be decided (checked) within complexity class $\C$.
  It was conjectured that $\coDual$ is in \GC{$\log^2 n$}{\LogSpace}.
  In this paper we prove this conjecture and actually place the $\coDual$ problem into the complexity class \GC{$\log^2 n$}{\TC0} which is a subclass of \GC{$\log^2 n$}{\LogSpace}.
  We here refer to the logtime-uniform version of \TC0, which corresponds to FO(COUNT), i.e., first order logic augmented by counting quantifiers.
  We achieve the latter bound in two steps.
  First, based on existing problem decomposition methods, we develop a new nondeterministic algorithm for $\coDual$ that requires to guess $O(\log^2 n)$ bits.
  We then proceed by a logical analysis of this algorithm, allowing us to formulate its deterministic part in FO(COUNT).
  From this result, by the well known inclusion $\TC0\subseteq\LogSpace$, it follows that \Dual belongs also to $\DSpace[\log^2 n]$.
  Finally, by exploiting the principles on which the proposed nondeterministic algorithm is based, we devise a deterministic algorithm that, given two hypergraphs $\G$ and $\H$, computes in quadratic logspace a transversal of $\G$ missing in $\H$.
\end{abstract}

\section{Introduction}\label{sec:intro}
The hypergraph duality problem $\Dual$ is one of the most mysterious and challenging decision problems of Computer Science, as its complexity has been intensively investigated for almost 40 years without any indication that the problem is tractable, nor any evidence whatsoever, why it should be intractable.
Apart from a few significant upper bounds, which we review below, and a large number of restrictions that make the problem tractable, progress on pinpointing the complexity of $\Dual$ has been rather slow.
So far, the problem has been placed in relatively low complexity nondeterministic classes within coNP.
It is the aim of this paper to further narrow it down by using logical methods.

\medskip

{\bf The hypergraph duality problem.}
A hypergraph $\G$ consists of a finite set $V$ of vertices and a set $E\subseteq 2^V$ of (hyper)edges.
$\G$ is {\em simple} (or {\em Sperner}) if none of its edges is contained in any other of its edges.
A {\em transversal} or {\em hitting set} of a hypergraph $\G=\tuple{V,E}$ is a subset of $V$ that meets every edge in $E$.
A transversal of $\G$ is \emph{minimal}, if none of its proper subsets is a transversal.
The set of minimal transversals of a hypergraph $\G=\tuple{V,E}$ is denoted by $tr(\G)$.
Note that $tr(\G)$, which is referred to as the {\em dual\footnote{Note that sometimes in the literature the dual hypergraph of $\G$ was defined as the hypergraph derived from $\G$ in which the roles of the vertices and the edges are ``interchanged'' (see, e.g.,~\citep{Berge1989,Scheinerman1997}), and this is different from the transversal hypergraph. Nevertheless, lately in the literature the name ``dual hypergraph'' has been used with the meaning of ``transversal hypergraph'' (as in, e.g.,~\citep{EGMSurvey,Khachiyan2006,Boros2009,gott13,Khachiyan2007,Elbassioni2008}).}} of $\G$ or also as the {\em transversal hypergraph} of $\G$, defines itself a hypergraph on the vertex set $V$.
The decision problem $\Dual$ is now easily defined as follows: Given two simple hypergraphs $\G$ and $\H$ over vertex set $V$, decide whether $\H=tr(\G)$.

\begin{figure}[!ht]
  \centering%
  \includegraphics[width=0.42\textwidth]{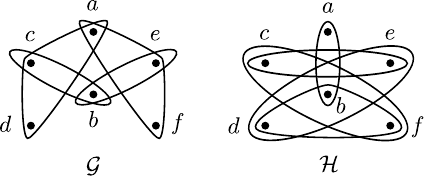}
  \caption{Hypergraph $\mathcal{G}$ and its transversal hypergraph $\mathcal{H}$.}\label{fig1}
\end{figure}

An example of a hypergraph and its dual is given in \cref{fig1}.
It is well-known that the duality problem has a nice symmetry property~\citep{Berge1989}: if $\G$ and $\H$ are simple hypergraphs over vertex set $V$, then $\H=tr(\G)$ iff $\G=tr(\H)$, and in this case $\G$ and $\H$ are said to be dual.
The $\Dual$ problem is also tightly related to the problem of actually {\em computing} $tr(\G)$ for an input hypergraph $\G$.
In fact, it is known that the computation problem is feasible in total polynomial time, that is, in time polynomial in $|\G|+|tr(\G)|$, if and only if $\Dual$ is solvable in polynomial time~\citep{bioc-ibar-95}.
These and several other properties of the duality problem are reviewed and discussed in~\citep{EGMSurvey,EG95,hagenDiss,EGM03}, where also many original references can be found.

\medskip

{\bf Applications of hypergraph duality.}
The $\Dual$ problem and its computational variant have a tremendous number of applications.
They range from data mining~\citep{guno-etal-97,ManiToiv,BGKM2002,BGKM2003}, functional dependency inference~\citep{mannila1992complexity,mannila1994algorithms,gott-libk-90}, and machine learning, in particular, learning monotone Boolean CNFs and DNFs with membership queries~\citep{guno-etal-97,mishra1997generating}, to model-based diagnosis~\citep{ReiterDiag87,grei-smit-wilk-90}, computing a Horn approximation to a non-Horn theory~\citep{kavv-etal-93,gogi-etal-98}, computing minimal abductive explanations to observations~\citep{Eiter-Makino-Abduction}, and computational biology, for example, discovering of metabolic networks and engineering of drugs preventing the production of toxic metabolites in cells~\cite{Klamt2004,Klamt2006}.
Surveys of these and other applications as well as further references can be found in~\citep{EGJelia,EG95,hagenDiss,Klamt2009}.

The simplest and foremost applications relevant to logic and hardware design are DNF duality testing and its computational version, DNF dualization.
A pair of Boolean formulas $f(x_1,x_2,\ldots,x_n)$ and $g(x_1,x_2,\ldots,x_n)$ on pro\-po\-si\-tio\-nal variables $x_1,x_2,\ldots,x_n$ are {\it dual} if, for any Boolean assignment to variables $x_1,\dots,x_n$,
$$f(x_1,x_2,\ldots,x_n)\equiv \neg g(\neg x_1,\neg x_2,\ldots,\neg x_n).$$
A monotone DNF is {\em irredundant} if the set of variables in none of its disjuncts is covered by the variable set of any other disjunct.
The {\em duality testing problem} is the problem of testing whether two irredundant monotone DNFs $f$ and $g$ are dual.
It is well-known and easy to see that monotone DNF duality and $\Dual$ are actually the same problem.\footnote{In fact, in the literature, the hypergraph transversal problem was tackled interchangeably from the perspective of monotone Boolean formula dualization, or from the perspective of hypergraphs. Readers wanting to know more about the relationships between some of the different perspectives adopted in the literature to deal with the \DualProbHyp problem are referred to \cref{sec:perspectives}.}
Two hypergraphs $\G$ and $\H$ are dual iff their associated DNFs $\G^*$ and $\H^*$ are dual, where the DNF $\F^*$ associated with a hypergraph $\F=\tuple{V,E}$ is $\bigvee_{e\in E}\bigwedge_{v\in e}v$, where, obviously, vertices $v\in V$ are interpreted
as propositional variables.
For example, the hypergraphs $\G$ and $\H$ of \cref{fig1} give rise to DNFs
\begin{gather*}
\G^*=(a\land c\land d)\lor(a\land e\land f)\lor(c\land b)\lor(e\land b), \text{ and}\\
\H^*=(a\land b)\lor(c\land e)\lor(c\land b\land f)\lor(e\land b\land d)\lor(d\land b\land f),
\end{gather*}
which are indeed mutually dual.
The duality problem for irredundant monotone DNFs corresponds, in turn, to $\Dual$, and the problem instances $\tuple{\F,\G}$ and $\tuple{\F^*,G^*}$ can be inter-translated by extremely low-level reductions, in particular, LOGTIME reductions, and even projection reductions.
In many publications, the $\Dual$ problem is thus right away introduced as the problem of duality checking for irredundant monotone DNFs.
An equivalent problem is the problem of checking whether a monotone CNF and a monotone DNF are logically equivalent.

\medskip

{\bf Previous Complexity bounds}.
$\Dual$ is easily seen to reside in coNP.
In fact, in order to show that a $\Dual$ instance is a ``no''-instance, it suffices to show that either some edge of one of the two hypergraphs is not a minimal transversal of the other hypergraph (which is feasible in polynomial time), or to find (guess and check) a missing transversal to one of the input hypergraphs.
The complement $\coDual$ of $\Dual$ is therefore in NP.
In their landmark paper, \citet*{fred-khac-96} have shown that {\DUAL} is in DTIME[$n^{o(\log n)}$], more precisely, that it is contained in DTIME[$n^{4\chi(n)+O(1)}$], where $\chi(n)$ is defined by $\chi(n)^{\chi(n)}=n$.
Note that $\chi(n)\sim\log n/\log\log n = o(\log n)$.

Let \GC{$f(n)$}{$\C$} denote the complexity class of all problems that after a nondeterministic guess of $O(f(n))$ bits can be decided (checked) in complexity class $\C$.
\citet*{EGM03}, and independently, \citet*{kavv-stav-03} have shown that $\coDUAL$ is in \GC{$\log^2 n$}{\PP}; note that this class is also known as $\beta_2$P, see~\citep{Goldsmith}.

Recently, the nondeterministic bound for $\coDual$ was further pushed down to \GC{$\log^2 n$}{$\SUPER$} \citep{gott13}, see \cref{fig2a}.

\begin{figure}
  \centering%
  \begin{subfigure}[b]{0.5\textwidth}
    \includegraphics[width=\textwidth]{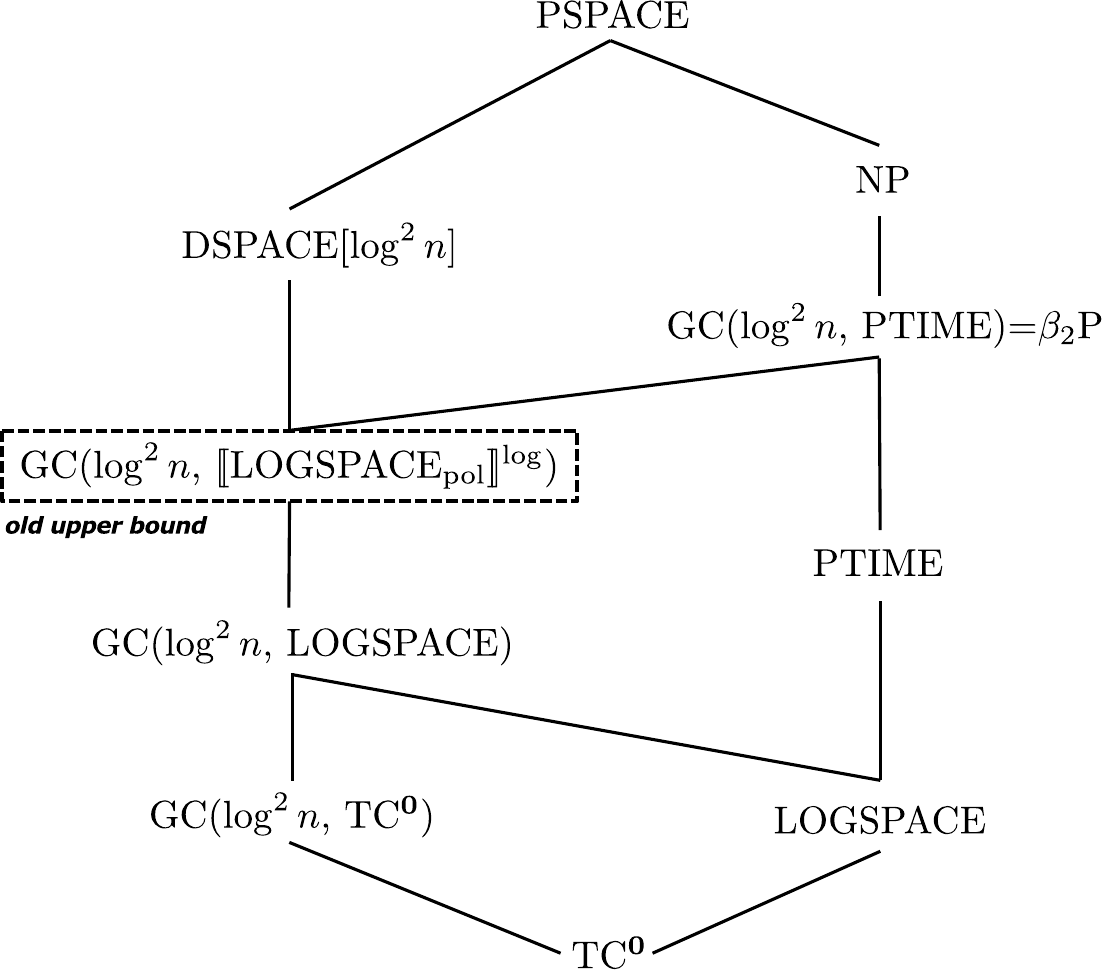}
    \caption{Old upper bound}
    \label{fig2a}
  \end{subfigure}%
  \hspace{-0.4cm}
  \begin{subfigure}[b]{0.5\textwidth}
    \includegraphics[width=\textwidth]{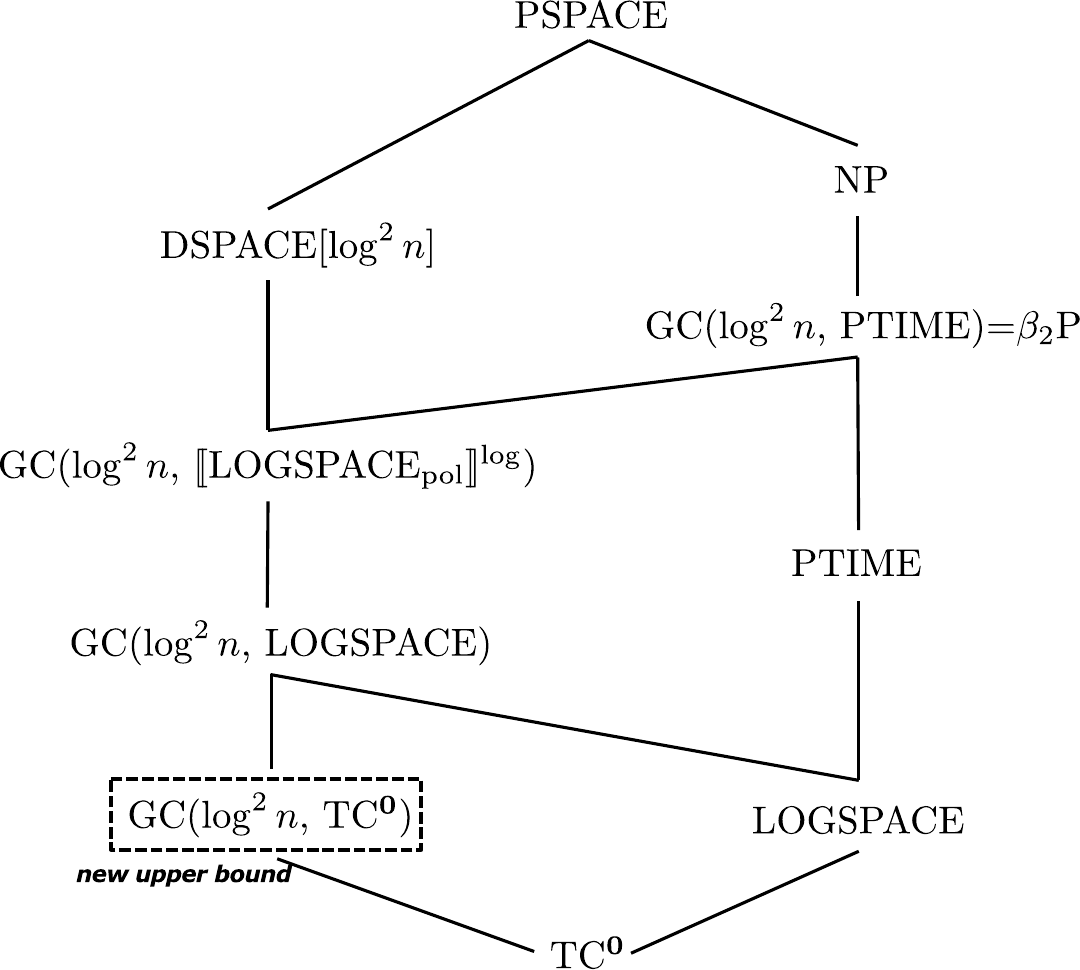}
    \caption{New upper bound}
    \label{fig2b}
  \end{subfigure}
  \caption{Complexity bound improvement obtained in this paper.}\label{fig2}
\end{figure}

A precise definition of $\SUPER$ is given in~\citep{gott13}.
We will not make use of this class in the technical part of the present paper.
Informally, $\SUPER$ contains those problems $\pi$ for which there exists a logspace-transducer $T$, a polynomial $p$, and a function $f$ in $O(\log n)$, such that each $\pi$-instance $I$ of size $n=|I|$ can be reduced by the $f(n)$-fold composition $T^{f(n)}$ of $T$ to a decision problem in LOGSPACE, where the size of all intermediate results $T^i(I)$, for $1\leq i\leq f(n)$, is polynomially bounded by $p(n)$.
For the relationship of \GC{$\log^2 n$}{$\SUPER$} to other classes, see \cref{fig2a}.
In~\citep{gott13}, it was shown that \GC{$\log^2 n$}{$\SUPER$} is not only a subclass of \GC{$\log^2 n$}{\PP}, but also of $\dspace{\log^2 n}$, i.e., of quadratic logspace.
Therefore, as also proven in a new and more direct way in the present paper, $\Dual$ is in $\dspace{\log^2 n}$ (\cref{corol:DUAL_in_DSPACE_log2_n}), and it is thus most unlikely for $\Dual$ to be PTIME-hard, which answered a previously long standing question.
Given that PTIME and $\dspace{\log^2 n}$ are believed to be incomparable, it is also rather unlikely that $\Dual$ is closely related to another interesting logical problem of open complexity, namely, to validity-checking for the modal $\mu$-calculus, or, equivalently, to the winner determination problem for parity games~\citep{henzinger2006universal,jurdzinski1998deciding}, as these latter problems are PTIME-hard, but in NP\,$\cap$\,coNP.

\medskip

{\bf Main complexity problem tackled.}
In~\citep{gott13} it was asked whether the upper bound of \GC{$\log^2 n$}{$\SUPER$} could be pushed further downwards, and the following conjecture was made:

\medskip

{\em Conjecture (\citep{gott13})\ \
$\coDual\in\,$\GC{$\log^2 n$}{LOGSPACE}.}

\medskip

It was unclear, however, how to prove this conjecture based on the algorithms and methods used in~\citep{gott13}.
There, a problem decomposition strategy by \citet*{Boros2009} was used, that decomposed an original $\Dual$ instance into a conjunction of smaller instances according to a specific conjunctive self-reduction.
Roughly, this strategy constructs a decomposition tree of logarithmic depth for $\Dual$, each of whose nodes represents a sub-instance of the original instance; more details on decomposition trees are given in \cref{sec:decomposition}.
To prove that the original instance is a ``no''-instance (and thus a ``yes''-instance of $\coDual$), it is sufficient to guess, in that tree, a path $\Pi$ from the root to a single node $v$ associated with a ``no''-sub-instance that can be recognized as such in logarithmic space.
Guessing the path to $v$ can be easily done using $O(\log^2 n)$ nondeterministic bits, but it is totally unclear how to actually {\em compute} the sub-instance associated with node $v$ in LOGSPACE.
In fact, it seems that the only way to compute the sub-instance at node $v$ is to compute---at least implicitly---all intermediate $\Dual$ instances arising on the path from the root to the decomposition node $v$.
This seems to require a logarithmic composition of LOGSPACE transducers, and thus a computation in the complexity class $\SUPER$.
It was therefore totally unclear how $\SUPER$ could be replaced by its subclass LOGSPACE, and new methods were necessary to achieve this goal.

\medskip

{\bf New results: Logic to the rescue.}
To attack the problem, we studied various alternative decomposition strategies for $\Dual$, among them the strategy of \citet{Gaur1999}, which also influenced the method of \citet{Boros2009}.
In the present paper, we build on Gaur's original strategy, as it appears to be the best starting point for our purposes.
However, Gaur's method still does not directly lead to a guess-and-check algorithm whose checking procedure is in LOGSPACE, and thus new techniques needed to be developed.

In a first step, by building creatively on Gaur's deterministic decomposition strategy~\citep{Gaur1999}, we develop a new nondeterministic guess-and-check algorithm \textsc{\NDAlg} for $\coDUAL$, that is specifically geared towards a computationally simple checking part.
In particular, the checking part of \textsc{\NDAlg} avoids certain obstructive steps that would require more memory than just plain LOGSPACE, such as the successive minimization of hypergraphs in sub-instances of the decomposition (as used by \citet{Boros2009}) and the performance of counting operations between subsequent decomposition steps so to determine sets of vertices to be included in a new transversal (as used by \citet{Gaur1999}).
Our new approach is thus influenced by Gaur's, but differs noticeably from it, as well as from the algorithm of \citeauthor{Boros2009}.

In a second step, we proceed with a careful logical analysis of the checking part of \textsc{\NDAlg}.
We transform all sub-tasks of \textsc{\NDAlg} into logical formulas.
However, it turns out that first order logic (FO) is not sufficient, as an essential step of the checking phase of \textsc{\NDAlg} is to check for specific hypergraph vertices $v$ whether $v$ is contained in at least half of the hyperedges of some hypergraph.
To account for this, we need to resort to FO(COUNT), which augments FO with counting quantifiers.
Note that we could have used in a similar way FOM, i.e., FO augmented by majority quantifiers, as FO(COUNT) and FOM have the same expressive power~\citep{Immerman1999}.
By putting all pieces together, we succeed to describe the entire checking phase by a single fixed FO(COUNT) formula that has to be evaluated over the input $\coDual$ instance.
Note that FO(COUNT) model-checking is complete for logtime-uniform $\TC0$.

In summary, by putting the guessing and checking parts together, we achieve as main theorem a complexity result that is actually better than the one conjectured:

\medskip

{\em Theorem.} \  $\coDual\in\GC{$\log^2 n$}{\TC0}$.

\medskip

By the well-known inclusion $\TC0\subseteq \mbox{\rm LOGSPACE}$, we immediately obtain a corollary that proves the above mentioned conjecture:

\medskip

{\em Corollary A.} \  \, $\coDual\in\GC{$\log^2 n$}{\rm LOGSPACE}$.

\medskip

Moreover, by the inclusion $\GC{$\log^2 n$}{\LogSpace}\subseteq\DSpace[\log^2 n]$, and the fact that \DSpace[$\log^2 n$] is closed under complement, we can easily obtain as a simple corollary that:

\medskip

{\em Corollary B.} \  \, $\Dual\in\DSpace[\log^2 n]$.

\medskip

To conclude, by an easy adaptation of our algorithm \textsc{\NDAlg} we devise a simple deterministic algorithm \textsc{\ComputeNT} to \emph{compute} a new (not necessarily minimal) transversal in quadratic logspace.

\medskip

{\bf Significance of the new results and directions for future research.}
The progress achieved in this paper is summarized in \cref{fig2}, whose left part (\cref{fig2a}) shows the previous state of knowledge about the complexity, while the right part (\cref{fig2b}) depicts the current state of knowledge we have achieved.
We have significantly narrowed down the ``search space'' for the precise complexity of $\Dual$ (or $\coDual$).
We believe that our new results are of value to anybody studying the complexity of this interesting problem.
In particular, the connection to logic opens new avenues for such studies.
First, our results show where to dig for tighter bounds.
It may be rewarding to study subclasses of \GC{$\log^2 n$}{\TC0}, and in particular, {\em logically defined subclasses} that replace $\TC0$ by low-level prefix classes of \FOC.
Classes of this type can be found in~\citep{Cai1997,cook2010logical,Grohe2006,Toran1988,Toran1989}.
More details on this will be given in the conclusive remarks in \cref{sec:concl_and_future}.
Second, the membership of $\coDual$ in \GC{$\log^2 n$}{\TC0} provides valuable information for those trying to prove hardness results for $\Dual$, i.e., to reduce some presumably intractable problem $X$ to $\Dual$.
Our results restrict the search space to be explored to hunt for such a problem $X$.
Moreover, given that LOGSPACE is not known to be in \GC{$\log^2 n$}{\TC0}, and given that it is not generally believed that \GC{$\log^2 n$}{\TC0} contains LOGSPACE-hard problems (under logtime reductions), our results suggest that LOGSPACE-hard problems are rather unlikely to reduce to $\Dual$, and that it may thus be advisable to look for a problem $X$ that is not (known to be) LOGSPACE-hard, in order to find a lower bound for $\Dual$.
Our new results are of theoretical nature.
This does not rule out the possibility that they may be used for improving practical algorithms, but this has yet to be investigated.
Finally, we believe that the methods presented in this paper are a compelling example of how logic and descriptive complexity theory can be used together with suitable problem decomposition methods to achieve new complexity results for a concrete decision problem.

\medskip

{\bf Organization of the paper.}
After some preliminaries in \cref{sec:prelim}, we discuss problem decomposition strategies and introduce the concept of a decomposition tree for $\coDual$ in \cref{sec:decomposition}.
Based on this, in \cref{sec:new_upper_bound} we present the nondeterministic algorithm \textsc{\NDAlg} for $\coDual$, prove it correct, and then analyze this algorithm to derive our main complexity results.
To conclude, by exploiting the method used in the nondeterministic algorithm, we present a deterministic algorithm to actually compute a new (not necessarily minimal) transversal in quadratic logspace.

\section{Preliminaries}\label{sec:prelim}
In what follows, when we identify a hypergraph $\G$ with its edge\nbdash-set $E$ by writing $\G=E$, we mean $\G=\tuple{\bigcup_{G\in E}G,E}$.
By writing $G\in\G$ we mean $G\in E$.
Generally, we denote by $V$ the set of vertices of a hypergraph, and, if not stated otherwise, pairs of hypergraphs $\G$ and $\H$, and the hypergraphs of a \DUAL instance $\tuple{\G,\H}$, are assumed to have the same set of vertices.
The number of edges of $\G$ is denoted by $|\G|$, and, given an instance $\tuple{\G,\H}$ of \DUAL, $m$ is the total number $|\G|+|\H|$ of edges of $\G$ and $\H$.
By $\|\G\|$ we denote the size of the hypergraph $\G$, that is the space (in terms of the number of bits) required to represent $\G$.
It is reasonable to assume that a hypergraph $\G$ is represented through the adjacency lists of its edges (i.e., each edge $G$ of $\G$ is represented through the list of the vertices belonging to $G$).
It is easy to see that $|V|\leq \|\G\|$, and $|\G|\leq\|\G\|$.
We denote by $\InputSize=\|\G\|+\|\H\|$ the size of the input of the \DUAL problem.

We say that $\G$ is an \emph{empty hypergraph} if $\G$ does not have any edge, and $\G$ is an \emph{empty\nbdash-edge hypergraph} if $\G$ contains only an empty edge (i.e., an edge without vertices in it).
Note that empty and empty\nbdash-edge hypergraphs can actually contain vertices, and thus there are various empty and empty\nbdash-edge hypergraphs.
For notational convenience, by $\G = \emptyset$ we indicate that $\G$ is an empty hypergraph, while by $\G = \{\emptyset\}$ we indicate that $\G$ is an empty\nbdash-edge hypergraph.
However, with these notations we do not mean that these hypergraphs do not have vertices.
The notation $\emptyset \in \G$ clearly means that $\G$ contains an empty edge.
Observe that, if $\G = \emptyset$, then any set of vertices is a transversal of $\G$.
For this reason, there is only one minimal transversal of $\G$ and it is the empty set.
On the other hand, if $\emptyset \in \G$, then there is no transversal of $\G$ at all.
Hence, by definition, the dual of an empty hypergraph is an empty\nbdash-edge hypergraph, and vice\nbdash-versa~\cite{EG95}.
Two hypergraphs $\G$ and $\H$ over the same vertex set are \emph{trivially dual}, if one of them is an empty hypergraph and the other is an empty\nbdash-edge hypergraph.

Given a hypergraph $\G$ and a set of vertices $T$, a vertex $v\in T$ is \emph{critical} in $T$ (\Wrt $\G$) if there is an edge $G\in\G$ such that $G\cap T=\{v\}$.
We say that $G$ witnesses the criticality of $v$ in $T$.
Observe that, if $v$ is a critical vertex in $T$, $v$ may have various witnesses of its criticality, that is, more than one edge of $\G$ can intersect $T$ only on $v$.

A set of vertices $S$ is an \emph{independent set} of a hypergraph $\G$ if, for all $G\in\G$, $G\not\subseteq S$.
If $\G$ and $\H$ are two hypergraphs, a set of vertices $T$ is a \emph{new transversal} of $\G$ \Wrt $\H$\footnote{We will often omit ``\Wrt $\H$'' when the hypergraph $\H$ we are referring to is understood.} if $T$ is a transversal of $\G$, and $T$ is also an independent set of $\H$.
Intuitively, a new transversal $T$ of $\G$ is a transversal of $\G$ missing in $\H$.
More formally, $T$ is a new transversal of $\G$ \Wrt $\H$ if, for any transversal $T'\subseteq T$ of $\G$, $T'\notin\H$.
Note that a new transversal is not necessarily a minimal transversal, however it contains a new minimal transversal.

In the following, we state the main properties about transversals to be used in this paper.
Some of the following properties are already known (see, e.g., \cite{Berge1989,fred-khac-96,Gaur2004,Boros2009,Elbassioni2008,EG95,Garcia-Molina1985}), and some of them were stated over Boolean formulas.
We state over hypergraphs all the properties relevant for us, and for completeness we prove them in \cref{sec:app_prove_property_transversal}.

\begin{lemma}\label{lemma:transMin-ogniVertexCritico}
Let $\G$ be a hypergraph, and let $T\subseteq V$ be a transversal of $\G$.
Then, $T$ is a minimal transversal of $\G$ if and only if every vertex $v\in T$ is critical (and hence there is an edge $G_v\in\G$ witnessing so).
\end{lemma}

For a set of vertices $T\subseteq V$, let $\compl{T}$ denote $V\setminus T$, i.e., the \emph{complement} of $T$ in $V$.

\begin{lemma}\label{lemma:newTransversals_complementari}
Let $\G$ and $\H$ be two hypergraphs.
A set of vertices $T\subseteq V$ is a new transversal of $\G$ \Wrt $\H$ if and only if $\compl{T}$ is a new transversal of $\H$ \Wrt $\G$.
\end{lemma}

We say that hypergraphs $\G$ and $\H$ satisfy the \emph{intersection property} if all edges of $\G$ are transversals of $\H$, and thus, vice\nbdash-versa all edges of $\H$ are transversal of $\G$.
Note here that, for the intersection property to hold, the edges of one hypergraph are \emph{not} required to be \emph{minimal} transversal of the other.

\begin{lemma}\label{lemma:no_new_transversal_iff_dual}
Let $\G$ and $\H$ be two hypergraphs.
Then, $\G$ and $\H$ are dual if and only if $\G$ and $\H$ are simple, satisfy the intersection property, and there is no new transversal of $\G$ \Wrt $\H$.
\end{lemma}

\section{Decomposing the \texorpdfstring{\DualProbHyp}{DUAL} problem}\label{sec:decomposition}

\subsection{Decomposition principles}
A way to recognize ``no''\nbdash-instances $\tuple{\G,\H}$ of \DualProbHyp is to find a new transversal of $\G$ \Wrt $\H$, i.e., a transversal of $\G$ that is also an independent set of $\H$.
In fact, many algorithms in the literature follow this approach (see, e.g., \cite{fred-khac-96,EGM03,kavv-stav-03,Elbassioni2008,Boros2009,Gaur1999}).
These algorithms try to build such a new transversal by successively including vertices in and excluding vertices from a candidate for a new transversal.
To give an example, the classical algorithm ``A'' of \citet{fred-khac-96} tries to include a vertex $v$ in a candidate for a new transversal, and if this does not result in a new transversal, then $v$ is excluded.
Moreover if the exclusion of $v$ does not lead to a new transversal, then no new transversal exists which is coherent with the choices having made before considering vertex $v$.
(If $v$ is the first vertex considered, then there is no new transversal at all.)

We speak about \emph{included} and \emph{excluded} vertices because most of the algorithms proposed in the literature implicitly or explicitly keep track of two sets:
the set of the vertices considered included in and the set of the vertices considered excluded from the attempted new transversal.
Similarly, those algorithms working on Boolean formulas keep track of the truth assignment:
the variables to which \valtrue has been assigned, those to which \valfalse has been assigned, and (obviously) those to which no Boolean value has been assigned yet.

If $\G$ and $\H$ are two hypergraphs, an \emph{assignment} $\sigma=\assign{\In,\Ex}$ is a pair of subsets of $V$ such that $\In\cap \Ex=\emptyset$.
Intuitively, set $\In$ contains the vertices considered included in (or, inside) an attempted new transversal $T\supseteq\In$ of $\G$ \Wrt $\H$, while set $\Ex$ contains the vertices considered excluded from (or, outside) $T$ (i.e., $T\cap\Ex=\emptyset$).
We say that a vertex $v\in V$ is free in (or of) an assignment $\sigma=\assign{\In,\Ex}$, if $v\notin \In$ and $v\notin \Ex$.
Note that the \emph{empty assignment} $\emptyassign=\tuple{\emptyset,\emptyset}$ is a valid assignment.
Given assignments $\sigma_1=\assign{\In_1,\Ex_1}$ and $\sigma_2=\assign{\In_2,\Ex_2}$, if $\In_1\cap \Ex_2=\emptyset$ and $\Ex_1\cap \In_2=\emptyset$, we denote by $\sigma_1\addassign\sigma_2=\assign{\In_1\cup \In_2,\Ex_1\cup \Ex_2}$ the extension of $\sigma_1$ with $\sigma_2$.
An assignment $\sigma_2=\assign{\In_2,\Ex_2}$ is said to be an \emph{extension} of assignment $\sigma_1=\assign{\In_1,\Ex_1}$, denoted by $\sigma_1\dblsubeq\sigma_2$, whenever $\In_1\subseteq\In_2$ and $\Ex_1\subseteq\Ex_2$.
If $\sigma_1\dblsubeq\sigma_2$ and $\In_1\subset\In_2$ or $\Ex_1\subset\Ex_2$ we say that $\sigma_2$ is a \emph{proper extension} of $\sigma_1$, denoted by $\sigma_1\dblsub\sigma_2$.

Given a set of vertices $S\subseteq V$, the associated assignment is $\sigma_S=\assign{S,\compl{S}}$.
We say that an assignment $\sigma=\assign{\In,\Ex}$ is \emph{coherent} with a set of vertices $S$, and vice\nbdash-versa, whenever $\sigma\dblsubeq\sigma_S$.
This is tantamount to $\In\subseteq S$ and $\Ex\subseteq \compl{S}$ (or, equivalently, $\Ex\cap S=\emptyset)$.
With a slight abuse of notation we denote that an assignment $\sigma$ is coherent with a set $S$ by $\sigma\dblsubeq S$.
Observe that, by \cref{lemma:newTransversals_complementari}, if $\sigma=\assign{\In,\Ex}$ is coherent with a new transversal $T$ of $\G$ \Wrt $\H$, then the \emph{reversed assignment} $\reverse{\sigma}=\assign{\Ex,\In}$ is coherent with the new transversal $\compl{T}$ of $\H$ \Wrt $\G$.
Intuitively, this means that for an assignment $\sigma=\assign{\In,\Ex}$, set $\In$ is (a subset of) an attempted new transversal of $\G$ \Wrt $\H$, and, symmetrically, set $\Ex$ is (a subset of) an attempted new transversal of $\H$ \Wrt $\G$.

Most algorithms proposed in the literature essentially try different assignments by successively extending in different ways the current assignment.
Each extension performed induces a ``reduced'' instance of \DualProbHyp on which the algorithm is recursively invoked.
Intuitively, the size of the instance decreases for two reasons.
Including vertices in the new attempted transversal of $\G$ increases the number of edges of $\G$ met by the new transversal under construction, and hence there is no need to consider these edges any longer.
Symmetrically, excluding vertices from the new attempted transversal of $\G$ increases the number of edges of $\H$ certainly not contained in the new transversal under construction, and hence, again, there is no need to consider these edges any longer.

Given a hypergraph $\G$ and a set $S$ of vertices, as in~\cite{Elbassioni2008,Boros2009}, we define hypergraphs $\G_S=\tuple{S,\{G\in \G\mid G\subseteq S\}}$, and $\G^S=\tuple{S,\min(\{G\cap S\mid G\in\G\})}$, where $\min(\H)$, for any hypergraph $\H$, denotes the set of inclusion minimal edges of $\H$.
Observe that $\G^S$ is always a simple hypergraph, and that if $\G$ is simple, then so is $\G_S$.
If $\emptyset \in \G$, then $\min(\G) = \{\emptyset\}$.

Let $\I=\tuple{\G,\H}$ be an instance of \DualProbHyp.
While constructing a new transversal of $\G$, when the assignment $\sigma=\assign{\In,\Ex}$ is considered let us denote by $\I_\sigma=\tuple{\G(\sigma),\H(\sigma)}=\tuple{(\G_{V\setminus \In})^{V\setminus (\In\cup\Ex)},(\H_{V\setminus \Ex})^{V\setminus (\In\cup\Ex)}}$ the reduced instance derived from $\I$ and induced by $\sigma$.
Observe that both $\G(\sigma)$ and $\H(\sigma)$ are simple by definition, because they undergo a minimization operation, and that, if $\G$ and $\H$ have the same vertex set, then $\G(\sigma)$ and $\H(\sigma)$ have the same vertex set.
Intuitively, since we are interested in finding new transversals of $\G$ \Wrt $\H$, we can avoid to analyze and further extend an assignment $\sigma=\assign{\In,\Ex}$ for which $\In$ is not an independent set of $\H$ or $\Ex$ is not an independent set of $\G$.
In fact, on the one hand, if $\Ex$ is not an independent set of $\G$, then no set of vertices $T$ coherent with $\sigma$ can be a transversal of $\G$ (because, by $\sigma\dblsubeq T$, $T$ and $\Ex$ are disjoint).
On the other hand, if $\In$ is not an independent set of $\H$, then no set of vertices $T$ coherent with $\sigma$ can be an independent set of $\H$, and hence a new transversal of $\G$ (because, by $\sigma\dblsubeq T$, $\In\subseteq T$).
To this purpose, for an assignment $\sigma=\assign{\In,\Ex}$, if there is an edge $H\in\H$ with $H\subseteq\In$ or an edge $G\in\G$ with $G\subseteq\Ex$, we say that $\In$ and $\Ex$ are \emph{covering}, respectively.
We also say that $\sigma=\assign{\In,\Ex}$ is a \emph{covering} assignment if $\In$ or $\Ex$ are covering.
For future reference, let us highlight the just mentioned property in the following lemma.

\begin{lemma}\label{lemma:assignment_coherent_new_Tr_then_not_covering}
Let $\G$ and $\H$ be two hypergraphs, and let $\sigma=\assign{\In,\Ex}$ be an assignment.
\begin{enumerate}[label=$(\mathit{\alph*})$]
\item If $\Ex$ is covering, then there is no set of vertices coherent with $\sigma$ that is a transversal of $\G$.
\item If $\In$ is covering, then there is no set of vertices coherent with $\sigma$ that is an independent set of $\H$.
\end{enumerate}
Hence, if $\sigma$ is covering, then there is no set of vertices coherent with $\sigma$ that is a new transversal of $\G$ \Wrt $\H$.
\end{lemma}

Decomposing an original instance of \DUAL into multiple sub\nbdash-instances has the advantage of generating smaller sub\nbdash-instances for which it is computationally easier to check their duality.
In fact, many algorithms proposed in the literature decompose the original instance into smaller sub\nbdash-instances for which the duality test is feasible in \PP or even in subclasses of it.
The following property of \DUAL sub\nbdash-instances is of key importance for the correctness of all the approaches tackling \DUAL through decomposition techniques.

\begin{lemma}\label{lemma:dual_iff_all_sub-instances_dual}
Two hypergraphs $\G$ and $\H$ are dual if and only if $\G$ and $\H$ are simple, satisfy the intersection property, and, for all assignments $\sigma$, $\G(\sigma)$ and $\H(\sigma)$ are dual (or, equivalently, there is no new transversal of $\G(\sigma)$ \Wrt $\H(\sigma)$).
\end{lemma}

\Cref{lemma:dual_iff_all_sub-instances_dual} can be easily proven by exploiting \cref{lemma:no_new_transversal_iff_dual} and the equivalence between the hypergraph transversal problem and the duality problem of non-redundant monotone Boolean CNF/DNF formulas, where (partial) assignments here considered correspond to partial Boolean truth assignments.
However, in \cref{sec:app_prove_decomposition}, we provide a detailed proof of the previous lemma without resorting to the equivalence with Boolean formulas and we also analyze other interesting properties of decompositions.

Although checking the duality of sub\nbdash-instances of an initial instance of \DUAL can be computationally easier, it is evident that, in order to find a new transversal of $\G$, naively trying all the possible non\nbdash-covering assignments would require exponential time.
Nevertheless, as already mentioned, there are deterministic algorithms solving \DualProbHyp in quasipolynomial time.
To meet such a time bound, those algorithms do not try all the possible combinations of assignments, but they consider specific assignment extensions.

Common approaches---here referred to as extension\nbdash-types---to extend a currently considered assignment $\sigma=\assign{\In,\Ex}$ to an assignment $\sigma'$ are:
\begin{enumerate}[label=(\roman*)]
  \item\label{ext:inc}
  include a free vertex $v$ into the new attempted transversal of $\G$ \Wrt $\H$, i.e., $\sigma'={}\sigma\addassign\assign{\{v\},\emptyset}$;

  \item\label{ext:inc_crt}
  include a free vertex $v$ as a critical vertex into the new attempted transversal of $\G$ \Wrt $\H$ with an edge $G\in\G(\sigma)$, such that $v\in G$, witnessing $v$'s criticality, i.e., $\sigma'={}\sigma\addassign\assign{\{v\},G\setminus\{v\}}$;

  \item\label{ext:exc}
  exclude a free vertex $v$ from the new attempted transversal of $\G$ \Wrt $\H$ (or, equivalently, include $v$ into the new attempted transversal of $\H$ \Wrt $\G$), i.e., $\sigma'={}\sigma\addassign\assign{\emptyset,\{v\}}$;

  \item\label{ext:exc_crt}
  include a free vertex $v$ as a critical vertex into the new attempted transversal of $\H$ \Wrt $\G$ with an edge $H\in\H(\sigma)$, such that $v\in H$, witnessing $v$'s criticality, i.e., $\sigma'={}\sigma\addassign\assign{H\setminus\{v\},\{v\}}$.
\end{enumerate}

Observe that the extension\nbdash-types listed above always generate consistent assignments, i.e., assignment for which the sets of included and excluded vertices do not overlap.
This is easy to see for extension\nbdash-types~\ref*{ext:inc} and~\ref*{ext:exc}.
For extension\nbdash-type~\ref*{ext:inc_crt}, observe the following.
Let $\sigma=\assign{\In,\Ex}$ be an assignment.
Since $G\in\G(\sigma)$, $G\subseteq V\setminus(\In\cup\Ex)$ by definition of $\G(\sigma)$.
Therefore, $G\cap\In = \emptyset$ and $G\cap\Ex = \emptyset$.
Moreover, the sets of vertices $\{v\}$ and $(G\setminus\{v\})$ clearly constitute a partition of vertices in $G$, thus $\assign{\In\cup\{v\},\Ex\cup(G\setminus\{v\})}$ is a consistent assignment.
Similarly, it can be shown that extesion\nbdash-type~\ref*{ext:exc_crt} generates a consistent assignment.

The size reduction attained in the considered sub\nbdash-instances tightly depends on the assignment extension performed, and in particular on the frequencies with which vertices belong to the edges of $\G(\sigma)$ and $\H(\sigma)$.
We will analyze this in detail below.

\subsection{Assignment trees and the definition of \texorpdfstring{$\Tree(\G,\H)$}{T(G,H)}}\label{sec:assignment_trees}
Most algorithms proposed in the literature adopt their own specific assignment extensions in specific sequences.
The assignments successively considered during the recursive execution of the algorithms can be analyzed through a tree\nbdash-like structure.
Intuitively, each node of the tree can be associated with a tried assignment, and nodes of the tree are connected when their assignments are one the direct extension of the other.
We can call these trees \emph{assignment trees}.

Inspired by the algorithm proposed by \citet{Gaur1999},\footnote{Readers can find in \cref{sec:alg-gaur-det} a \emph{deterministic} algorithm, based on that of \citet{Gaur1999} (see also~\citep{Gaur2004}), deciding hypergraph duality. Note that the original algorithm proposed by \citeauthor{Gaur1999} aims instead at deciding \emph{self}\nbdash-duality of DNF Boolean formulas.
It is from the algorithm reported in the appendix that we have taken ideas to devise our nondeterministic algorithm.
Note that the exposition in \cref{sec:alg-gaur-det} builds up on concepts, definitions, and lemmas discussed in \cref{sec:assignment_trees,sec:logarithmic_refuters}.} we will now describe the construction of a general assignment tree $\Tree(\G,\H)$ that simultaneously represents all possible decompositions of an input \DualProbHyp instance $\tuple{\G,\H}$ according to (assignment extensions directly derived from) extension\nbdash-types~\ref*{ext:inc_crt} and~\ref*{ext:exc}.
This tree is of super\nbdash-polynomial size.
However, it will be shown later that whenever $\G$ and $\H$ are not dual, then there must be in this tree a node at depth $O(\log \InputSize)$ which can be recognized with low computational effort as a witness of the non\nbdash-duality of hypergraphs $\G$ and $\H$.

Intuitively, each node $p$ of the tree $\Tree(\G,\H)$ is associated with an assignment $\sigma_p$.
In particular, the root is labeled with the empty assignment.
Node $p$ of the tree has a child $q$ for each assignment $\sigma_q$ that can be obtained from $\sigma_p$ through an elementary extension of type~\ref*{ext:inc_crt} or~\ref*{ext:exc}.
Edge $(p,q)$ is then labeled by precisely this extension.

For our purposes, drawing upon the algorithm of \citeauthor{Gaur1999}, for each node $p$ of the tree whose assignment is $\sigma_p=\assign{\In_p,\Ex_p}$ we do not (explicitly) consider sub\nbdash-instance $\I_p=\tuple{\G(\sigma_p),\H(\sigma_p)}$.
We refer instead to the following sets:
\begin{itemize}
  \item $\Sep_{\G,\H}(\sigma_p)=\{G\in\G\mid G\cap\In_p=\emptyset\}$, the set of all edges of $\G$ \emph{not met} by (or, equivalently, \emph{separated} from) $\sigma_p$; and
  \item $\Com_{\G,\H}(\sigma_p)=\{H\in\H\mid H\cap\Ex_p=\emptyset\}$, the set of all edges of $\H$ \emph{compatible} with $\sigma_p$.
\end{itemize}
We will often omit the subscript ``$\G,\H$'' of $\Sep_{\G,\H}(\sigma_p)$ and $\Com_{\G,\H}(\sigma_p)$ when hypergraphs $\G$ and $\H$ we are referring to are understood.
The sets $\Sep(\sigma_p)$ and $\Com(\sigma_p)$ are very similar to $\G(\sigma_p)$ and $\H(\sigma_p)$, respectively.
However, roughly speaking, unlike $\G(\sigma_p)$ and $\H(\sigma_p)$, the edges in $\Sep(\sigma_p)$ and $\Com(\sigma_p)$ are neither ``projected'' over the free vertices of $\sigma_p$, nor minimized to obtain simple hypergraphs.
In fact, $\Sep(\sigma_p)$ and $\Com(\sigma_p)$ are (more or less) equivalent to $\G_{V\setminus\In_p}$ and to $\H_{V\setminus\Ex_p}$, respectively (and not to $(\G_{V\setminus\In_p})^{V\setminus(\In_p\cup\Ex_p)}$ and to $(\H_{V\setminus\Ex_p})^{V\setminus(\In_p\cup\Ex_p)}$, respectively).
The difference lies in the fact that $\Sep(\sigma_p)$ and $\Com(\sigma_p)$ are formally defined as sets of sets of vertices, while $\G_{V\setminus\In_p}$ and $\H_{V\setminus\Ex_p}$ are hypergraphs.

More formally, let $\Tree(\G,\H)=\tuple{\NodesTree,A,r,\sigma,\ell}$ be a tree whose nodes $\NodesTree$ are labeled by a function $\sigma$, and whose edges $A$ are labeled by a function $\ell$.
The root $r\in \NodesTree$ of the tree is labeled with the empty assignment $\emptyassign=\assign{\emptyset,\emptyset}$.
Each node $p$ is labeled with the assignment $\sigma_p=\assign{\In_p,\Ex_p}$ (specified below).
The leaves of $\Tree(\G,\H)$ are all nodes $p$ whose assignment $\sigma_p$ is covering or has no free vertex (remember that, by \cref{lemma:assignment_coherent_new_Tr_then_not_covering}, there is no benefit in considering (and further extending) covering assignments).
Each non\nbdash-leaf node $p$ of $\Tree(\G,\H)$ has precisely the following children:

\begin{itemize}
\item Derived from extension\nbdash-type~\ref*{ext:exc}: for each free vertex $v$ of $\sigma_p$, $p$ has a child $q$ such that $\sigma_q=\sigma_p\addassign\assign{\emptyset,\{v\}}$, and the edge connecting $p$ to $q$ is labeled \ExcNode{v}.

\item Derived from extension\nbdash-type~\ref*{ext:inc_crt}: for each $G\in\Sep(\sigma_p)$ and each vertex $v\in G$ that is free in $\sigma_p$, $p$ has a child $q$ such that $\sigma_q=\sigma_p\addassign\assign{\{v\},G\setminus\{v\}}$, and the edge connecting $p$ to $q$ is labeled \IncNodeCrit{v}{G}.
\end{itemize}

Observe that, by this definition of $\Tree(\G,\H)$, the edges leaving a node are all labeled differently, and, moreover, siblings are always differently labeled.
Note, however, that different (non-sibling) nodes may have the same label, and so may edges originating from different nodes.

Also in this case, the extensions considered in tree $\Tree(\G,\H)$ generate consistent assignments.
This is obvious for the extension derived from extension\nbdash-type~\ref*{ext:exc} since the vertex considered in the exclusion is free in the assignment.
For the extension derived from extension\nbdash-type~\ref*{ext:inc_crt} observe the following.
Let $\sigma_p=\assign{\In_p,\Ex_p}$.
The sets of vertices $\{v\}$ and $(G\setminus\{v\})$ constitute a partition of vertices in $G$.
Since $G\in\Sep(\sigma_p)$, $G\cap\In_p = \emptyset$, and hence $(G\setminus\{v\})\cap\In_p = \emptyset$.
It could be the case that $G\cap\Ex_p \neq \emptyset$, however $v\notin\Ex_p$, because $v$ is free.
Therefore, $\assign{\In_p\cup\{v\},\Ex_p\cup(G\setminus\{v\})}$ is a consistent assignment.

To give an example, consider \cref{fig:decomposition_example}.
\begin{figure}[!t]
\centering
\includegraphics[width=0.98\textwidth]{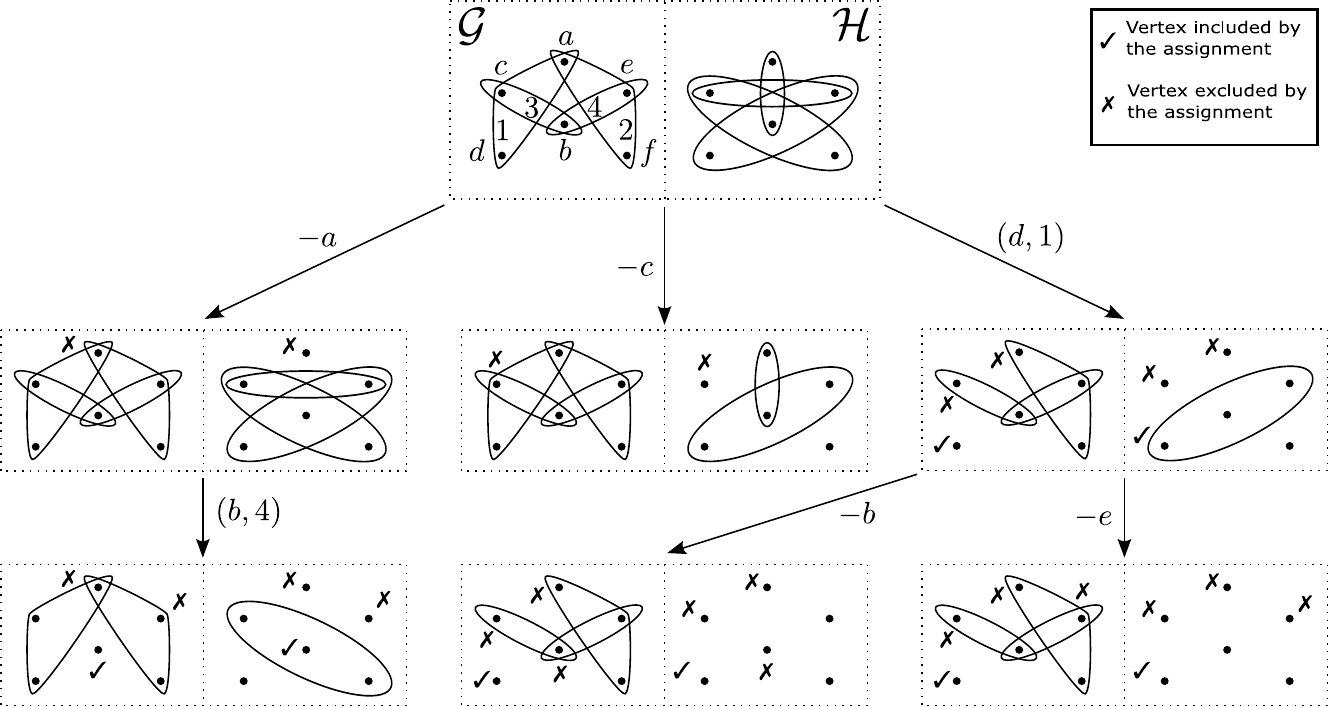}
\caption{A subtree of the decomposition tree $\Tree(\G,\H)$.}\label{fig:decomposition_example}
\end{figure}
Hypergraph vertices are denoted by letters, and hypergraph edges are denoted by numbers.
In the tree illustrated, the root coincides with the pair of hypergraphs of \cref{fig1}, except that the transversal $\{d,b,f\}$ of $\G$ is now missing  in $\H$.
The root is associated with the empty assignment $\emptyassign$, and, correspondingly, the sets depicted with the root node are $\Sep(\emptyassign)$, and $\Com(\emptyassign)$.
Each other node $p$ represents an assignment $\sigma_p$ whose included vertices are indicated by a checkmark (\cmark) and whose excluded vertices by a cross (\xmark).
In addition, each node $p$ shows, on the left\nbdash-hand side, the separated edges of $\G$ and, on the right\nbdash-hand side, the compatible edges of $\H$ in $\sigma_p$, respectively.
The left\nbdash-most edge leaving the root is labeled with \ExcNode{a} which stands for the exclusion of vertex $a$.
This reflects the application of an extension\nbdash-type~\ref*{ext:exc}.
On the other hand, the right\nbdash-most edge leaving the root is labeled with \IncNodeCrit{d}{1} which stands for the inclusion of vertex $d$ as a critical vertex, along with edge $1$ of $\G$ witnessing $d$'s criticality in the attempted new transversal under construction.
This reflects the application of an extension\nbdash-type~\ref*{ext:inc_crt}.
In the given example, not all but only some nodes of the tree are depicted.
Indeed, observe that the bottom right node of the figure is not a leaf, because its assignment is non\nbdash-covering and still contains two free vertices that can be either included (as critical vertices), or excluded.

On the other hand, the bottom central node of the figure is a leaf, because its assignment is covering (in particular, edge $3$ of hypergraph $\G$ is covered by the excluded vertices).

A path $\Pi=(\ell_1,\ell_2,\dots,\ell_k)$ in $\Tree(\G,\H)$ is a sequence of edge labels describing the path from the root to a node following the edges labeled in turn $\ell_1$, $\ell_2$,\dots,$\ell_k$.

For example, in \cref{fig:decomposition_example}, the path $(\IncNodeCrit{d}{1},\ExcNode{e})$ leads to the bottom right node of the figure.

Since the edges leaving a node are assumed to be all differently labelled, a path identifies unequivocally a node in the tree.
Given a path $\Pi$, we denote by $\PathToNode(\Pi)$ the end\nbdash-node of $\Pi$.
Note that the end\nbdash-node of a path is not necessarily a leaf of the decomposition tree.

The next lemma, which shows how to compute the assignment $\sigma_{\PathToNode(\Pi)}$ of node $\PathToNode(\Pi)$, immediately follows from the definition of the concept of path.
For notational convenience we define $\sigma(\Pi)=\sigma_{\PathToNode(\Pi)}$.
\begin{lemma}\label{lemma:assignment_from_path}
\begin{equation}\label{eq:assignment_from_path}
\sigma_{\PathToNode(\Pi)}=\sigma(\Pi)=\assign{\bigcup_{\IncNodeCrit{v}{G}\in\Pi}\{v\},\big(\bigcup_{\ExcNode{v}\in\Pi}\{v\}\big)\cup\big(\bigcup_{\IncNodeCrit{v}{G}\in\Pi}(G\setminus\{v\})\big)}.
\end{equation}
\end{lemma}

Let us now analyze what are the reductions in size achieved when specific extension\nbdash-types are performed.
We denote by
\[
\varepsilon_v^{\Sep(\sigma)}=\frac{|\{G\in\Sep(\sigma)\mid v\in G\}|}{|\{\Sep(\sigma)\}|}
\quad\text{and}\quad
\varepsilon_v^{\Com(\sigma)}=\frac{|\{H\in\Com(\sigma)\mid v\in H\}|}{|\{\Com(\sigma)\}|}\]
the ratio of the edges in $\Sep(\sigma) $ and $\Com(\sigma)$, respectively, containing vertex $v$.

\begin{lemma}\label{lemma:riduzione_taglia_assignment_tree}
Let $\G$ and $\H$ be two hypergraphs satisfying the intersection property, and let $\sigma$ be a non\nbdash-covering assignment.
If $v$ is a free vertex of $\sigma$, then
\begin{align*}
|\Sep(\sigma\addassign\assign{\{v\},\emptyset}) & | = (1-\varepsilon_v^{\Sep(\sigma)})\cdot|\Sep(\sigma)|\text{, and}\\
|\Com(\sigma\addassign\assign{\emptyset,\{v\}}) & | = (1-\varepsilon_v^{\Com(\sigma)})\cdot|\Com(\sigma)|.
\end{align*}
On the other hand, if $G\in\Sep(\sigma)$, $H\in\Com(\sigma)$, and $v\in G$ and $w\in H$ are free vertices of $\sigma$, then
\begin{align*}
|\Com(\sigma\addassign\assign{\{v\},G\setminus\{v\}})| & \leq \varepsilon_v^{\Com(\sigma)}\cdot|\Com(\sigma)|\text{, and}\\
|\Sep(\sigma\addassign\assign{H\setminus\{w\},\{w\}})| & \leq \varepsilon_w^{\Sep(\sigma)}\cdot|\Sep(\sigma)|.
\end{align*}
\end{lemma}
\begin{proof}
Given an assignment $\sigma$, if $v$ is a free vertex in $\sigma$ and $v$ belongs to $\varepsilon_v^{\Sep(\sigma)}\cdot|\Sep(\sigma)|$ many edges of $\Sep(\sigma)$, then, for the assignment $\sigma'=\sigma\addassign\assign{\{v\},\emptyset}$ (extension\nbdash-type~\ref*{ext:inc}), it is easy to see that $|\Sep(\sigma')| = (1-\varepsilon_v^{\Sep(\sigma)})\cdot|\Sep(\sigma)|$.
Similarly, when the assignment $\sigma'=\sigma\addassign\assign{\emptyset,\{v\}}$ is considered (extension\nbdash-type~\ref*{ext:exc}), it is easy to see that $|\Com(\sigma')| = (1-\varepsilon_v^{\Com(\sigma)})\cdot|\Com(\sigma)|$.

Let us now consider extension\nbdash-type~\ref*{ext:inc_crt}, and let $\sigma'=\sigma\addassign\assign{\{v\},G\setminus\{v\}}$.
We will show that $|\Com(\sigma')|\leq\varepsilon_v^{\Com(\sigma)}\cdot|\Com(\sigma)|$.
Let $K_{\{v\}}=\{\wt{H}\in\Com(\sigma)\mid \wt{H}\cap G=\{v\}\}$ be the set of edges in $\Com(\sigma)$ whose intersection with $G$ is exactly vertex $v$, and let $K_{[v]}=\{\wt{H}\in\Com(\sigma)\mid \wt{H}\cap G\ni v\}$ be the set of edges in $\Com(\sigma)$ whose intersection with $G$ contains vertex $v$.
Clearly $K_{\{v\}}\subseteq K_{[v]}$, and hence $|K_{\{v\}}|\leq|K_{[v]}|$.
By definition, $|K_{[v]}|=\varepsilon_v^{\Com(\sigma)}\cdot|\Com(\sigma)|$, therefore $|K_{\{v\}}|\leq \varepsilon_v^{\Com(\sigma)}\cdot|\Com(\sigma)|$.
Since $\G$ and $\H$ satisfy the intersection property, $\Sep(\sigma)$ contains edges of $\G$, and $\Com(\sigma)$ contains edges of $\H$, all edges of $\Com(\sigma)$ intersect $G$.
This, together with the fact that all vertices $G\setminus\{v\}$ are excluded in $\sigma'$, implies that $\Com(\sigma')=K_{\{v\}}$.
Therefore, $|\Com(\sigma')|\leq\varepsilon_v^{\Com(\sigma)}\cdot|\Com(\sigma)|$.
For extension\nbdash-type~\ref*{ext:exc_crt} the proof is similar.
\end{proof}

Observe that in the proof of \cref{lemma:riduzione_taglia_assignment_tree} we do not require $\Sep(\sigma)$ and $\Com(\sigma)$ to be ``simple'', i.e., it is not required that edges belonging to $\Sep(\sigma)$ (resp., $\Com(\sigma)$) are not subsets of other edges in $\Sep(\sigma)$ (resp., $\Com(\sigma)$).
In fact, the lemma is valid regardless of that.
\Cref{lemma:riduzione_taglia_assignment_tree} states a general property about the reduction in size of $\Sep(\sigma)$ and $\Com(\sigma)$ when some assignment extensions are considered.
Indeed, edges $G$ from $\Sep(\sigma)$ do not need to be considered any longer as soon as they contain an included vertex, and this is irrespective of $\Sep(\sigma)$ being actually simple.
A similar discussion extends to edges $H$ of $\Com(\sigma)$ containing excluded vertices.

Before proceeding with our discussion, we recall that
each node $p$ of the tree corresponds to the sub\nbdash-instance $\I_p=\tuple{\G(\sigma_p),\H(\sigma_p)}$ even though we do not explicitly represent it, and we refer instead to $\Sep(\sigma_p)$ and $\Com(\sigma_p)$.
Therefore, $\Tree(\G,\H)$ is indeed a decomposition of the original instance into smaller sub\nbdash-instances.

We claim that, by construction of tree $\Tree(\G,\H)$, if $\G$ and $\H$ are simple hypergraphs satisfying the intersection property, then the pair $\tuple{\G,\H}$ is a ``yes''\nbdash-instance of \DualProbHyp if and only if, for each node $p$ of $\Tree(\G,\H)$, $\I_p$ is a ``yes''\nbdash-instance of \DUAL.
Indeed, if $\G$ and $\H$ are dual, then, by \cref{lemma:dual_iff_all_sub-instances_dual}, there is no assignment $\sigma$ for which $\G(\sigma)$ and $\H(\sigma)$ are not dual, and hence all nodes $p$ of $\Tree(\G,\H)$ are such that $\I_p$ is a ``yes''\nbdash-instance of \DUAL.
On the other hand, if $\G$ and $\H$ are not dual, since we are assuming that they are simple and satisfy the intersection property, by \cref{lemma:no_new_transversal_iff_dual} there is a new transversal $T$ of $\G$ \Wrt $\H$.
The fact that $\G$ and $\H$ are simple implies also that $\G=\G(\emptyassign)=\G(\sigma_r)$ and $\H=\H(\emptyassign)=\H(\sigma_r)$, where $r$ is the root of $\Tree(\G,\H)$.
Therefore, already the root $r$ of $\Tree(\G,\H)$ is such that $\G(\sigma_r)$ and $\H(\sigma_r)$ are not dual, and hence $\I_r$ is a ``no''\nbdash-instance of \DUAL.

The critical point here is that, in order to prove that hypergraphs $\G$ and $\H$ are actually not dual, it would be much better to identify in $\Tree(\G,\H)$ those nodes $p$ for which it is computationally easy to verify that $\G(\sigma_p)$ and $\H(\sigma_p)$ are not dual, and the root of the tree may not always fit the purpose (of an efficient check).
Therefore, to show that $\tuple{\G,\H}$ is a ``no''\nbdash-instance of \DualProbHyp, the ideal solution would be that of finding/guessing a path from the root to a node $p$, where $\I_p$ is easily recognizable as a ``no''\nbdash-instance, e.g., as in the case in which $\sigma_p=\assign{\In_p,\Ex_p}$ is such that the set $\In_p$ or $\Ex_p$ is a new transversal of $\G$ or $\H$, respectively.
The interesting fact here is that, by construction of $\Tree(\G,\H)$, if $\G$ and $\H$ are simple non\nbdash-dual hypergraphs satisfying the intersection property, there is always a node $p$ in $\Tree(\G,\H)$ such that $\sigma_p$ has the required property for an easy check.
Indeed, let us assume that $T$ is a new transversal of $\G$, and consider the assignment $\sigma=\assign{\emptyset,\compl{T}}$.
Since $\sigma$ is an assignment which only excludes vertices, there is a node $p$ in $\Tree(\G,\H)$ such that $\sigma_p=\sigma$ because we can build the assignment $\sigma$ by successively using extensions of type~\ref*{ext:exc}.
It is clear that, in general, such a node may be at depth linear in the tree.
However, we will show in the next section that if $\tuple{\G,\H}$ is actually a ``no''\nbdash-instance of \Dual, then there must also be a node $p$ at depth $O(\log \InputSize)$ such that checking that $\I_p$ is a ``no''\nbdash-instance of \Dual is feasible within complexity class \TC0.

\subsection{Logarithmic refuters in \texorpdfstring{$\Tree(\G,\H)$}{T(G,H)}}\label{sec:logarithmic_refuters}
An assignment $\sigma=\assign{\In,\Ex}$ is said to be a \emph{witness} of the existence of a new transversal of $\G$ \Wrt $\H$ if $\In$ is a new transversal of $\G$ \Wrt $\H$, \emph{or} $\Ex$ is a new transversal of $\H$ \Wrt $\G$ (see \cref{lemma:newTransversals_complementari}).
We remind the reader that, for a set of vertices $T$ to be a new transversal of $\G$ (resp., $\H$) \Wrt $\H$ (resp., $\G$), $T$ has to be a transversal of $\G$ (resp., $\H$) \emph{and} an independent set of $\H$ (resp., $\G$), i.e., $T$ must not be a superset of any edge of $\H$ (resp., $\G$).
Similarly, $\sigma$ is a \emph{double witness} of the existence of a new transversal of $\G$ \Wrt $\H$ if $\In$ is a new transversal of $\G$ \Wrt $\H$, \emph{and} $\Ex$ is a new transversal of $\H$ \Wrt $\G$.
Recall that a new transversal does not need to be minimal.
Double witnesses are easily proven to be characterized as follows.
\begin{lemma}\label{lemma:double_witnessing_condition}
Let $\G$ and $\H$ be two hypergraphs.
An assignment $\sigma$ is a double (non\nbdash-duality) witness if and only if
\begin{equation}\label{eq:double_witnessing_condition}
\Sep(\sigma)=\emptyset \land \Com(\sigma)=\emptyset.
\end{equation}
\end{lemma}

Note here that a node $p$ of $\Tree(\G,\H)$ whose assignment $\sigma_p$ is a witness is not necessarily a leaf of the tree.
For example, in \cref{fig:decomposition_example} the assignment of the bottom right node is a witness, but this node, as already observed, is not a leaf of the full tree.
From now on, we will often refer to properties of an assignment $\sigma_p$ as properties of the node $p$.
For example, we say that a node $p$ of $\Tree(\G,\H)$ is a witness when $\sigma_p$ is actually a witness.

We have already seen that, if $\G$ and $\H$ are simple non\nbdash-dual hypergraphs satisfying the intersection property, then in $\Tree(\G,\H)$ there is always a witness at linear depth.
But this is not enough for our purposes.
Our aim in the rest of this section is to prove a stronger property.
In fact, we will show that there is always at only logarithmic depth a node that is either a witness or can be easily extended to be a witness.

Before proceeding we need the following property.
\begin{lemma}\label{lemma:aggiungere_un_nodo_come_critico}
Let $\G$ and $\H$ be two hypergraphs, and let $\sigma=\assign{\In,\Ex}$ be an assignment coherent with a new \emph{minimal} transversal $T$ of $\G$ \Wrt $\H$, such that $\In\neq T$.
Then, every vertex $v\in(T\setminus \In)$ is free, and for each such vertex, there is an edge $G_v\in\Sep(\sigma)$, with $v\in G_v$, such that $\sigma\addassign\assign{\{v\},G_v\setminus\{v\}}$ is coherent with $T$.
\end{lemma}
\begin{proof}
Let $v$ be any vertex belonging to $T\setminus \In$.
From $\sigma\dblsubeq T$ it follows that $v\notin\Ex$ and hence $v$ is free.
Since $T$ is a minimal transversal of $\G$, by \cref{lemma:transMin-ogniVertexCritico}, $v$ is critical (in $T$).
For this reason, there is an edge $G_v\in\G$ such that $T\cap G_v=\{v\}$.
Now, simply observe that $G_v\in\Sep(\sigma)$ as well (because $T\cap G_v = \{v\}$ and $v\notin\In$), and that $\sigma\addassign\assign{\{v\},G_v\setminus\{v\}}\dblsubeq T$.
\end{proof}

Consider now the following situation.
Let $\G$ and $\H$ be two hypergraphs satisfying the intersection property, for which $T$ is a new \emph{minimal} transversal of $\G$ \Wrt $\H$.
Let us use as a running example again the hypergraphs in \cref{fig:decomposition_example}:
the minimal transversal of $\G$ missing in $\H$ is $T=\{d,b,f\}$.
Assume that we are in the process of building a witness by extending intermediate assignments to new ones that are still coherent with $T$ (essentially we extend assignments with the aim of ``converging'' to $T$, so to have the witness that we are looking for).
From what we have said, intuitively, we can recognize the built assignment $\sigma$ as a witness when sets $\Sep(\sigma)$ and $\Com(\sigma)$ become empty.

In general, if $\sigma = \assign{\In,\Ex}$ is a non\nbdash-witnessing assignment coherent with $T$, we could extend $\sigma$, for example, by excluding a vertex $v$ that is free in $\sigma$ and does not belong to $T$ (i.e., an extension\nbdash-type~\ref*{ext:exc}), or by including, as critical vertex, a vertex $v$ that is free in $\sigma$ and belongs to $T$ (i.e., an extension\nbdash-type~\ref*{ext:inc_crt}).
Note that, by \cref{lemma:aggiungere_un_nodo_come_critico}, if $\sigma=\assign{\In,\Ex}$ is an assignment coherent with a new \emph{minimal} transversal $T$ of $\G$, for \emph{any} vertex $v\in (T\setminus\In)$, it is always possible to include $v$ as critical vertex into $\sigma$ (along with the suitable edge witnessing $v$'s criticality).%
\footnote{To the contrary, if $\sigma$ is coherent with a non\nbdash-minimal new transversal $T$ of $\G$, then not all vertices in $T\setminus\In$ can be included as critical vertices. For more on this, see \cref{lemma:intersezione_minima} in \cref{sec:alg-gaur-det}.}

For the extension\nbdash-type~\ref*{ext:exc}, if $v$ belongs to at least half of the edges in $\Com(\sigma)$, then by excluding $v$ we get rid of at least half of the edges in $\Com(\sigma)$ (see \cref{lemma:riduzione_taglia_assignment_tree}).
On the other hand, for the extension\nbdash-type~\ref*{ext:inc_crt}, if $v$ belongs to less than half of the edges in $\Com(\sigma)$, then by including $v$ as critical vertex with the proper edge witnessing $v$'s criticality we again get rid of at least half of the edges in $\Com(\sigma)$.

In our example, at the root of the tree (i.e., for the empty assignment $\emptyassign$), vertex $c$ is such that $c\notin T$ and $c$ belongs to two out of four edges of $\Com(\emptyassign)$.
So, excluding $c$ is a very good choice because the number of edges from $\Com(\assign{\emptyset,\emptyset})$ to $\Com(\assign{\emptyset,\{c\}})$ would be halved (from four to two).
Symmetrically, vertex $d$ is such that $d\in T$ and $d$ belongs to only one edge of $\Com(\emptyassign)$.
Hence, including $d$ as critical vertex (along with edge $1$ witnessing $d$'s criticality) is again a very good choice because, also in this case, the number of edges from $\Com(\assign{\emptyset,\emptyset})$ to $\Com(\assign{\{d\},\{a,c\}})$ would be (more than) halved (from four to one).

It is evident that halving the size of the set of compatible edges each time we perform an assignment extension would be one of the best ways to asymptotically speed up the construction of the witness sought for.
This is the key intuition to prove that in $\Tree(\G,\H)$ there are ``duality refuters'' at \emph{logarithmic depth}.
The following definitions serve the purpose of formalizing the just illustrated intuition.

Given an assignment $\sigma$, a free vertex $v$ of $\sigma$ is called \emph{frequent} (in $\sigma$) if $v$ belongs to at least half of the edges in $\Com(\sigma)$, otherwise we say that $v$ is \emph{infrequent} (in $\sigma$).
For example, we have already observed that in \cref{fig:decomposition_example} vertex $c$ is frequent at the root and vertex $d$ is infrequent at the root.
Let us denote by $\Freq_{\G,\H}(\sigma)$ and $\Infreq_{\G,\H}(\sigma)$ the frequent and infrequent vertices in $\sigma$, respectively.
Again, we will often omit the subscript ``$\G,\H$'' of $\Freq_{\G,\H}(\sigma)$ and $\Infreq_{\G,\H}(\sigma)$ when the hypergraphs $\G$ and $\H$ we are referring to are understood.

Let $\sigma$ be an assignment coherent with a new \emph{minimal} transversal $T$ of $\G$ such that $\Com(\sigma) \neq \emptyset$.
A free vertex $v$ of $\sigma$ is said to be \emph{appealing to exclude (for $\sigma$) \Wrt $T$} if $v\in\Freq(\sigma)$ and $v\notin T$.
On the other hand, a free vertex $v$ of $\sigma$ is said to be \emph{appealing to include (as a critical vertex) (for $\sigma$) \Wrt $T$} if $v\in\Infreq(\sigma)$ and $v\in T$.

To summarize the first intuition, given a node $p$ coherent with a new minimal transversal $T$, edges leaving $p$ labeled with the exclusion of an appealing vertex to exclude (\Wrt $T$), or labeled with the inclusion as a critical vertex (with a suitable criticality's witness) of an appealing vertex to include (\Wrt $T$), lead to a node $q$ such that $\sigma_q\dblsubeq T$ and $|\Com(\sigma_q)|\leq\frac{1}{2}|\Com(\sigma_p)|$ (see \cref{lemma:riduzione_taglia_assignment_tree}).

The intuition behind the fact that in $\Tree(\G,\H)$ there are, at logarithmic depth, ``duality refuters'' which are moreover \emph{easily recognizable} is the following.
Assume $\G$ and $\H$ satisfy the intersection property and $T$ is a new \emph{minimal} transversal of $\G$ \Wrt $\H$.
As presented above, we are extending assignments via appealing vertices so to quickly ``converge'' to $T$.
If during this process we build a witness, then we are all done, because recognizing a witnessing assignment as such is computationally quite easy.
However, it may well be the case that this process gets stuck if a point is reached where, on the one hand, there are no appealing vertices to extend the current assignment $\sigma$ and, on the other hand, $\sigma$ is not yet a witness.
Remember that $\sigma$ is by construction coherent with $T$.
The key point is the following.
If there are no appealing vertices to extend $\sigma$, then all the free vertices of $\sigma$ that are frequent belong to $T$ and all the free vertices of $\sigma$ that are infrequent do not belong to $T$.
Therefore, we can easily extend $\sigma$ to obtain a witness by simply including all the free vertices of $\sigma$ that are frequent and by excluding all the free vertices of $\sigma$ that are infrequent.
In the next section we show that this final assignment extension based on the frequencies of the vertices can be computed extremely efficiently.

We now formally prove that in $\Tree(\G,\H)$ there are duality refuters at logarithmic depth.
Given an assignment $\sigma=\assign{\In,\Ex}$, $\sigma^+=\assign{\In\cup\Freq(\sigma),\Ex\cup\Infreq(\sigma)}$ is the \emph{augmented assignment of $\sigma$}.

\begin{lemma}\label{lemma:logarithmic_double_witness}
Let $\G$ and $\H$ be two hypergraphs satisfying the intersection property.
Then, there is a new transversal of $\G$ \Wrt $\H$ if and only if there is a node $p$ in $\Tree(\G,\H)$, at depth at most $\lfloor\log |\H|\rfloor+1$, such that ${\sigma_p}^+$ is a double witness.
\end{lemma}
\begin{proof}\hspace{0pt}
$(\Rightarrow)$
Let $T$ be a new \emph{minimal} transversal of $\G$, and let $p$ be a generic node of $\Tree(\G,\H)$ such that $\sigma_p$ is coherent with $T$, $\Com(\sigma_p) \neq \emptyset$, and $\sigma_p$ can be extended via an appealing vertex.
Let $\sigma_q$ be the assignment obtained by extending $\sigma_p$ via the appealing vertex.
Since the structure of $\Tree(\G,\H)$ keeps track of all the possible extensions of types~\ref*{ext:inc_crt} and~\ref*{ext:exc}, there is, among the children of $p$ in $\Tree(\G,\H)$, a node whose assignment is precisely $\sigma_q$.

Observe that the empty assignment $\emptyassign$ associated with the root of $\Tree(\G,\H)$ is obviously coherent with $T$.
Therefore, starting from the root of $\Tree(\G,\H)$, we can build a sequence of nodes that can be successively extended via the inclusion (as a critical vertex) or the exclusion of an appealing vertex \Wrt $T$, until we reach a node $p$ such that $\Com(\sigma_p) = \emptyset$ (and in that moment the frequencies of the vertices become meaningless, because frequencies are evaluated \Wrt to $\Com(\sigma_p)$) or there are no more appealing vertices \Wrt $T$ at all.

Let $s=(p_0,p_1,\dots,p_k)$ be a sequence of maximal length having the just mentioned property.
By definition of $s$, all nodes $p_i$ are such that $\sigma_{p_i}$ is coherent with $T$ (see \cref{lemma:aggiungere_un_nodo_come_critico}).
Since $s$ is of maximal length, $\Com(\sigma_{p_k}) = \emptyset$ or there are no appealing vertices in $\sigma_{p_k}$ \Wrt $T$ to further extend $s$.
Observe that, by the definition of appealing vertex, for all $1\leq i\leq k$, $|\Com(\sigma_{p_i})|\leq\frac{1}{2}|\Com(\sigma_{p_{i-1}})|$ (see \cref{lemma:riduzione_taglia_assignment_tree}).
Therefore, $s$ contains at most $\lfloor\log |\H|\rfloor+2$ nodes, and hence the length of the path from the root to $p_k$ is at most $\lfloor\log |\H|\rfloor+1$.

Let $\sigma_{p_k} = \assign{\In_{p_k},\Ex_{p_k}}$.
There are two cases: either~(1) $\Com(\sigma_{p_k}) = \emptyset$, or~(2) $\Com(\sigma_{p_k}) \neq \emptyset$.

Consider Case~(1).
From $\Com(\sigma_{p_k}) = \emptyset$ it follows that $\Ex_{p_k}$ is a transversal of $\H$.
Since $\sigma_{p_k}$ is coherent with $T$ and $T$ is a transversal of $\G$, $\Ex_{p_k}\cap T = \emptyset$, and hence $\Ex_{p_k}$ is an independent set of $\G$.
Therefore, $\Ex_{p_k}$ is a new transversal of $\H$ \Wrt $\G$.
Because $\Com(\sigma_{p_k}) = \emptyset$, all the free vertices of $\sigma_{p_k}$ are frequent, because each of them belongs to at least half of the edges of $\Com(\sigma_{p_k})$.
By this, ${\sigma_{p_k}}^+ = \assign{\In_{p_k}\cup\Freq(\sigma_{p_k}),\Ex_{p_k}\cup\Infreq(\sigma_{p_k})} = \assign{\In_{p_k}\cup\Freq(\sigma_{p_k}),\Ex_{p_k}} = \assign{\compl{\Ex_{p_k}},\Ex_{p_k}}$, because all the free vertices are frequent.
By \cref{lemma:newTransversals_complementari}, since $\Ex_{p_k}$ is a new transversal of $\H$ \Wrt $\G$, $\compl{\Ex_{p_k}}$ is a new transversal of $\G$ \Wrt $\H$.
Thus, ${\sigma_{p_k}}^+$ is a double witness.

Consider now Case~(2).
Since $p_k$ is the last node of the maximal length sequence $s$ and $\Com(\sigma_{p_k}) \neq \emptyset$, it must be the case that there are no appealing vertices in $\sigma_{p_k}$ \Wrt $T$.
For the following discussion, note that $\Com(\sigma_{p_k}) \neq \emptyset$ implies that the frequencies of the vertices are meaningful.
There are two cases: either~(a) $\Sep(\sigma_{p_k}) \neq \emptyset$, or~(b) $\Sep(\sigma_{p_k}) = \emptyset$.

Consider Case~(a).
Since $\Sep(\sigma_{p_k}) \neq \emptyset$, $\Com(\sigma_{p_k}) \neq \emptyset$, and there are no appealing vertices in $\sigma_{p_k}$ \Wrt $T$, all the frequent vertices of $\sigma_{p_k}$ belong to $T$ and all the infrequent vertices of $\sigma_{p_k}$ are outside $T$.
By this, ${\sigma_{p_k}}^+ = \assign{\In_{p_k}\cup\Freq(\sigma_{p_k}),\Ex_{p_k}\cup\Infreq(\sigma_{p_k})} = \assign{T,\compl{T}}$.
By \cref{lemma:newTransversals_complementari}, since $T$ is a new transversal of $\G$ \Wrt $\H$, $\compl{T}$ is a new transversal of $\H$ \Wrt $\G$.
Thus, ${\sigma_{p_k}}^+$ is a double witness.

Consider now Case~(b).
Since $\sigma_{p_k}$ is coherent with $T$, $T$ is a minimal transversal of $\G$, and $\Sep(\sigma_{p_k}) = \emptyset$, $\In_{p_k} = T$.
Because $\In_{p_k} = T$, any free vertex $v$ of $\sigma_{p_k}$ is such that $v\in (\compl{T}\setminus\Ex_{p_k})$ and, since there are no appealing vertices in $\sigma_{p_k}$ \Wrt $T$, $v$ is infrequent (for otherwise $v$ would have been an appealing vertex to exclude).
By this, ${\sigma_{p_k}}^+ = \assign{\In_{p_k}\cup\Freq(\sigma_{p_k}),\Ex_{p_k}\cup\Infreq(\sigma_{p_k})} = \assign{\In_{p_k},\Ex_{p_k}\cup\Infreq(\sigma_{p_k})} = \assign{\In_{p_k},\compl{\In_{p_k}}}$, because all the free vertices are infrequent.
By \cref{lemma:newTransversals_complementari}, since $\In_{p_k} = T$ is a new transversal of $\G$ \Wrt $\H$, $\compl{\In_{p_k}}$ is a new transversal of $\H$ \Wrt $\G$.
Thus, ${\sigma_{p_k}}^+$ is a double witness.

$(\Leftarrow)$
Clearly, if there is, within the required depth, a node $p$ such that ${\sigma_p}^+$ is a double witness, then there is a new transversal of $\G$ \Wrt $\H$.
\end{proof}

\section{New upper bounds for the \texorpdfstring{\DualProbHyp}{DUAL} problem}\label{sec:new_upper_bound}

\subsection{A new nondeterministic algorithm for \texorpdfstring{\DualProbHyp}{DUAL}}\label{sec:non-deterministic_algorithm}
We present in this section our new nondeterministic algorithm \textsc{\NDAlg} for \NonDualProbHyp and prove its correctness.
To prove the correctness of the algorithm we will exploit the property of $\Tree(\G,\H)$ of having easily recognizable duality refuters at logarithmic depth (\cref{lemma:logarithmic_double_witness}).

To disprove that two hypergraphs $\G$ and $\H$ are dual, we know that it is sufficient to show that at least one of them is not simple, or that the intersection property does not hold between them, or that there is a new transversal of $\G$ \Wrt $\H$ (see \cref{lemma:no_new_transversal_iff_dual}).
Ruled out the first two conditions, intuitively, our nondeterministic algorithm, to compute such a new transversal, guesses in the tree $\Tree(\G,\H)$ a path of logarithmic length leading to a node $p$ such that ${\sigma_p}^+$ is a double witness (see \cref{lemma:logarithmic_double_witness}).
More precisely, \textsc{\NDAlg} nondeterministically generates a \emph{set} $\Sigma$ of logarithmic\nbdash-many labels, which is then checked to verify whether it is possible to derive from it a new transversal of $\G$ \Wrt $\H$.

To be more formal, for a hypergraph $\G$, let $\LabelsOfHyp_\G$ denote the set of all possible labels that can potentially be generated from edges of $\G$ and vertices, i.e., $\LabelsOfHyp_\G = \{\ExcNode{v}\mid v\in V\}\cup\{\IncNodeCrit{v}{G}\mid G\in\G \land v\in G\}$.
A set of labels $\Sigma$ of the tree $\Tree(\G,\H)$ is any subset of $\LabelsOfHyp_\G$.
Note the difference between a path of $\Tree(\G,\H)$ and a set of labels of $\Tree(\G,\H)$.
The former is an ordered sequence of labels coherent with the structure of $\Tree(\G,\H)$, while the latter is just a set.
Given the above notation, we define the set
\begin{equation*}
\SetOfLogSeqs(\G,\H)=\{\Sigma\mid \Sigma \subseteq \LabelsOfHyp_\G \land 0\leq |\Sigma|\leq\lfloor\log |\H|\rfloor+1\}.
\end{equation*}
Given a set $\Sigma\in\SetOfLogSeqs(\G,\H)$, the following expressions
\begin{equation}\label{eq:assignment_from_sequence}
\begin{aligned}
\In(\Sigma) & =\bigcup_{\IncNodeCrit{v}{G}\in\Sigma}\{v\}\\
\Ex(\Sigma) & =\big(\bigcup_{\ExcNode{v}\in\Sigma}\{v\}\big)\cup\big(\bigcup_{\IncNodeCrit{v}{G}\in\Sigma}(G\setminus\{v\})\big)
\end{aligned}
\end{equation}
indicate the sets of the included and excluded vertices in $\Sigma$, respectively.
These two expressions are similar to the formulas to compute an assignment given a (valid) path in $\Tree(\G,\H)$ (Formula~\eqref{eq:assignment_from_path} of \cref{lemma:assignment_from_path}).
Since a set $\Sigma$ is merely a set of labels, it may happen that $\In(\Sigma)\cap\Ex(\Sigma)\neq\emptyset$.
When this is \emph{not} the case we say that $\Sigma$ is a \emph{consistent} set of labels.

Given a set of labels $\Sigma$, we define $\sigma(\Sigma)$ as the pair $\assign{\In(\Sigma),\Ex(\Sigma)}$.
If $\Sigma$ is consistent, then $\sigma(\Sigma)=\assign{\In(\Sigma),\Ex(\sigma)}$ is a (consistent) assignment as well.
By a slight abuse of terminology and notation, given a set of labels $\Sigma$, regardless of whether $\Sigma$ is actually consistent or not, we extend, in the natural way, the definitions given for (consistent) assignments to the pair $\sigma(\Sigma)=\assign{\In(\Sigma),\Ex(\Sigma)}$.

The following property is fundamental for the correctness of algorithm \textsc{\NDAlg}.
In particular, it highlights that the order in which the labels of a path $\Pi$ appear is not actually relevant to recognize that ${\sigma(\Pi)}^+$ is a witness.

\begin{lemma}\label{lemma:ipergrafi_non-dual_iff_witness_nel_set}
Let $\G$ and $\H$ be two hypergraphs satisfying the intersection property.
Then, there is a new transversal of $\G$ \Wrt $\H$ if and only if there is a consistent set $\Sigma\in\SetOfLogSeqs(\G,\H)$ such that $\sigma(\Sigma)^+$ is a double non\nbdash-duality witness.
\end{lemma}
\begin{proof}
$(\Rightarrow)$
By \cref{lemma:logarithmic_double_witness}, because $\G$ and $\H$ satisfy the intersection property, if there is a new transversal of $\G$ \Wrt $\H$, then, in $\Tree(\G,\H)$, there is a path $\Pi$ of length at most $\lfloor\log |\H|\rfloor+1$ such that ${\sigma(\Pi)}^+$ is a double non\nbdash-duality witness.
Let $\Sigma^\Pi$ be the set of labels containing exactly the labels of $\Pi$ ($\Sigma^\Pi$, unlike $\Pi$, is a set without any order over its elements).
Because the length of $\Pi$ is at most $\lfloor\log |\H|\rfloor+1$, $\Sigma^\Pi$ actually belongs to $\SetOfLogSeqs(\G,\H)$.
By comparing Formula~\eqref{eq:assignment_from_path} of \cref{lemma:assignment_from_path} and Formula~\eqref{eq:assignment_from_sequence} of this section, it is clear that the included and excluded vertices of $\sigma(\Pi)$ do not depend on the actual order of the labels in $\Pi$.
Hence, $\sigma(\Sigma^\Pi) = \sigma(\Pi)$ and also ${\sigma(\Sigma^\Pi)}^+ = {\sigma(\Pi)}^+$.
Given that $\Pi$ is a path in $\Tree(\G,\H)$, $\sigma(\Pi)$ is a (consistent) assignment by definition, and hence, from $\sigma(\Sigma^\Pi) = \sigma(\Pi)$ follows that $\Sigma$ is a consistent set.
To conclude, since $\sigma(\Pi)^+$ is a double witness and ${\sigma(\Sigma^\Pi)}^+ = {\sigma(\Pi)}^+$, ${\sigma(\Sigma^\Pi)}^+$ is a double witness as well.

$(\Leftarrow)$
Clearly, if there is a consistent set of labels $\Sigma\in\SetOfLogSeqs(\G,\H)$ such that $\sigma(\Sigma)^+$ is a double witness, then there is a new transversal of $\G$ \Wrt $\H$.
\end{proof}

We now present our nondeterministic algorithm.
The pseudo\nbdash-code of algorithm \textsc{\NDAlg} is listed as \cref{alg:non-det_algorithm}; ``\textbf{return accept}'' and ``\textbf{return reject}'' are two commands causing a transition to a final accepting state and to a final rejecting state, respectively, of a nondeterministic machine.

\begin{algorithm}[!ht]
\caption{A nondeterministic algorithm for \coDUAL.}\label{alg:non-det_algorithm}
\begin{algorithmic}[1]
\Procedure{\NDAlg}{$\G,\H$}
\State $\Sigma\gets{}$\textbf{guess}(a set of labels from $\SetOfLogSeqs(\G,\H)$);\label{line:nd-alg_guess}
\If{$\lnot$\Call{\CheckSimpleIP}{$\G,\H$}}\label{line:nd-alg_checkSimpleIP}
    \textbf{return accept};
\EndIf%
\If{$\lnot$\Call{\CheckConsistencySet}{$\G,\Sigma$}}\label{line:nd-alg_checkConsistencyGuess}
    \textbf{return reject};
\EndIf%
\If{\Call{\CheckDoubleWitnessAug}{$\G,\H,\Sigma$}}\label{line:nd-alg_checkDoubleWitnessAug}
    \textbf{return accept};
\EndIf%
\State\textbf{return reject};\label{line:nd-alg_final-rejection}
\EndProcedure
\end{algorithmic}
\end{algorithm}

The three checking\nbdash-procedures used in the algorithm implement the four deterministic tests needed after the guess has been carried out.
The aims of the subprocedures are the following.
\begin{description}[font=\normalfont\scshape,nosep]
\item[\CheckSimpleIP] checks whether the two hypergraphs $\G$ and $\H$ are simple and satisfy the intersection property.\footnote{This condition does not depend on the guessed set, and could thus be checked at the beginning of the algorithm, before the guess is made. However, for uniformity, and to adhere to a strict guess-and-check paradigm we check it after the guess.}

\item[\CheckConsistencySet] checks whether the guessed set of labels $\Sigma$ is actually consistent.

\item[\CheckDoubleWitnessAug] checks whether ${\sigma(\Sigma)}^+$ is a double witness.
In order to perform this check, Condition~\eqref{eq:double_witnessing_condition} of \cref{lemma:double_witnessing_condition} is evaluated on ${\sigma(\Sigma)}^+$.
\end{description}

The following property shows that algorithm \textsc{\NDAlg} is correct.
We remind the reader that a nondeterministic algorithm $A$ is correct for a decision problem $P$ if $A$ admits a computation branch terminating in an accepting state if and only if $A$ is executed on a `yes'\nbdash-instance of $P$.
\begin{theorem}
Let $\G$ and $\H$ be two hypergraphs.
Then, there is a computation branch of \textsc{\NDAlg}$(\G,\H)$ halting in an accepting state if and only if $\G$ and $\H$ are not dual.
\end{theorem}
\begin{proof}
$(\Rightarrow)$
Assume that there is a computation branch of \textsc{\NDAlg}$(\G,\H)$ halting in an accepting state.
A transition to an accepting state may happen only at line~\ref*{line:nd-alg_checkSimpleIP} or~\ref*{line:nd-alg_checkDoubleWitnessAug}.
If such a transition happens at line~\ref*{line:nd-alg_checkSimpleIP}, then $\G$ or $\H$ is not simple, or $\G$ and $\H$ do not satisfy the intersection property.
Hence, by \cref{lemma:no_new_transversal_iff_dual}, $\G$ and $\H$ are not dual.

If a transition to the accepting state occurs at line~\ref*{line:nd-alg_checkDoubleWitnessAug}, it means that the execution flow of the algorithm has reached that point, implying that the test performed at line~\ref*{line:nd-alg_checkConsistencyGuess} passed.
Therefore, the guessed set of labels $\Sigma$ is consistent.
The fact that the check at line~\ref*{line:nd-alg_checkDoubleWitnessAug} is successful implies that $\sigma(\Sigma)^+$ is a double witness.
Thus, by \cref{lemma:ipergrafi_non-dual_iff_witness_nel_set}, there is a new transversal of $\G$ \Wrt $\H$, and hence, by \cref{lemma:no_new_transversal_iff_dual}, $\G$ and $\H$ are not dual.

$(\Leftarrow)$
Let us now assume that $\G$ and $\H$ are not dual.
By \cref{lemma:no_new_transversal_iff_dual}, $\G$ or $\H$ is not simple, or they do not satisfy the intersection property, or there is a new transversal of $\G$ \Wrt $\H$.

If $\G$ or $\H$ is not simple, this condition is recognized by the algorithm at line~\ref*{line:nd-alg_checkSimpleIP} and the algorithm correctly moves to an accepting state.
The same happens in case $\G$ and $\H$ do not satisfy the intersection property.

Consider now the case in which $\G$ and $\H$ are simple and satisfy the intersection property.
Since $\G$ and $\H$ are non\nbdash-dual, by \cref{lemma:no_new_transversal_iff_dual}, there is a new transversal of $\G$ \Wrt $\H$.

Because $\G$ and $\H$ satisfy the intersection property, by \cref{lemma:ipergrafi_non-dual_iff_witness_nel_set}, among the sets $\Sigma$ guessed by the algorithm at line~\ref*{line:nd-alg_guess} there must be a consistent one such that $\sigma(\Sigma)^+$ is a double witness.
This is recognized by the algorithm at line~\ref*{line:nd-alg_checkDoubleWitnessAug} and the algorithm correctly moves to an accepting state.
\end{proof}

Note that the approach of extending assignments used in our paper, while partly inspired by \citeauthor{Gaur1999}'s ideas, is fundamentally different from the method used in \citeauthor{Gaur1999}'s deterministic algorithm~\citep{Gaur1999,Gaur2004}.
In particular, \citeauthor{Gaur1999}'s algorithm may extend the set $\In$ of included vertices of an intermediate assignment $\sigma=\assign{\In,\Ex}$ in a single step by several vertices and not by just one.
In our approach this is only possible for end\nbdash-nodes of the path.
Moreover, algorithm \textsc{\NDAlg} could identify a witness by guessing a path of logarithmic length that is not a legal path according to \citeauthor{Gaur1999} because the single assignment\nbdash-extensions are not chosen according to frequency counts.
In fact, unlike \citeauthor{Gaur1999}'s algorithm, \textsc{\NDAlg} performs frequency counts only at the terminal nodes of a path.

\subsection{Logical analysis of the \texorpdfstring{{\sc \NDAlg}}{\NDAlg} algorithm}\label{sec:logic}
We will show that the deterministic tests performed by algorithm \textsc{\NDAlg} require (quite) low computational effort, and in particular they can be carried out within complexity class \TC0.
This will allow us to prove that $\NonDualProbHyp\in\GC{$\log^2 \InputSize$}{\TC0}$.
We begin by expressing the deterministic tests performed by \textsc{\NDAlg} in \FOC which is first order logic augmented with counting quantifiers ``$\exists! n$'' having the following semantics.
\FOC is a two\nbdash-sorted logic, i.e., a logic having two domain sets~\cite{Immerman1999}:
a numerical domain set containing objects used to interpret only numerical values, and another domain set containing all other objects.
Consider the formula $\Phi(n,x)=(\exists!n\ x)(\phi(x))$, where variable $x$ ranges over the non\nbdash-numerical domain objects and is bound by the counting quantifier $\exists!n$, and where variable $n$ ranges over the numerical domain objects and is left free by the counting quantifier.
$\Phi(n,x)$ is valid in all the interpretations in which $n$ is substituted by the exact number of non\nbdash-numerical domain values $a$ for which $\phi(a)$ evaluates to \valtrue.\footnote{For more on this, the reader is referred to any standard textbook on the topic. See, e.g., \citep{Ebbinghaus1994,Immerman1999,Ebbinghaus1999,Libkin2004}.}
Note that first order logic augmented with the majority quantifiers (\FOM) is known to be equivalent to \FOC~\cite{Immerman1999,Barrington1990,Vollmer1999}.
The model checking problem for both logics is complete for class \TC0~\cite{Vollmer1999,Immerman1999,Barrington1990}.

With a pair of hypergraphs $\tuple{\G,\H}$ we associate a relational structure $\A_{\tuple{\G,\H}}$.
Essentially, we represent hypergraphs through their incidence graphs.
In particular, the universe $A_{\tuple{\G,\H}}$ of $\A_{\tuple{\G,\H}}$ consists of an object for each vertex of $V$, an object for each hyperedge of the two hypergraphs, and two more objects, $o_\G$ and $o_\H$, for the two hypergraphs, i.e., $A_{\tuple{\G,\H}}=\{o_v\mid v\in V\}\cup\{o_G\mid G\in\G\}\cup\{o_H\mid H\in\H\}\cup\{o_\G,o_\H\}$.

The relations of $\A_{\tuple{\G,\H}}$ are as follows:
$\Vertex(x)$ is a unary relation indicating that object $x$ is a vertex;
$\Hyper(x)$ is a unary relation indicating that object $x$ is a hypergraph;
$\Edge(x,y)$ is a binary relation indicating that object $x$ is an edge of the hypergraph identified by object $y$;
and $\Incidence(x,y)$ is the binary incidence relation indicating that object $x$ is a vertex belonging to the edge identified by object $y$.

We also need to represent through relations the guessed set $\Sigma$.
Remember that in $\Sigma$ there are elements (which are labels) of two types: \ExcNode{v} where $v$ is a vertex, and \IncNodeCrit{v}{G} where $v$ is a vertex and $G$ is an edge of $\G$.\footnote{Note that an edge $G$ in a label of a path or a set is given by its identifier and not by the explicit list of its vertices.}
We assume a unary relation $\Sone$ storing those tuples $\tuple{v}$ where $v$ is a vertex such that $\ExcNode{v}\in\Sigma$, and we assume a binary relation $\Stwo$ containing those tuples $\tuple{v,G}$ where $v$ is a vertex and $G$ is an edge such that $\IncNodeCrit{v}{G}\in\Sigma$.

Remember that, by \cref{lemma:ipergrafi_non-dual_iff_witness_nel_set}, it is sufficient to guess a set of labels, and it is not required to guess a path.
This means that the exact order of the labels is not relevant, and hence the above relational representation of a guessed set is totally sufficient.

We use the following ``macros'' in our first order formulas:
\begin{align*}
v\in V & \equiv \Vertex(v)\displaybreak[0]\\
g\in\G & \equiv \Hyper(o_\G)\land\Edge(g,o_\G)\displaybreak[0]\\
h\in\H & \equiv \Hyper(o_\H)\land\Edge(h,o_\H)\displaybreak[0]\\
v\in g & \equiv \Incidence(v,g)
\end{align*}

We are now ready to prove some intermediate results.
\begin{lemma}\label{lemma:complexity_check_simple}
Let $\G$ be a hypergraph.
Deciding whether $\G$ is simple is expressible in \FO.
\end{lemma}
\begin{proof} We know that a hypergraph $\G$ is simple if and only if, for all pairs of distinct edges $G,H\in\G$, $G\not\subseteq H$.
Hence, the formula checking whether a hypergraph is simple is:
\begin{equation*}
\Simple(x)\equiv \Hyper(x)\land (\forall g,h)((g\in x\land h\in x\land g\neq h)\rightarrow (\exists v)(v\in V\land v\in g\land \lnot(v\in h))).
\end{equation*}
\end{proof}

\begin{lemma}\label{lemma:complexity_check_hitting_property}
Let $\G$ and $\H$ be two hypergraphs.
Deciding whether $\G$ and $\H$ satisfy the intersection property is expressible in \FO.
\end{lemma}
\begin{proof}
Two hypergraphs $\G$ and $\H$ satisfy the intersection property if and only if, for every pair of edges $G\in\G$ and $H\in\H$, $G\cap H\neq\emptyset$.
Hence, the formula encoding this test is:
\begin{equation*}
\IntersectionProperty\equiv (\forall g,h)((g\in\G \land h\in\H)\rightarrow (\exists v)(v\in V\land v\in g \land v\in h)).
\end{equation*}
\end{proof}

We say that the guess is congruent (which is different from being consistent) if, for every guessed tuple $\tuple{x}\in\Sone$, object $x$ is actually a vertex, and, for every tuple $\tuple{x,y}\in\Stwo$, object $y$ is actually an edge belonging to $\G$ containing the vertex identified by object $x$.

\begin{lemma}\label{lemma:complexity_check_consistency_guess}
Let $\G$ and $\H$ be two hypergraphs, and let $\Sigma$ be a (guessed) set of labels of $\Tree(\G,\H)$.
Deciding the congruency and the consistency of $\Sigma$ is expressible in \FO.
\end{lemma}
\begin{proof}
The congruency of the guessed set can be checked through:
\begin{equation*}
\CongruentGuess\equiv (\forall v)(\Sone(v)\rightarrow v\in V)\land(\forall w,g)(\Stwo(w,g)\rightarrow w\in V\land g\in\G\land w\in g).
\end{equation*}
The consistency check is expressed as the conjunction of two formulas.
The first verifies whether there is an \emph{in}consistency on a specific vertex, and the second checks the overall consistency.
\begin{align*}
\NotConsistency(w) & \equiv w\in V\land (\exists g)(\Stwo(w,g)\land(\Sone(w)\lor(\exists v,h)(\Stwo(v,h)\land v\neq w\land w\in h)))\\
\ConsistentGuess & \equiv (\forall v)(v\in V\rightarrow\lnot\NotConsistency(v)).
\end{align*}
\end{proof}

To conclude our complexity analysis of the deterministic tests performed by \textsc{\NDAlg}, let us formulate in \FOC the property that ${\sigma(\Sigma)}^+$ is a double non\nbdash-duality witness.

\begin{lemma}\label{lemma:complexity_check_guess_augmented_witness}
Let $\G$ and $\H$ be two hypergraphs, and let $\Sigma$ be a (guessed) set of labels of $\Tree(\G,\H)$.
Deciding whether ${\sigma(\Sigma)}^+$ is a double non\nbdash-duality witness is expressible in \FOC.
\end{lemma}
\begin{proof}
Let $\sigma(\Sigma)=\assign{\In(\Sigma),\Ex(\Sigma)}$ be the pair associated with $\Sigma$.
Essentially we need to prove that is possible to express in \FOC Condition~\eqref{eq:double_witnessing_condition} of \cref{lemma:double_witnessing_condition} on $\sigma(\Sigma)^+$.
Remember that $\sigma(\Sigma)$ is not explicitly represented, but it can be evaluated from $\Sigma$ through Formula~\eqref{eq:assignment_from_sequence} of \cref{sec:non-deterministic_algorithm}.

Let us define the following two formulas serving the purpose of evaluating whether a vertex belongs to $\In(\Sigma)$ or $\Ex(\Sigma)$, respectively.
\begin{align*}
\IGuess(v) & \equiv v\in V\land (\exists g)(\Stwo(v,g))\displaybreak[0]\\
\EGuess(v) & \equiv v\in V\land (\Sone(v)\lor (\exists w,g)(\Stwo(w,g)\land w\neq v\land v\in g)).
\end{align*}

We now exhibit the formulas to verify whether a given free vertex is frequent in $\sigma(\Sigma)$.
These formulas are the only ones in which we actually use counting quantifiers.
In the following formulas we will use predicate $\Plus(x,y,z)$, which holds \valtrue whenever $x+y=z$, and predicate $\Succ(x,y)$, which holds \valtrue whenever $x$ and $y$ are two domain values such that $y$ is the immediate successor of $x$ in the domain ordering.
Remember, indeed, that relational structures are assumed to have totally ordered domains (and predicate ``$\Prec$'' allow us to test the ordering), and to have a predicate $\Bit(i,j)$ that holds \valtrue whenever the $j^\text{th}$ bit of the binary representation of number $i$ is $1$. These assumptions allow to express in first\nbdash-order logic predicates $\Plus(x,y,z)$ and $\Succ(x,y)$ (see Section~1.2 of~\cite{Immerman1999}).

Since we need to evaluate whether a vertex $v$ is frequent in $\sigma(\Sigma)$, we have to check whether $v$ belongs to at least $\lceil|\Com(\sigma(\Sigma))|/2\rceil$ edges of $\Com(\sigma(\Sigma))$.
So, we exhibit a formula $\Half(x,y)$ which holds \valtrue whenever $y=\lceil x/2\rceil$:
\begin{equation*}
\Half(x,y)\equiv \Plus(y,y,x)\lor(\exists z)(\Plus(y,y,z)\land\Succ(x,z)).
\end{equation*}

The following formulas evaluate whether an edge belongs to $\Com(\sigma(\Sigma))$, the number of edges in $\Com(\sigma(\Sigma))$, the number of edges in $\Com(\sigma(\Sigma))$ containing a given vertex $v$, and whether $v$ is frequent in $\sigma(\Sigma)$, respectively (remember that being either frequent or infrequent is a property of free vertices).
\begin{gather*}
\begin{aligned}
\ComHAfterGuess(h) & \equiv h\in\H\land (\forall v)((v\in V\land v\in h)\rightarrow\lnot \EGuess(v))\\
\CountCompEdges(n) & \equiv (\exists!n\ h)(h\in\H\land \ComHAfterGuess(h))\\
\CountCompEdgesIncl(v,n) & \equiv v\in V\land (\exists!n\ h)(h\in\H\land \ComHAfterGuess(h)\land v\in h)
\end{aligned}\\
\begin{split}
\ALHCompEdges(v) & \equiv{} v\in V\land \lnot\IGuess(v)\land\lnot\EGuess(v)\land (\exists n,m,o)\\
& (\CountCompEdges(n)\land \CountCompEdgesIncl(v,m)\land \Half(n,o) \land (o = m\lor o \Prec m)).
\end{split}
\end{gather*}
After having defined a formula to evaluate whether a vertex is frequent in $\sigma(\Sigma)$, we show the formulas computing the included and excluded vertices of the augmented pair $\sigma(\Sigma)^+$, respectively:
\begin{align*}
\IModified(v) & \equiv v\in V\land (\IGuess(v)\lor \ALHCompEdges(v))\displaybreak[0]\\
\EModified(v) & \equiv v\in V\land (\EGuess(v)\lor \lnot\ALHCompEdges(v)).
\end{align*}

Now we show the formulas encoding the evaluation of Condition~\eqref{eq:double_witnessing_condition} of \cref{lemma:double_witnessing_condition} on $\sigma(\Sigma)^+$.
The formulas evaluating whether an edge belongs to $\Sep(\sigma(\Sigma)^+)$ and to $\Com(\sigma(\Sigma)^+)$ are, respectively:
\begin{align*}
\SepGModified(g) & \equiv g\in\G\land (\forall v)((v\in V\land v\in g)\rightarrow \lnot\IModified(v))\displaybreak[0]\\
\ComHModified(h) & \equiv h\in\H\land (\forall v)((v\in V\land v\in h)\rightarrow \lnot\EModified(v)).
\end{align*}

Finally, the formula verifying that $\sigma(\Sigma)^+$ meets Condition~\eqref{eq:double_witnessing_condition} of \cref{lemma:double_witnessing_condition} is as follows.
\begin{equation*}
\CheckGuessAugDoubleWitness \equiv{} (\forall g)(g\in\G\rightarrow \lnot \SepGModified(g)) \land (\forall h)(h\in\H\rightarrow \lnot \ComHModified(h).
\end{equation*}
\end{proof}

\subsection{Putting it all together}\label{sec:complexity_results}
We are now ready to prove our main results.
\begin{theorem}\label{theo:complexity_non-dual}
Let $\G$ and $\H$ be two hypergraphs.
Then, deciding whether $\G$ and $\H$ are not dual is feasible in \GC{$\log^2 \InputSize$}{\TC0}.
\end{theorem}
\begin{proof}
\Cref{lemma:complexity_check_simple,lemma:complexity_check_hitting_property,lemma:complexity_check_consistency_guess,lemma:complexity_check_guess_augmented_witness} show that the deterministic checks performed by algorithm \textsc{\NDAlg} are expressible in \FOC.
Hence, these tests are feasible in logtime\nbdash-uniform \TC0~\cite{Immerman1999,Vollmer1999,Barrington1990}.

Moreover, by analyzing algorithm \textsc{\NDAlg}, clearly only $O(\log^2 \InputSize)$ nondeterministic bits are sufficient to be guessed to properly identify a (double) non\nbdash-duality witness.
Indeed, let us assume that $\G$ and $\H$ are not dual.
If $\G$ or $\H$ is not simple, or $\G$ and $\H$ do not satisfy the intersection property, then the guessed set $\Sigma$ is completely ignored, because $\G$ and $\H$ are directly recognized to be non\nbdash-dual, and hence is totally irrelevant what the guessed bits are.
On the other hand, if $\G$ and $\H$ are simple and satisfy the intersection property, since they are not dual, by \cref{lemma:no_new_transversal_iff_dual}, there is a new transversal of $\G$ \Wrt $\H$.
Therefore, by \cref{lemma:ipergrafi_non-dual_iff_witness_nel_set}, there is in $\SetOfLogSeqs(\G,\H)$ a set $\Sigma$ with $O(\log |\H|)$ elements such that $\sigma(\Sigma)^+$ is a double witness.
Remember that our definition of the size of the input of \DUAL is $\InputSize=\|\G\|+\|\H\|$, hence the number of elements of $\Sigma$ is also $O(\log \InputSize)$.
Since, by definition, $O(\log \InputSize)$ bits are sufficient to represent any vertex or edge ID of the input hypergraphs, each label of $\Sigma$ can be represented with only $O(\log \InputSize)$ bits.
By this, the whole set $\Sigma$ can be correctly represented and stored in the (set) variable $\Sigma$ with $O(\log^2 \InputSize)$ bits.

Therefore, \NonDualProbHyp belongs to \GC{$\log^2 \InputSize$}{\TC0}.
\end{proof}

From the previous theorem the following corollaries follows immediately, the first of which proves that the conjecture stated by \citet{gott13} actually holds.
\begin{corollary}\label{corol:DUAL_in_GC-log2-LOGSPACE}
Let $\G$ and $\H$ be two hypergraphs.
Deciding whether $\G$ and $\H$ are not dual is feasible in \GC{$\log^2 \InputSize$}{\LogSpace}.
\end{corollary}
\begin{proof}
Follows from \cref{theo:complexity_non-dual} and the inclusion $\TC0\subseteq\LogSpace$.
\end{proof}

\begin{corollary}[\citep{gott13}]\label{corol:DUAL_in_DSPACE_log2_n}
Let $\G$ and $\H$ be two hypergraphs.
Deciding whether $\G$ and $\H$ are dual is feasible in $\DSpace[\log^2 \InputSize]$.
\end{corollary}
\begin{proof}
$\coDUAL\in\DSpace[\log^2 \InputSize]$ follows from \cref{corol:DUAL_in_GC-log2-LOGSPACE} and the inclusion $\GC{$\log^2 \InputSize$}{\LogSpace}\subseteq\DSpace[\log^2 \InputSize]$.
Since $\DSpace[\log^2 \InputSize]$ is closed under complement, $\DUAL\in\DSpace[\log^2 \InputSize]$.
\end{proof}

\subsection{Computing a new transversal}\label{sec:complexity_computing}
In this section, we will show that computing a (not necessarily minimal) new transversal of $\G$ \Wrt $\H$ is feasible in space $O(\log^2 \InputSize)$.
First observe that it is possible to define a total order over $\SetOfLogSeqs(\G,\H)$.
Indeed, consider the totally ordered domain $A_{\tuple{\G,\H}}$ of the relational structure $\A_{\tuple{\G,\H}}$ (described in \cref{sec:logic}), and in particular consider the space of pairs $P=A_{\tuple{\G,\H}}\times A_{\tuple{\G,\H}}$.
First, let us define an order over $P$ by exploiting the ordering relation ``$\Prec$'' over $\A_{\tuple{\G,\H}}$ (see \cref{sec:logic}): Given two pairs $p_1=\tuple{a_1,b_1}$ and $p_2=\tuple{a_2,b_2}$ belonging to $P$, $p_1$ precedes $p_2$ in $P$ if and only if $(a_1 < a_2)\lor (a_1 = a_2 \land b_1 < b_2)$.

Now, we associate each label with a pair in $P$:
label $\IncNodeCrit{v}{G}$ is associated with pair $\tuple{o_v,o_G}\in P$, where $o_v$ and $o_G$ are the objects of $A_{\tuple{\G,\H}}$ associated with $v$ and $G$, respectively;
and label $\ExcNode{v}$ is associated with pair $\tuple{o_v,o_v}$, where $o_v$ is the object of $A_{\tuple{\G,\H}}$ associated with $v$.
Given two labels $\ell_1$ and $\ell_2$, $\ell_1$ precedes $\ell_2$ if and only if their respective associated pairs $p_1$ and $p_2$ are such that $p_1$ precedes $p_2$ in $P$.

To conclude, given two sets of labels $\Sigma_1,\Sigma_2\in\SetOfLogSeqs(\G,\H)$, $\Sigma_1$ precedes $\Sigma_2$ if and only if $\Sigma_1$ contains strictly fewer labels than $\Sigma_2$, or $\Sigma_1$ and $\Sigma_2$ contain the same number of labels and the least labels $\ell_1\in\Sigma_1$ and $\ell_2\in\Sigma_2$ on which $\Sigma_1$ and $\Sigma_2$ differ are such that $\ell_1$ precedes $\ell_2$.

Given this order, it is possible to enumerate all sets belonging to $\SetOfLogSeqs(\G,\H)$ without repetitions.

Consider now the following \emph{deterministic} algorithm \textsc{\ComputeNT} listed as \cref{alg:compute_det}, which, given two hypergraphs $\G$ and $\H$, successively generates all sets $\Sigma$ belonging to $\SetOfLogSeqs(\G,\H)$ to verify whether one of them is a good starting point to build a new transversal of $\G$ \Wrt $\H$.
A prerequisite for the correct execution of algorithm \textsc{\ComputeNT} is that the input hypergraphs satisfy the intersection property, and the purpose of procedure \textsc{\CheckIP} used in \textsc{\ComputeNT} is precisely that.
This is required because \textsc{\ComputeNT} looks for sets of labels only among those in $\SetOfLogSeqs(\G,\H)$, and \cref{lemma:ipergrafi_non-dual_iff_witness_nel_set} holds only if $\G$ and $\H$ satisfy the intersection property.
In the pseudo\nbdash-code of the algorithm, ``\textbf{return error}'' is a command triggering an error state/signal.

\begin{algorithm}[!ht]
\caption{A deterministic algorithm, derived from \textsc{\NDAlg}, computing a new transversal of $\G$ \Wrt $\H$. Here neither $\sigma(\Sigma)$ nor $\sigma(\Sigma)^+$ are explicitly stored, but they are dynamically computed as needed. We assume that $\sigma(\Sigma)=\assign{\In(\Sigma),\Ex(\Sigma)}$.}\label{alg:compute_det}
\begin{algorithmic}[1]
\Require
\Statex Hypergraphs $\G$ and $\H$ satisfy the intersection property.

\Procedure{\ComputeNT}{$\G,\H$}
\If{$\lnot$\Call{\CheckIP}{$\G,\H$}}\label{line:alg-compute-checkIP}
    \textbf{return error};
\EndIf
\For{each $\Sigma\colon \Sigma\in\SetOfLogSeqs(\G,\H)$}\label{line:alg-compute-for}
    \If{\Call{\CheckConsistencySet}{$\G,\Sigma$}}\label{line:alg-compute-checkConsistency}
        \If{\Call{\CheckDoubleWitnessAug}{$\G,\H,\Sigma$}}\label{line:alg-compute-DoubleWitness}
            \textbf{return} $\In(\Sigma)\cup\Freq(\sigma(\Sigma))$;
        \EndIf%
    \EndIf%
\EndFor
\State \textbf{return} \NullID;\label{line:alg-compute-outputNull}
\EndProcedure
\end{algorithmic}
\end{algorithm}

\begin{lemma}\label{lemma:algoritmo_quadratic_logspace}
Let $\G$ and $\H$ be two hypergraphs satisfying the intersection property.
Then, algorithm \textsc{\ComputeNT} correctly computes a new transversal of $\G$ \Wrt $\H$ (if it exists) in space $O(\log^2 \InputSize)$.
\end{lemma}
\begin{proof}
$\textsc{\ComputeNT}$ always terminates because sets belonging to $\SetOfLogSeqs(\G,\H)$ are finite and, by exploiting the order defined, can be enumerated successively without repetitions.

First, \textsc{\ComputeNT} checks the intersection property, and if this property does not hold between the two input hypergraphs, then an error state/signal is triggered.
Next, \textsc{\ComputeNT} successively enumerates all possible elements belonging to $\SetOfLogSeqs(\G,\H)$ to find a set $\Sigma$ (if one exists) such that $\sigma(\Sigma)^+$ meets Condition~\eqref{eq:double_witnessing_condition} of \cref{lemma:double_witnessing_condition}, and hence such that $\sigma(\Sigma)^+$ is a double non\nbdash-duality witness.

Observe that at line~\ref*{line:alg-compute-checkConsistency} the currently analyzed set $\Sigma$ is checked for consistency, and at line~\ref*{line:alg-compute-DoubleWitness} pair $\sigma(\Sigma)^+$ is tested to be a double non\nbdash-duality witness.
Hence, since $\G$ and $\H$ are checked at line~\ref*{line:alg-compute-checkIP} to satisfy the intersection property, by \cref{lemma:ipergrafi_non-dual_iff_witness_nel_set}, there is a set $\Sigma=\assign{\In(\Sigma),\Ex(\Sigma)}$ passing the test at line~\ref*{line:alg-compute-DoubleWitness} if and only if there is a new transversal of $\G$ \Wrt $\H$.
From $\sigma(\Sigma)^+$ being a double witness, it follows that $\In(\Sigma)\cup\Freq(\sigma(\Sigma))$ is a new transversal of $\G$ \Wrt $\H$ (the reader can see from the proof of \cref{lemma:logarithmic_double_witness} that $\In(\Sigma)\cup\Freq(\sigma(\Sigma))$ is not always a minimal transversal of $\G$).

Algorithm \textsc{\ComputeNT} correctly outputs \NullID at line~\ref*{line:alg-compute-outputNull} if no new transversal of $\G$ exists.

To conclude, let us now show that \textsc{\ComputeNT} executes within a quadratic logspace bound.
All sets $\Sigma$ generated at line~\ref*{line:alg-compute-for} contain at most $\lfloor{\log |\H|}\rfloor+1$ labels, which is $O(\log \InputSize)$, and each of these sets can be represented with $O(\log^2 \InputSize)$ bits (see the proof of \cref{theo:complexity_non-dual}).
For this reason, by re\nbdash-using of workspace, the algorithm needs only $O(\log^2 \InputSize)$ bits to represent all the sets successively tried.
\Cref{lemma:complexity_check_hitting_property,lemma:complexity_check_consistency_guess,lemma:complexity_check_guess_augmented_witness} show that all tests can be executed in $\TC0$ and hence in logarithmic space (by the inclusion $\TC0\subseteq\LogSpace$).
In fact, in order to implement those tests within a logarithmic space bound, pairs $\sigma(\Sigma)$ and ${\sigma(\Sigma)}^+$, the sets of the separated and compatible edges, and the sets of frequent and infrequent vertices, are dynamically computed in \LogSpace when needed, rather than being explicitly stored.

Observe also that the output operations can be carried out in logarithmic space.
Indeed, the elements belonging to $\In(\Sigma)\cup\Freq(\sigma(\Sigma))$ can be output successively one by one using only logarithmic workspace by exploiting Formula~\eqref{eq:assignment_from_sequence} of \cref{sec:non-deterministic_algorithm} (for each vertex $v$ it is decided whether $v$ has to be output or not).
Note that, given a vertex $v$, checking whether $v$ is a free vertex of $\sigma(\Sigma)$ is feasible in \TC0 (see the proof of \cref{lemma:complexity_check_guess_augmented_witness}), and hence in \LogSpace (from $\TC0\subseteq\LogSpace$).
Moreover, deciding whether a free vertex of $\sigma(\Sigma)$ is frequent is feasible in \TC0 (see the proof of \cref{lemma:complexity_check_guess_augmented_witness}), and hence in \LogSpace.
\end{proof}

It is an open problem whether it is possible to compute a \emph{minimal} new transversal of $\G$ in space $O(\log^2 \InputSize)$.

From the previous lemma, the following theorem directly follows. Note that the result here reported is actually a (slight) improvement over the result in the conference paper~\citep{gott13}, because we require here that the input hypergraphs satisfy the intersection property instead of the tighter condition of $\G$ and $\H$ being such that $\G\subseteq\Tr(\H)$ and $\H\subseteq\Tr(\G)$.

\begin{theorem}[improved over \citep{gott13}]
Let $\G$ and $\H$ be two hypergraphs satisfying the intersection property.
Then, computing a new (not necessarily minimal) transversal of $\G$ \Wrt $\H$ is feasible in $O(\log^2 \InputSize)$ space.
\end{theorem}

\section{Conclusions and future research}\label{sec:concl_and_future}
In this paper, we studied the computational complexity of the \Dual problem.
By using standard decomposition techniques for \Dual, we proved that after logarithmic many decomposition steps it is possible to individuate a sub\nbdash-instance of the original instance of \Dual for which verifying that it corresponds to a new transversal is feasible within complexity class \TC0.

From this we devised a new nondeterministic algorithm for \coDual whose analysis allowed us to recognize \GC{$\log^2 n$}{\TC0} as a new complexity upper bound for \coDual.
As a simple corollary of this result, we obtained also that $\coDual\in\GC{$\log^2 n$}{$\LogSpace$}$, which was conjectured by \citet{gott13}.

Moreover, the nondeterministic algorithm proposed in this paper is used to develop a simple deterministic algorithm whose space complexity is $O(\log^2 n)$.

Now some questions arise.
Is it possible to avoid the counting in the final deterministic check phase of the nondeterministic algorithm without the need of guessing more bits than $O(\log^2 n)$?
Or, more in general, without exceeding the upper bound of $O(\log^2 n)$ nondeterministic guessed bits, is it possible to devise a final deterministic test requiring a formula with strictly fewer quantifier alternations?

In the quest of finding the exact complexity of the \DUAL problem, there is another direction of investigation that could be interesting to explore.

\citet*{Grohe2006} (see also~\cite{Flum2006}) defined a hierarchy of nondeterministic classes containing those languages that, after guessing $O(\log^2 n)$ bits, can be checked by an \FO formula with a bounded number of quantifier alternations.
A different definition of this very same hierarchy can be found also in a paper by \citet{Cai1997}.
For notational convenience, let us denote these classes by \GC{$\log^2 n$}{$\mathcal{S}$}, where $\mathcal{S}$ is a sequence of logical quantifiers characterizing the quantifiers alternation in the formulas for the check of the languages in the class.
Interestingly, there are natural decision problems that are complete for classes in this hierarchy, for example the \emph{tournament dominating set problem}, and the \emph{Vapnik--Chervonenkis dimension problem}.

The former problem is defined as follows: Given a tournament $G$ (i.e., a directed graph such that for each pair of vertices $v$ and $w$ there is either an edge from $v$ to $w$, or an edge from $w$ to $v$ (but not both)) and an integer $k$, decide whether there is a dominating set in $G$ of size $k$.
It can be shown that this problem is complete for the class \GC{$\log^2 n$}{$\forall\exists$}~\cite{Cai1997,Grohe2006,Flum2006}.

The latter problem is defined as follows: Given a hypergraph $\G$ and an integer $k$, decide whether the Vapnik--Chervonenkis dimension of $\G$ is at least $k$.
It can be shown that this problem is complete for the class \GC{$\log^2 n$}{$\forall\exists\forall$}~\cite{Cai1997,Grohe2006,Flum2006}.

In fact, a new hierarchy of classes characterized by limited nondeterminism could be defined.
Indeed, we could extend the definitions given by \citet{Cai1997}, and \citet*{Grohe2006}, to define the classes of languages that, after a nondeterministic guess of $O(\log^2 n)$ bits, can be checked by a \FOC formula with a bounded number of quantifier alternations.

In particular, for \coDual, we hypothesize that the formula checking the guess of our nondeterministic algorithm, possibly rearranged and rewritten, is characterized by a quantifiers alternation $\forall\exists\mathrm{C}$ (where $\mathrm{C}$ is the counting quantifier).\footnote{Note here that \FOC formulas with bounded quantifiers could be put in relation with the levels of the logarithmic time counting hierarchy defined by \citet{Toran1988,Toran1989}.}

Given such a new hierarchy, it would be interesting to verify whether \coDual belongs to the class \GC{$\log^2 n$}{$\forall\exists\mathrm{C}$}, and even whether \coDual is complete for this class.

\section*{Acknowledgments}
G.~Gottlob's work was supported by the EPSRC Programme Grant EP/M025268/ ``VADA: Value Added Data Systems -- Principles and Architecture''.
E.~Malizia's work was mainly supported by the European Commission through the European Social Fund and by the Region of Calabria (Italy).
Malizia received additional funding from the ERC grant 246858 (DIADEM) and the above mentioned EPSRC grant ``VADA''.
We thank Thomas Eiter and Nikola Yolov for their very helpful comments on the preliminary version of this technical report.
In addition, we are grateful to the anonymous referees of the conference paper~\cite{Gottlob_Malizia2014} and of the journal paper~\cite{Gottlob_Malizia:DUAL_sicomp} for their competent and helpful comments.

\bibliographystyle{abbrvnat}
%\bibliography{ref}
%\input{ms.bbl}

\begin{appendices}

\section{It's a matter of perspectives\dots}\label{sec:perspectives}
In the literature, the hypergraph transversal problem was faced essentially from two points of view: that of Boolean functions dualization, and that of hypergraphs themselves.
In order to ease the (non\nbdash-expert) reader's task to place this work in the literature's landscape we illustrate here their connections.

Consider the \emph{Boolean domain} $\{0,1\}$.
A \emph{Boolean vector} is an element of the $n$\nbdash-dimensional Boolean space $\{0,1\}^n$.
If $x$ is a Boolean vector, $x_i$ denotes the $i$\nbdash-th component of $x$.
If $x$ and $y$ are two vectors belonging to the same $n$\nbdash-dimensional Boolean space (and hence having the same number of components), by $x\leq y$ we denote the fact that $x_i\leq y_i$ for all $1\leq i\leq n$.

An $n$\nbdash-ary \emph{Boolean function} $f$ is a mapping $f\colon \{0,1\}^n \mapsto \{0,1\}$ from the $n$\nbdash-dimensional Boolean space to a Boolean value.
Given two functions $f$ and $g$ defined on the same $n$\nbdash-dimensional Boolean space, with $f\leq g$ we mean that $f(x)\leq g(x)$ for all the Boolean vectors $x\in\{0,1\}^n$ (i.e., for all the vectors of the domain).
A function $f$ is said to be \emph{monotone} (or \emph{positive}), if, for any two vectors $x,y\in\{0,1\}^n$, $x\leq y$ implies $f(x)\leq f(y)$.

A way to represent an $n$\nbdash-ary Boolean function is through a \emph{Boolean formula} in $n$ variables $x_1,\dots,x_n$.
Boolean variables $x_1,\dots,x_n$ and their complements $\lnot x_1,\dots,\lnot x_n$ are called \emph{literals}.
A Boolean formula is a formula consisting of Boolean constants $0$ and $1$ (which we associates with \valfalse and \valtrue, respectively), literals, logical connectives `$\land$' (logical \emph{and}, or \emph{conjunction}), and `$\lor$' (logical \emph{or}, or \emph{disjunction}), and the parentheses symbols `$($' and `$)$'.
Parentheses give priority to the evaluation of the subformula between them enclosed, altering, in this way, the standard precedence of the conjunction over the disjunction.

A \emph{clause} and a \emph{term} are a disjunction and a conjunction of literals, respectively.
A clause $c$ is an \emph{implicate} of the function $f$ if $f\leq c$, while a term $t$ is an \emph{implicant} of $f$ if $t\leq f$.
A clause $c$ is a \emph{prime implicate} of $f$, and a term $t$ is a \emph{prime implicant} of $f$, if they are minimal.
Here ``minimal'' means that they are no longer an implicate of $f$ and an implicant of $f$, respectively, if a literal is removed from them.
Implicates and implicants are said to be \emph{monotone} if they consist only of positive literals.

A Boolean formula is said to be in \emph{conjunctive normal form} (or, \emph{CNF}) if it is a conjunction of clauses, while it is said to be in \emph{disjunctive normal form} (or, \emph{DNF}) if it is a disjunction of terms.
A CNF, or a DNF, formula is said to be \emph{prime} or \emph{monotone} if all its clauses, or terms, are prime or monotone, respectively.

It is well known that, a function $f$ is monotone if and only if it can be represented through a monotone CNF, or a monotone DNF, formula.
These representations are moreover unique if only prime implicates and implicants are considered.
Given a monotone Boolean function $f$, we denote by $\cnf(f)$ and $\dnf(f)$ its unique prime CNF and DNF representations, respectively.

If $f$ is a Boolean function, we denote by $f^d$ its \emph{dual} that is defined as $f^d(x)=\lnot{f}(\lnot{x})$, where $\lnot{f}$ and $\lnot{x}$ are the complements of $f$ and $x$ (each components of $x$ is complemented), respectively.
By definition, a function $f$ and its dual $f^d$ are such that $(f^d)^d=f$.
If $f$ is a monotone Boolean function, then also $f^d$ is monotone.
Let $f$ be a Boolean function, and $\gamma$ be a Boolean formula representing $f$.
From De Morgan's law, it is possible to easily compute a Boolean formula representing $f^d$ by simply exchanging the logical connectives $\land$ and $\lor$, and the constants $0$ and $1$, of $\gamma$.
We denote by $\gamma^\DM$ the Boolean formula obtained from $\gamma$ by exchanging its logical connectives and constants.
Observe that, if $f$ is a monotone Boolean function, and $\varphi=\cnf(f)$, then $\varphi^\DM=\dnf(f^d)$ (and, obviously, if $\psi=\dnf(f)$ then $\psi^\DM=\cnf(f^d)$).
This means that it is trivial to compute from $\cnf(f)$ ($\dnf(f)$, resp.) the formula $\dnf(f^d)$ ($\cnf(f^d)$, resp.).
In fact, it is a more involved task to compute from $\cnf(f)$ the formula $\cnf(f^d)$ (and from $\dnf(f)$ the formula $\dnf(f^d)$; a task that has the very same complexity).

With the aim of studying the computational complexity of the dual function computation problem, its decision variant was introduced in the literature.
This decision task consists in determining whether two given monotone prime CNF, or DNF, formulas represent dual functions.

To go back to the hypergraph transversal problem, one of the forms in which the transversal problem was approached in the literature is through the \emph{dualization} or the \emph{duality} (that is, the decision problem) perspective.
Among the vast literature published so far on the topic, we find that computing the formula $\cnf(f^d)$ from $\cnf(f)$ (or its decision flavour, that is, given two CNF formulas deciding whether they are dual) was studied, for example, by \citet{EGM03,EGMSurvey}; and that computing the formula $\dnf(f^d)$ from $\dnf(f)$ (or its decision flavour, that is, given two DNF formulas deciding whether they are dual) was studied, for example, by \citet{fred-khac-96}, and \citet{kavv-stav-03,Kavvadias2003a}.

Since the dual of a monotone Boolean function $f$ is itself monotone, if $\psi=\cnf(f^d)$ then $\psi^\DM=\dnf(f)$.
This implies that it is trivial to derive the DNF formula of a function $f$ if it has been previously computed the CNF formula of the dual function $f^d$ (and vice\nbdash-versa).
For this reason, sometimes the problem of dualization was approached as the problem of computing the formula $\dnf(f)$ when the formula $\cnf(f)$ is given in input (or vice\nbdash-versa).
Also in this case there is a decision counterpart of this problem.
Given two monotone prime Boolean formulas, one in CNF and the other in DNF, deciding whether they represent equivalent functions.
Dealing with the transversal problem from this point of view was an approach adopted, for example, by \citet{Elbassioni2006}, and \citet{Boros2010}.

Another perspective is that of self\nbdash-duality, consisting in deciding whether a Boolean function $f$ is such that $f^d=f$.
This problem is known to be characterized by the very same complexity of deciding the duality of two Boolean formulas \cite{EG95}.
In the literature, the works of \citet{Gaur1999}, and \citet{Gaur2004}, approached the transversal problem from this point of view (self\nbdash-duality of DNF prime monotone Boolean formulas).

Apart from the approaches mentioned above, there are also works in which the authors tackled ``directly'' the transversal problem dealing explicitly with hypergraphs, vertices, hyperedges, transversal sets and related concepts \cite{Eiter1991,EG95,Elbassioni2008,Boros2009,gott13}.
Some of these works at first introduce the problem from a Boolean formulas perspective before shifting to an approach more focused on hypergraphs.

Since all these problems are strictly connected each other, sometimes in the literature the transversal hypergraph is called the dual hypergraph.
By extension, often the task of computing the transversal/dual hypergraph of a given one is referred to as the dualization of a hypergraph, while the task of deciding whether two given hypergraphs are each the transversal/dual hypergraph of the other is referred to as the task of checking the duality between two hypergraphs.

In this paper we approach the hypergraph transversal problem in a explicit way by dealing with hypergraph ``entities'', as vertices and hyperedges.
In \cref{sec:decomposition} we will refer to vertices included in and excluded from the ongoing building (possibly new) transversal of $\G$.
We want to relate this notion of including/excluding vertices with the other approaches used in the literature to make clearer the similarities and the differences of the algorithm here presented with the other proposed.

Let $\varphi$ be a prime monotone CNF formula.
Let us denote by $\Hyp(\varphi)$ the hypergraph obtained from $\varphi$, consisting of a vertex $v_i$ for each variable $x_i$ of $\varphi$, and a hyperedge $H_c$ for each clause $c$ of $\varphi$.
Hyperedge $H_c$ contains exactly the vertices associated with the literals of the clause $c$ (remember that all the literals of $\varphi$ are positive since the formula is assumed to be monotone).
For example, let $f$ be the $4$\nbdash-ary monotone Boolean function such that $f(x)=1$ if and only if $x\in\{(1,1,0,1),(0,1,1,0),(1,0,1,0)\}$.
From this information it is easy to obtain the formula $\dnf(f)$, but, for the sake of the presentation, let us first show the formula $\cnf(f)$.
We ask the reader at first to ``trust'' that it is correct, we will prove it shortly.
Let
\[
\varphi=\cnf(f)=\underbrace{(x_1\lor x_2)}_{c_1}\land\underbrace{(x_2\lor x_3)}_{c_2}\land\underbrace{(x_1\lor x_3)}_{c_3}\land\underbrace{(x_3\lor x_4)}_{c_4},
\]
the hypergraph $\G=\Hyp(\varphi)$ is depicted in \cref{fig:hyp_phi}.
Note that for any prime monotone CNF formula $\phi$, the hypergraph $\Hyp(\phi)$ is Sperner.
Similarly, for a prime monotone DNF formula $\psi$ we define the hypergraph $\Hyp(\psi)$ in a similar way as above, with the only difference that hyperedges are associated with terms instead of clauses.
Also in this case, $\Hyp(\psi)$ is a Sperner hypergraph if $\psi$ is a prime monotone DNF formula.

\begin{figure}
  \centering%
  \begin{subfigure}[b]{0.33\textwidth}
  \centering%
    \includegraphics[width=0.5\textwidth]{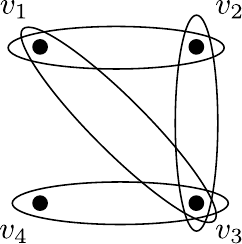}
    \caption{The hypergraph $\G=\Hyp(\varphi)$.}
    \label{fig:hyp_phi}
  \end{subfigure}%
  \hspace{0.01\textwidth}
  \begin{subfigure}[b]{0.33\textwidth}
  \centering%
    \includegraphics[width=0.5\textwidth]{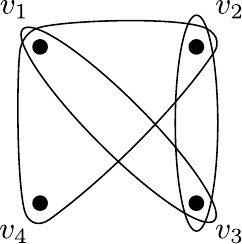}
    \caption{The hypergraph $\H=\Hyp(\psi)$.}
    \label{fig:hyp_psi}
  \end{subfigure}
  \caption{The hypegraphs obtained from $\varphi=(x_1\lor x_2)\land(x_2\lor x_3)\land(x_1\lor x_3)\land(x_3\lor x_4)$ and $\psi=(x_1\land x_2\land x_4)\lor(x_1\land x_3)\lor(x_2\land x_3)$. Since $\varphi=\cnf(f)$ and $\psi=\dnf(f)$, for the very same monotone Boolean function $f$, $\Hyp(\psi)=\Tr(\Hyp(\varphi))$ (and vice\nbdash-versa).}\label{fig:hypegraphs_from_formulas}
\end{figure}

Let us now focus on the task of computing $\psi=\dnf(f)$ from $\varphi$.
Since we have to build a prime DNF formula $\psi$ equivalent to the monotone CNF $\varphi$, all the prime terms of $\psi$ will be monotone too.
Each of the terms of $\psi$ will contain as (positive) literals a minimal set of variables sufficient to satisfy $\varphi$ (when the Boolean value \valtrue is assigned to them).
This is tantamount to choose, for every clause of $\varphi$, a literal to which we assign \valtrue.
For example, choosing $x_1$ satisfies clause $c_1$, choosing $x_2$ satisfies $c_2$, and $c_3$ is already satisfied by having chosen $x_1$.
To satisfy clause $c_4$ we can choose $x_4$, and so we obtain the prime term $(x_1\land x_2\land x_4)$, or $x_3$, and so we obtain the non\nbdash-prime term $(x_1\land x_2\land x_3)$.
By removing $x_1$ or $x_2$ from the non\nbdash-prime term we obtain the two prime terms $(x_1\land x_3)$ and $(x_2\land x_3)$.
The reader can easily check that, in fact, $\psi=\dnf(f)=(x_1\land x_2\land x_4)\lor(x_1\land x_3)\lor(x_2\land x_3)$.
So, to derive the DNF form of $f$ from its CNF representation, we have ``covered'' all the clauses of $\varphi$, that is, we have essentially evaluated the minimal transversals of $\G=\Hyp(\varphi)$.
In \cref{fig:hyp_psi} is reported the hypergraph $\H=\Hyp(\psi)$, and the reader can see that $\Hyp(\psi)=\Tr(\Hyp(\varphi))$.

Computing the prime CNF form of a monotone Boolean function from its prime DNF form can be done in a similar way.
Let us take into consideration again the function $f$ represented in DNF by
\[
\psi=\dnf(f)=\underbrace{(x_1\land x_2\land x_4)}_{t_1}\lor\underbrace{(x_1\land x_3)}_{t_2}\lor\underbrace{(x_2\land x_3)}_{t_3}.
\]
In order to compute the formula $\varphi=\cnf(f)$ from $\psi$, just consider the folowing.
Every clause $c$ of $\varphi$ is an implicate of $f$, i.e., $f\leq c$, from which it follows that $\lnot c\leq\lnot f$.
So, by looking for how to falsify $f$, and $\psi$ in particular, we can compute every single complemented prime clause $\lnot c$ of $\varphi$, which, complemented again, allows us to easily obtain the prime clause $c$.
A way to falsify $\psi$ is that of choosing, for every term of $\psi$, a literal to which we assign \valfalse.
For example, to falsify term $t_1$ we can choose $x_1$, and term $t_2$ is already falsified by the previous choice of $x_1$.
To falsify $t_3$ we can take $x_2$ or $x_3$, hence obtaining $\lnot c_1=(\lnot x_1\land\lnot x_2)$ and $\lnot c_3=(\lnot x_1\land\lnot x_3)$, and thus $c_1=\lnot(\lnot c_1)=(x_1\lor x_2)$, and $c_3=\lnot(\lnot c_3)=(x_1\lor x_3)$.
Following this line of reasoning the reader can compute the whole formula $\varphi=\cnf(f)=(x_1\lor x_2)\land(x_2\lor x_3)\land(x_1\lor x_3)\land(x_3\lor x_4)$.
Also in this case, the way to compute $\varphi$ from $\psi$ was essentially that of computing all the minimal transversals of the hypergraph $\H=\Hyp(\psi)$ (see \cref{fig:hyp_psi}) to obtain the hypergraph $\G=\Hyp(\varphi)$ (see \cref{fig:hyp_phi}).
Observe that $\G$ and $\H$ are such that $\G=\Tr(\H)$.

We have just seen that the transversal hypergraph computation is essentially the same task of computing the DNF prime form of a monotone Boolean function when its CNF prime form is given in input, and vice\nbdash-versa.
To see that this is also related to the computation of the dual function, i.e., the dualization problem properly said, the step is simple.
If, from the formula $\varphi=\cnf(f)$, we derive the formula $\psi=\dnf(f)$ through a procedure that essentially computes $\Tr(\Hyp(\varphi))$, then $\psi^\DM$ is a prime monotone CNF formula such that $\psi^\DM=\cnf(f^d)$ (remember that the `$\DM$' operator only interchanges $\land$ with $\lor$, and $0$ with $1$, and no negation $\lnot$ operator is introduced in the newly generated formula).
Similarly, for DNF formulas, if from $\psi=\dnf(f)$ we obtain $\varphi=\cnf(f)$ (through the computation of $\Tr(\Hyp(\psi))$), then $\varphi^\DM=\dnf(f^d)$.

Since we will deal with the decision problem version of the transversal hypergraph problem, let us now focus our attention on the duality problem.
This problem can be formulated in two equivalent forms: (1) given two prime CNF (or DNF) formulas, decide whether they represent dual Boolean functions; (2) given a prime monotone CNF formula and a prime monotone DNF formula, decide whether these two formulas are equivalent (i.e., they represent the same monotone Boolean function).
These problems, for the discussion above, are equivalent to deciding whether two given hypergraphs are each the transversal hypergraph of the other.

To relate our work to the other in the literature, let us see what including or excluding a vertex in our approach means for the other works focused on Boolean formulas.
Suppose that $\G$ and $\H$ are two hypergraphs, and we want to decide whether $\H=\Tr(\G)$.
In order to answer ``no'' to this question, if $\H$ consists only of transversal of $\G$, we need to find a transversal of $\G$ that is not in $\H$.
We call such a transversal, a \emph{new transversal of $\G$ \Wrt $\H$.}
To do so, inspired also by the algorithm of Gaur~\cite{Gaur1999,Gaur2004}, it is useful to keep track of two sets.
The set of the vertices already included in the ongoing building transversal, and the set of the vertices already excluded from the ongoing building transversal.
In this way, it is easy to check what hyperedges of $\G$ are already covered by and what hyperedges of $\H$ are already different from the candidate for a new transversal (more details of this will be given in \cref{sec:decomposition}).

Following the notation of \citet{Elbassioni2008} and \citet{Boros2009}, given a hypergraph $\G$ and a set $S$ of vertices,
we define hypergraphs $\G_S=\tuple{S,\{G\in \G\mid G\subseteq S\}}$, and $\G^S=\tuple{S,\min(\{G\cap S\mid G\in\G\})}$, where $\min(\H)$, for any hypergraph $\H$, denotes the set of inclusion minimal edges of $\H$.

Now assume, for example, that we want to include a particular vertex $v$ in the new transversal of $\G$.
After the inclusion of $v$, all the hyperedges of $\G$ containing $v$ are covered and hence they no longer need to be considered in the construction of the new transversal of $\G$.
So, the new hypergraph ``$\G$'' (of the pair) to be taken into consideration is $\G_{V\setminus\{v\}}$.
See, for example, \cref{fig:hyp_selection} in which is represented the graph $\G_{V\setminus\{v_1\}}$, where $\G$ is the hypergraph of \cref{fig:hyp_phi}.
On the other hand, only the exclusion of a vertex belonging to an edge $H\in\H$ makes $H$ certainly different from the ongoing building transversal.
In fact, when, on the contrary, vertices are included, the ongoing building transversal could grow to the point to be a superset of $H$, and hence the just built transversal would not be new \Wrt $\H$.
So, since including a vertex in the new transversal of $\G$ does not make any of the hyperedges of $\H$ different from the new candidate one, then all the hyperedges of $\H$ still need to be considered, apart from removing the included vertex and shrinking the proper hyperedges.
That is, the new hypergraph ``$\H$'' (of the pair) to be considered is $\H^{V\setminus\{v\}}$.
See \cref{fig:hyp_restriction} in which is depicted the graph $\H^{V\setminus\{v_1\}}$, where $\H$ is the hypergraph of \cref{fig:hyp_psi}.

\begin{figure}
  \centering%
  \begin{subfigure}[b]{0.33\textwidth}
  \centering%
    \includegraphics[width=0.5\textwidth]{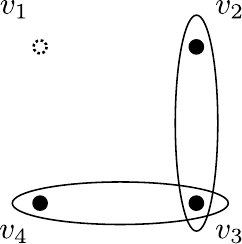}
    \caption{The hypergraph $\G_{V\setminus\{v_1\}}$.}
    \label{fig:hyp_selection}
  \end{subfigure}%
  \hspace{0.01\textwidth}
  \begin{subfigure}[b]{0.33\textwidth}
  \centering%
    \includegraphics[width=0.5\textwidth]{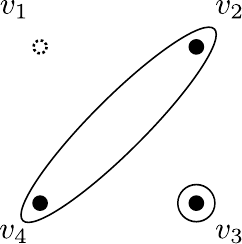}
    \caption{The hypergraph $\H^{V\setminus\{v_1\}}$.}
    \label{fig:hyp_restriction}
  \end{subfigure}
  \caption{The result of including vertex $v_1$ in the aim of building a new transversal of $\G$ \Wrt $\H$. Hypergraphs $\G$ and $\H$ are those depicted in \cref{fig:hyp_phi,fig:hyp_psi}, respectively. The dashed vertices do not belong any longer to the hypergraphs.}\label{fig:example_hypegraphs_including_vertex}
\end{figure}

We can observe the following connection.
Given a pair of hypergraphs $\tuple{\G,\H}$, with the aim of building a new transversal of $\G$ \Wrt $\H$, in our approach including a vertex $v_i$ in a candidate new transversal is equivalent to consider the ``updated'' pair $\tuple{\G_{V\setminus\{v_i\}},\H^{V\setminus\{v_i\}}}$ when we use the notation of \citet{Elbassioni2008} and \citet{Boros2009}.

Let us link this to the works focused on Boolean formulas, and in particular to the problem of deciding whether a CNF and a DNF formula are equivalent.
Consider the CNF formula $\varphi$ such that $\G=\Hyp(\varphi)$, and the DNF formula $\psi$ such that $\H=\Hyp(\psi)$.
Let $\tuple{\varphi,\psi}$ be the pair of formulas for which it has to be decided whether $\varphi$ and $\psi$ are equivalent.
We are going to see that including a vertex $v_i$ is equivalent to assigning \valtrue to the variable $x_i$ related to $v_i$, in both $\varphi$ and $\psi$.
Indeed, in the formula $\varphi$, since it is in CNF, assigning \valtrue to $x_i$ means covering all the clauses containing $x_i$, and hence they do not need to be considered further (as in $\G_{V\setminus\{v_i\}}$).
While assigning \valtrue to $x_i$ in $\psi$, since it is in DNF, alters its terms containing $x_i$ by deleting that literal, and does not alter at all the other terms (as in $\H^{V\setminus\{v_1\}}$).
Instead, if $\varphi$ is in DNF, and $\psi$ in CNF, including $v_i$ is equivalent to assigning \valfalse to $x_i$ in $\varphi$ and $\psi$.

By extension, including a set of vertices $S$ is equivalent to consider the new pair of hypergraphs $\tuple{\G_{V\setminus S},H^{V\setminus S}}$, and to assign \valtrue (\valfalse, resp.) to all the variables whose related vertices are in $S$ if $\varphi$ is in CNF (DNF, resp.), and $\psi$ is in DNF (CNF, resp.).

Similarly, excluding a set of vertices $S$ from the new transversal is equivalent to consider the new pair of hypergraphs $\tuple{\G^{V\setminus S},\H_{V\setminus S}}$, and to assign \valfalse (\valtrue, resp.) to all the variables whose related vertices are in $S$ if $\varphi$ is in CNF (DNF, resp.), and $\psi$ is in DNF (CNF, resp.).

By combining the two observations above, when $\In$ is a set of included vertices, and $\Ex$ is a set of excluded vertices, the updated pair is $\tuple{(\G_{V\setminus\In})^{V\setminus(\In\cup\Ex)},(\H_{V\setminus\Ex})^{V\setminus(\In\cup\Ex)}}$.
This expression will be used in \cref{sec:decomposition}.

When we consider the problem of deciding whether two formulas are dual, the link is as follows.
Consider the two CNF (DNF, resp.) formulas $\varphi$ and $\psi$ such that $\G=\Hyp(\varphi)$ and $\H=\Hyp(\psi)$.
Let $\tuple{\varphi,\psi}$ be the pair of formulas for which it has to be decided whether $\varphi$ and $\psi$ are dual.
Including a vertex $v_i$ in a candidate new transversal of $\G$ is equivalent to assign to $x_i$ \valtrue (\valfalse, resp.) in $\varphi$ and \valfalse (\valtrue, resp.) in $\psi$.
While excluding $v_i$ from the candidate new transversal of $\G$ is equivalent to assign to $x_i$ \valfalse (\valtrue, resp.) in $\varphi$ and \valtrue (\valfalse, resp.) in $\psi$.
These considerations can be generalized in the obvious way to the inclusion and exclusion of sets of vertices.

\section{Proofs of properties stated in Section~\ref{sec:prelim}}\label{sec:app_prove_property_transversal}
\noindent\textbf{\Cref*{lemma:transMin-ogniVertexCritico}.}
\textit{Let $\G$ be a hypergraph, and let $T\subseteq V$ be a transversal of $\G$.
Then, $T$ is a minimal transversal of $\G$ if and only if every vertex $v\in T$ is critical (and hence there is an edge $G_v\in\G$ witnessing so).}\nopagebreak\begin{proof}
If $\emptyset \in \G$, then the statement trivially holds because there is no transversal of $\G$ at all.
On the other hand, if $\G = \emptyset$, the only minimal transversal of $\G$ is the empty set for which the statement of the lemma holds.
Let us now consider the case in which $\G$ contains only non\nbdash-empty edges.

$(\Rightarrow)$
Let $T$ be a minimal transversal of $\G$, and assume by contradiction that there is a vertex $v\in T$ that is not critical.
This means that, for every edge $G\in\G$, if $v\in (G\cap T)$, then $|G\cap T|\geq 2$.
Now consider the set $T'=T\setminus\{v\}$.
Clearly, $|G\cap T'|\geq 1$ for every edge $G \in \G$.
Therefore, set $T'\subset T$ is a transversal of $\G$: a contradiction, because we are assuming that $T$ is a minimal transversal of $\G$.

$(\Leftarrow)$
Assume that all vertices of transversal $T$ are critical.
Consider any set $T'\subset T$ and let $v\in T\setminus T'$ be a vertex.
Since $v$ is critical, there is an edge $G_v\in\G$ such that $T\cap G_v=\{v\}$.
Hence, $T'\cap G_v=\emptyset$ and $T'$ is not a transversal of $\G$.
Therefore, $T$ is a minimal transversal of $\G$.
\end{proof}

\noindent\textbf{\Cref*{lemma:newTransversals_complementari}.}
\textit{Let $\G$ and $\H$ be two hypergraphs.
A set of vertices $T\subseteq V$ is a new transversal of $\G$ \Wrt $\H$ if and only if $\compl{T}$ is a new transversal of $\H$ \Wrt $\G$.}\nopagebreak\begin{proof}
If $\G$ or $\H$ contains an empty edge, then they do not admit transversals and independent sets.
Hence, in these cases, the statement trivially holds.
Let us consider the case in which $\G$ and $\H$ do not have empty edges.
If $T$ is a new transversal of $\G$ \Wrt $\H$, $T$ is an independent set of $\H$, and hence, for every edge $H\in\H$, there is a vertex $v\in H\setminus T$.
Therefore, $v\in\compl{T}$, and so $\compl{T}$ is a transversal of $\H$.
Moreover, set $\compl{T}$ cannot be a superset of any edge $G$ of $\G$, for otherwise $T$ would not intersect $G$, and $T$ would not be a transversal of $\G$.
Hence, $\compl{T}$ is an independent set of $\G$, which proves that $\compl{T}$ is a new transversal of $\H$ \Wrt $\G$.
For the other direction, just swap the roles of $\G$ and $\H$.
Therefore, if $\compl{T}$ is a new transversal of $\H$, then $\compl{\compl{T}}=T$ is a new transversal of $\G$.
\end{proof}

\noindent\textbf{\Cref*{lemma:no_new_transversal_iff_dual}.}
\textit{Let $\G$ and $\H$ be two hypergraphs.
Then, $\G$ and $\H$ are dual if and only if $\G$ and $\H$ are simple, satisfy the intersection property, and there is no new transversal of $\G$ \Wrt $\H$.}\nopagebreak\begin{proof}\hspace{0pt}
$(\Rightarrow)$
If $\G$ and $\H$ are dual, then they have to be simple.
Indeed, assume by contradiction that $\G$ is not simple.
Then there are two different edges $G',G''\in\G$ such that $G'\subset G''$.
There are two cases: either $G'$ is a transversal of $\H$, or it is not.
If $G'$ is a transversal of $\H$, then $G''$ is not a minimal transversal of $\H$, and hence $\G\neq\Tr(\H)$: a contraction, because we are assuming $\G$ and $\H$ to be dual.
On the other hand, if $G'$ is not a transversal of $\H$, then, again, $\G\neq\Tr(\H)$: a contraction, because we are assuming $\G$ and $\H$ to be dual.
Similarly, it can be shown that $\H$ has to be simple.

Assume now by contradiction that $\G$ and $\H$ do not satisfy the intersection property.
Then there are edges $G\in\G$ and $H\in\H$ such that $G\cap H=\emptyset$: a contradiction, because $\G$ and $\H$ are dual.

We now show that there is no new transversal of $\G$ \Wrt $\H$.
If $\G$ or $\H$ contains an empty edge, then they do no admit transversals and independent sets, hence there is no new transversal of $\G$ \Wrt $\H$.
If $\G$ or $\H$ is an empty hypergraph, then they have to be trivially dual.
This means that one of them contains an empty edge and hence there is no new transversal of $\G$ \Wrt $\H$ (see above).
Let us now consider the case in which both $\G$ and $\H$ contain only non\nbdash-empty edges.
Assume by contradiction that there is a new transversal $T$ of $\G$ \Wrt $\H$.
We can assume \Wlog that $T$ is minimal.
Since $T$ is an independent set of $\H$, for all edges $H\in\H$, there is a vertex $v$ such that $v\in H\setminus T$, and hence $H\neq T$.
Thus, $T$ is a minimal transversal of $\G$ missing in $\H$, and hence $\H\neq\Tr(\G)$: a contradiction, because we are assuming $\G$ and $\H$ to be dual.

$(\Leftarrow)$
First consider the case in which one of the hypergraphs contains an empty edge, and let us assume that $\emptyset \in \G$.
Since $\G$ is simple, it must be the case that $\G = \{\emptyset\}$.
Moreover, from $\G$ and $\H$ satisfying the intersection property follows that $\H = \emptyset$.
Observe that there is no new transversal of $\G$ \Wrt $\H$ and that $\G$ and $\H$ are (trivially) dual.
Similarly, if $\emptyset \in \H$, by swapping the roles of $\G$ and $\H$ in the discussion above, we can show that $\G$ and $\H$ are (trivially) dual.

Consider now the case in which one of the hypergraphs is empty, and let us assume that $\G = \emptyset$.
We claim that $\emptyset \in \H$.
Assume by contradiction that $\emptyset \notin \H$.
Then, the empty set is a transversal of $\G$ that is an independent set of $\H$: a contradiction, because we are assuming that there is no new transversal of $\G$ \Wrt $\H$.
Therefore, if $\G = \emptyset$, then $\emptyset \in \H$, and we are again in the case analyzed above.
Similarly, it can be shown that, if $\H = \emptyset$, then $\emptyset \in \G$, and we are again in the case analyzed above.

Let us now consider the case in which both $\G$ and $\H$ contain only non\nbdash-empty edges.
Let $T$ be a minimal transversals of $\G$.
Since there is no new transversal of $\G$, $T$ is not an independent set of $\H$, and hence there is an edge $H\in\H$ such that $H\subseteq T$.
Since all edges of $\H$ are transversals of $\G$ (because $\G$ and $\H$ satisfy the intersection property), we claim that $H=T$.
Indeed, since $H\subseteq T$ and $T$ is a minimal transversal of $\G$, from $H$ being a transversal of $\G$ it cannot happen that $H$ is strictly contained in $T$, hence $H=T$.
Because this holds for all the minimal transversals of $\G$, i.e., for any minimal transversal $T$ of $\G$, $T$ is in $\H$, it follows that $\Tr(\G)\subseteq \H$.
Moreover, because $\H$ is a simple hypergraph it must be the case that $\H=\Tr(\G)$.
Thus, hypergraphs $\G$ and $\H$ are dual.
\end{proof}

\section{Proofs of properties stated in Section~\ref{sec:decomposition}}\label{sec:app_prove_decomposition}

\noindent\textbf{\Cref*{lemma:assignment_coherent_new_Tr_then_not_covering}.}
\textit{%
Let $\G$ and $\H$ be two hypergraphs, and let $\sigma=\assign{\In,\Ex}$ be an assignment.
\begin{enumerate}[label=$(\mathit{\alph*})$]
\item If $\Ex$ is covering, then there is no set of vertices coherent with $\sigma$ that is a transversal of $\G$.
\item If $\In$ is covering, then there is no set of vertices coherent with $\sigma$ that is an independent set of $\H$.
\end{enumerate}
Hence, if $\sigma$ is covering, then there is no set of vertices coherent with $\sigma$ that is a new transversal of $\G$ \Wrt $\H$.
}\nopagebreak\begin{proof}\hspace{0pt}
\begin{enumerate}[label=$(\mathit{\alph*})$]
\item
If $\emptyset \in \G$, then there is no transversal of $\G$ at all and point $(a)$ trivially holds, because the consequent is true.
On the other hand, if $\G = \emptyset$, then any set $\Ex$ is trivially non\nbdash-covering and point $(a)$ trivially holds, because the antecedent is false.
Consider now the case in which $\G$ contains only non\nbdash-empty edges.
Assume by contradiction that there is an edge $G\in\G$ such that $G\subseteq\Ex$ and that there is a transversal $T$ of $\G$ coherent with $\sigma$.
Since $T$ is a transversal of $\G$, from $T\cap G\neq\emptyset$ follows $T\cap\Ex\neq\emptyset$: a contradiction, because $\sigma$ is coherent with $T$.

\item
If $\emptyset \in \H$, then there is no independent set of $\H$ at all and point $(b)$ trivially holds, because the consequent is true.
On the other hand, if $\H = \emptyset$, then any set $\In$ is trivially non\nbdash-covering and point $(b)$ trivially holds, because the antecedent is false.
Consider now the case in which $\H$ contains only non\nbdash-empty edges.
If $\In$ is covering, then there is an edge $H\in\H$ such that $H\subseteq\In$.
Clearly, no set of vertices $T$ coherent with $\sigma$ can be an independent set of $H$, because from $T$ being coherent with $\sigma$ follows $\In\subseteq T$, and hence $H\subseteq T$.
\end{enumerate}
\end{proof}

\begin{lemma}\label{lemma:summary_properties}
Let $\G$ and $\H$ be two hypergraphs, and let $\sigma=\assign{\In,\Ex}$ be an assignment.
Then:
\begin{enumerate}[label=$(\mathit{\alph*})$]
\item
\begin{itemize}[label=$-$,leftmargin=*,noitemsep]
\item $\emptyset \in \G \Leftrightarrow \emptyset \in \G_{V\setminus\In}$; and
\item $\emptyset \in \H \Leftrightarrow \emptyset \in \H_{V\setminus\Ex}$;
\end{itemize}

\item
\begin{itemize}[label=$-$,leftmargin=*,noitemsep]
\item $\emptyset \in \G_{V\setminus\In} \Rightarrow \G(\sigma) = \{\emptyset\}$; and,
\item $\emptyset \in \H_{V\setminus\Ex} \Rightarrow \H(\sigma) = \{\emptyset\}$;
\end{itemize}

\item
\begin{itemize}[label=$-$,leftmargin=*,noitemsep]
\item $\G = \emptyset \Rightarrow \G_{V\setminus\In} = \emptyset$; and
\item $\H = \emptyset \Rightarrow \H_{V\setminus\Ex} = \emptyset$;
\end{itemize}

\item
\begin{itemize}[label=$-$,leftmargin=*,noitemsep]
\item $\G_{V\setminus\In} = \emptyset \Leftrightarrow \G(\sigma) = \emptyset$; and,
\item $\H_{V\setminus\Ex} = \emptyset \Leftrightarrow \H(\sigma) = \emptyset$;
\end{itemize}

%%%%%%%%%%%%%%%
%%%%%%%%%%%%%%%
%%%%%%%%%%%%%%%

\item
\begin{itemize}[label=$-$,leftmargin=*,noitemsep]
\item $\G(\sigma) = \emptyset$ if and only if $\In$ is a transversal of $\G$; and,
\item $\H(\sigma) = \emptyset$ if and only if $\Ex$ is a transversal of $\H$;
\end{itemize}

\item
\begin{itemize}[label=$-$,leftmargin=*,noitemsep]
\item $\G(\sigma) = \{\emptyset\}$ if and only if $\Ex$ is covering; and,
\item $\H(\sigma) = \{\emptyset\}$ if and only if $\In$ is covering;
\end{itemize}

\item
\begin{itemize}[label=$-$,leftmargin=*,noitemsep]
\item $\G(\sigma)$ contains non\nbdash-empty edges if and only if $\In$ is not a transversal of $\G$ and $\Ex$ is not covering; and,
\item $\H(\sigma)$ contains non\nbdash-empty edges if and only if $\Ex$ is not a transversal of $\H$ and $\In$ is not covering.
\end{itemize}
\end{enumerate}
\end{lemma}
\begin{proof}\hspace{0pt}
\begin{enumerate}[label=$(\mathit{\alph*})$]
\item
The property follows from the fact that the empty edge is subset of any set of vertices and from $\G_{V\setminus\In} \subseteq \G$ and $\H_{V\setminus\Ex} \subseteq \H$.

\item
The property follows from the fact that the intersection of the empty edge with any set of vertices is the empty set and from the fact that $\G(\sigma)$ and $\H(\sigma)$ undergo a minimization operation.

\item
The property follows from $\G_{V\setminus\In} \subseteq \G$ and $\H_{V\setminus\Ex} \subseteq \H$.

\item
We prove the property for $\G$.
The proof for $\H$ is symmetric.

$(\Rightarrow)$
If $\G_{V\setminus\In} = \emptyset$, then $\G(\sigma) = \emptyset$ by definition.

$(\Leftarrow)$
Assume that $\G(\sigma) = \emptyset$ and let us assume by contradiction that $\G_{V\setminus\In} \neq \emptyset$ (i.e., $\G_{V\setminus\In}$ contains edges).
There are two cases: either (1) there is an edge $G\in\G_{V\setminus\In}$ such that $G\cap(V\setminus(\In\cup\Ex)) = \emptyset$ (observe that it may be the case that such edge $G$ is the empty one), or (2) there is not such an edge (i.e., all edges in $\G_{V\setminus\In}$ have a non\nbdash-empty intersection with $V\setminus(\In\cup\Ex)$).
In Case (1), since there is an edge $G\in\G_{V\setminus\In}$ such that $G\cap(V\setminus(\In\cup\Ex))$ is empty, $\emptyset \in \G(\sigma)$: a contradiction, because we are assuming $\G(\sigma) = \emptyset$.
For Case (2), since, for all edges $G\in\G_{V\setminus\In}$, $G\cap(V\setminus(\In\cup\Ex))$ is not empty, $\G(\sigma)$ contains non\nbdash-empty edges: a contradiction, because we are assuming $\G(\sigma) = \emptyset$.
Therefore, it must be the case that $\G_{V\setminus\In} = \emptyset$.

\item
We prove the property for $\G(\sigma)$.
The proof for $\H(\sigma)$ is symmetric.

If $\emptyset \in \G$, then there is no transversal of $\G$ at all and, by points $(a)$ and $(b)$, $\G(\sigma) = \{\emptyset\}$.
Hence, the property trivially holds ($\mathit{false}$ if and only if $\mathit{false}$).

If $\G = \emptyset$, then any set of vertices is a transversal of $\G$ and, by points $(c)$ and $(d)$, $\G(\sigma) = \emptyset$.
Hence, the property trivially holds ($\mathit{true}$ if and only if $\mathit{true}$).

Assume that $\G$ contains only non\nbdash-empty edges.
Consider $\G_{V\setminus\In}$ and observe that, by definition, when $\G \neq \emptyset$, $\G_{V\setminus\In} = \emptyset$ if and only if $\In$ is a transversal of $\G$.
Therefore, by point $(d)$, this property follows.

\item
We prove the property for $\G(\sigma)$.
The proof for $\H(\sigma)$ is symmetric.

If $\emptyset \in \G$, then any set of vertices $\Ex$ is covering, because any $\Ex$ is superset of the empty edge contained in $\G$, and, by points $(a)$ and $(b)$, $\G(\sigma) = \{\emptyset\}$.
Hence, the property trivially holds ($\mathit{true}$ if and only if $\mathit{true}$).

If $\G = \emptyset$, then any set of vertices $\Ex$ is non\nbdash-covering, because there is no edge $G\in\G$ for which $G\subseteq\Ex$, and, by points $(c)$ and $(d)$, $\G(\sigma) = \emptyset$.
Hence, the property trivially holds ($\mathit{false}$ if and only if $\mathit{false}$).

Assume that $\G$ contains only non\nbdash-empty edges.
First note that, since $\G$ contains only non\nbdash-empty edges, by point $(a)$, $\emptyset \notin \G_{V\setminus\In}$.

$(\Rightarrow)$
Since $\G(\sigma) = \{\emptyset\}$, by point $(d)$, $\G_{V\setminus\In} \neq \emptyset$.
Therefore, $\G_{V\setminus\In}$ contains only non\nbdash-empty edges.
Let us assume by contradiction that $\Ex$ is not covering, and let $G\in\G_{V\setminus\In}$ be an edge.
By definition of $\G_{V\setminus\In}$, for any $G\in\G_{V\setminus\In}$, $G\subseteq V\setminus\In$, and hence $G\cap\In=\emptyset$.
Since $\G_{V\setminus\In}\subseteq\G$ and $\Ex$ is not covering, there is a vertex $v\in(G\setminus\Ex)$.
From $G\cap\In=\emptyset$ and $v\in(G\setminus\Ex)$, it follows that $v\in (G\cap(V\setminus(\In\cup\Ex)))$ and hence that $G\cap(V\setminus(\In\cup\Ex))\neq\emptyset$.
Therefore, $\G(\sigma)$ contains a non-empty edge: a contradiction, because we are assuming $\G(\sigma) = \{\emptyset\}$.
Thus, $\Ex$ is covering.

$(\Leftarrow)$
Remember that $\G_{V\setminus\In}$ contains all and only the edges $G\in\G$ such that $G\cap\In=\emptyset$.
Since $\Ex$ is covering, there is an edge $G\in\G$ such that $G\subseteq\Ex$. Observe that $G$ must belong to $\G_{V\setminus\In}$ (because $\In$ and $\Ex$ are disjoint).
Therefore, from $G\in\G_{V\setminus\In}$ and $G\cap(V\setminus(\In\cup\Ex))=\emptyset$, it follows that $\G(\sigma) = \{\emptyset\}$.

\item
The property follows from points $(e)$ and $(f)$.
\end{enumerate}
\end{proof}

\begin{lemma}\label{lemma:covering_assign_generates_dual}
Let $\G$ and $\H$ be two hypergraphs satisfying the intersection property, and let $\sigma=\assign{\In,\Ex}$ be a covering assignment.
Then:
\begin{enumerate}[label=$(\mathit{\alph*})$]
\item
$\In$ and $\Ex$ cannot be both covering;

\item
$\G(\sigma)$ and $\H(\sigma)$ are trivially dual.
In particular, if $\In$ is covering, then $\G(\sigma) = \emptyset$ and $\H(\sigma) = \{\emptyset\}$; and, symmetrically, if $\Ex$ is covering, then $\G(\sigma) = \{\emptyset\}$ and $\H(\sigma) = \emptyset$.
\end{enumerate}
\end{lemma}
\begin{proof}\hspace{0pt}
\begin{enumerate}[label=$(\mathit{\alph*})$]
\item
If $\G = \emptyset$ (resp., $\H = \emptyset$), then $\Ex$ (resp., $\In$) is trivially non\nbdash-covering, because there is no edge $G\in\G$ (resp., $H\in\H$) such that $G\subseteq\Ex$ (resp., $H\subseteq\In$).

If $\emptyset \in \G$ (resp., $\emptyset \in \H$), then $\H = \emptyset$ (resp., $\G = \emptyset$) by the intersection property, and hence we are again in the case above.

Assume that both $\G$ and $\H$ contain only non\nbdash-empty edges.
Assume by contradiction that there are two edges $G\in\G$ and $H\in\H$ such that $G\subseteq\Ex$ and $H\subseteq\In$.
Because of the intersection property, from $G\cap H\neq\emptyset$ follows $\In\cap\Ex\neq\emptyset$: a contradiction, because $\sigma$ is an assignment.

\item
If $\emptyset \in \G$, then $\H = \emptyset$ by the intersection property.
In this case, any set of vertices $\Ex$ is covering, because any $\Ex$ is superset of the empty edge contained in $\G$.
By \cref{lemma:summary_properties} points $(a)$ and $(b)$, $\G(\sigma) = \{\emptyset\}$, and, by \cref{lemma:summary_properties} points $(c)$ and $(d)$, $\H = \emptyset$.
Symmetrically, if $\emptyset \in \H$, then $\G = \emptyset$ by the intersection property, $\In$ is covering, $\G = \emptyset$, and $\H(\sigma) = \{\emptyset\}$.

If $\G = \emptyset$, then $\H \neq \emptyset$ for otherwise $\sigma$ would not be covering.
Since $\G = \emptyset$, for $\sigma$ to be covering it must be the case that $\In$ is covering.
By \cref{lemma:summary_properties} point $(f)$, $\H(\sigma) = \{\emptyset\}$, and by \cref{lemma:summary_properties} points $(c)$ and $(d)$, $\G(\sigma) = \emptyset$.
Symmetrically, if $\H = \emptyset$, then $\G \neq \emptyset$ for otherwise $\sigma$ would not be covering, and it can be shown that $\Ex$ is covering, $\G(\sigma) = \{\emptyset\}$, and $\H = \emptyset$.

Assume that both $\G$ and $\H$ contain only non\nbdash-empty edges.
If $\In$ is covering, then $\In$ is a transversal of $\G$ because $\G$ and $\H$ satisfy the intersection property.
Thus, by \cref{lemma:summary_properties} point $(e)$, $\G(\sigma)=\emptyset$.
Since $\In$ is covering, by \cref{lemma:summary_properties} point $(f)$, $\H(\sigma)=\{\emptyset\}$.
Symmetrically, it can be shown that, if $\Ex$ is covering, then $\G(\sigma)=\{\emptyset\}$ and $\H(\sigma)=\emptyset$.
\end{enumerate}
\end{proof}

\begin{lemma}\label{lemma:proiezione_hitting_property}
Let $\G$ and $\H$ be two hypergraphs.
Then, $\G$ and $\H$ satisfy the intersection property if and only if, for all assignments $\sigma$, $\G(\sigma)$ and $\H(\sigma)$ satisfy the intersection property.
\end{lemma}
\begin{proof}
$(\Rightarrow)$
If $\G = \emptyset$ (resp., $\H = \emptyset$), then, by \cref{lemma:summary_properties} points $(c)$ and $(d)$, $\G(\sigma) = \emptyset$ (resp., $\H(\sigma) = \emptyset$).
Since $\G(\sigma) = \emptyset$ (resp., $\H(\sigma) = \emptyset$), trivially $\G(\sigma)$ and $\H(\sigma)$ satisfy the intersection property.

If $\emptyset \in \G $ (resp., $\emptyset \in \H$), then $\H = \emptyset$ (resp., $\G = \emptyset$) by the intersection property, and we are again in the case above.

Assume that both $\G$ and $\H$ contain only non\nbdash-empty edges.
Let $\sigma=\assign{\In,Ex}$.
If $\sigma$ is covering, then, by \cref{lemma:covering_assign_generates_dual} point $(b)$, $\G(\sigma)$ and $\H(\sigma)$ are (trivially) dual, and hence they satisfy also the intersection property.

In case $\sigma$ is non\nbdash-covering, if $\In$ (resp.\ $\Ex$) is a transversal of $\G$ (resp.\ $\H$), then, by \cref{lemma:summary_properties} point $(e)$, $\G(\sigma)=\emptyset$ (resp.\ $\H(\sigma)=\emptyset$).
In these cases, since at least one of the two hypergraphs $\G(\sigma)$ and $\H(\sigma)$ is empty, they satisfy the intersection property.

Let us consider now the case in which both $\In$ and $\Ex$ are not transversals of $\G(\sigma)$ and $\H(\sigma)$, respectively.
By \cref{lemma:summary_properties} point $(g)$, both $\G(\sigma)$ and $\H(\sigma)$ contain only non\nbdash-empty edges.
Assume by contradiction that there are two edges $G'\in\G(\sigma)$ and $H'\in\H(\sigma)$ such that $G'\cap H'=\emptyset$.
Since $G'\in\G(\sigma)$ there is an edge $G\in\G$ such that $G=G'\cup A$, with $\emptyset\subseteq A\subseteq\Ex$, and with $G\cap\In=\emptyset$ (for otherwise $G'$ would not be in $\G(\sigma)$).
Moreover, since $H'\in\H(\sigma)$ there is an edge $H\in\H$ such that $H=H'\cup B$, with $\emptyset\subseteq B\subseteq\In$, and with $H\cap\Ex=\emptyset$ (for otherwise $H'$ would not be in $\H(\sigma)$).
From this follows $G\cap H=\emptyset$: a contradiction, because $\G$ and $\H$ satisfy the intersection property.
Therefore, $\G(\sigma)$ and $\H(\sigma)$ satisfy the intersection property.

$(\Leftarrow)$
Since $\G$ and $\H$ do not satisfy the intersection property, there are two edges $G\in\G$ and $H\in\H$ such that $G \cap H = \emptyset$.
Assume \Wlog that $G$ and $H$ are not superset of other edges (remember that, in the statement of the lemma, $\G$ and $\H$ are not assumed to be simple).
Consider the empty assignment $\emptyassign$.
Observe that $G\in\G(\emptyassign)$ and $H\in\H(\sigma)$ and hence $\G(\emptyassign)$ and $\H(\emptyassign)$ do not satisfy the intersection property.
\end{proof}

\begin{lemma}\label{lemma:composition_property_transversals_expansion}
Let $\G$ and $\H$ be two hypergraphs, and let $\sigma=\assign{\In,\Ex}$ be an assignment.
\begin{enumerate}[label=$(\mathit{\alph*})$]
\item
If $T'$ is a transversal of $\G(\sigma)$, then $T=T'\cup\In$ is a transversal of $\G$.
\item
If $T'$ is an independent set of $\H(\sigma)$, then $T=T'\cup\In$ is an independent set of $\H$.
\end{enumerate}
Hence, if $T'$ is a new transversal of $\G(\sigma)$ \Wrt $\H(\sigma)$, then $T=T'\cup\In$ is a new transversal of $\G$ \Wrt $\H$.
\end{lemma}
\begin{proof}
We remind the reader that, since $\G(\sigma)$ (resp., $\H(\sigma)$) undergoes a minimization operation, it is not possible that $\emptyset \in \G(\sigma)$ and $\G(\sigma) \neq \{\emptyset\}$ (resp., $\emptyset \in \H(\sigma)$ and $\H(\sigma) \neq \{\emptyset\}$) at the same time.
Therefore, $\G(\sigma)$ (resp., $\H(\sigma)$) fulfills exactly one of these conditions: is empty, or contains only the empty edge, or contains only non\nbdash-empty edges.

\begin{enumerate}[label=$(\mathit{\alph*})$]
\item
If $\G(\sigma)=\{\emptyset\}$, then there are no transversals of $\G(\sigma)$ at all, and hence the property trivially holds because the antecedent is false.

If $\G(\sigma) = \emptyset$, then, by \cref{lemma:summary_properties} point $(e)$, $\In$ is a transversal of $\G$.
Hence, for any $T'$, $T = T' \cup \In$ is a transversal of $\G$ as well.

Assume that $\G(\sigma)$ contains only non\nbdash-empty edges.
By \cref{lemma:summary_properties} points $(a)$, $(b)$, $(c)$, and $(d)$, $\G$ contains only non\nbdash-empty edges.
Observe that, for each edge $G\in\G$, there are two cases: either $G\cap\In\neq\emptyset$ or $G\cap\In=\emptyset$.
If $G\cap\In\neq\emptyset$, then $T\cap G\neq\emptyset$ because $\In\subseteq T$.
If $G\cap\In=\emptyset$, then there is an edge $\emptyset\neq G'\in\G(\sigma)$ such that $G'\subseteq G$.
Hence, $T\cap G\neq\emptyset$ because $T'\cap G'\neq\emptyset$, $G'\subseteq G$, and $T'\subseteq T$.

\item
If $\H(\sigma)=\{\emptyset\}$, then there are no independent sets of $\H(\sigma)$ at all, and hence the property trivially holds because the antecedent is false.

If $\H(\sigma) = \emptyset$, then, by \cref{lemma:summary_properties} point $(e)$, $\Ex$ is a transversal of $\H$.
Therefore, $V\setminus\Ex$ is an independent set of $\H$.
Observe that the vertex set of $\H(\sigma)$ is $V\setminus(\In\cup\Ex)$, hence $T' \subseteq V\setminus(\In\cup\Ex)$.
Thus, any set $T = T' \cup \In$ is such that $T\subseteq V\setminus\Ex$, and, because $V\setminus\Ex$ is an independent set of $\H$, $T$ is an independent set of $\H$ as well.

Assume that $\H(\sigma)$ contains only non\nbdash-empty edges.
Since the vertex set of $\H(\sigma)$ is $V\setminus(\In\cup\Ex)$, $T' \subseteq V\setminus(\In\cup\Ex)$.
Observe that, for each edge $H\in\H$, there are two cases: either $H\cap\Ex\neq\emptyset$ or $H\cap\Ex=\emptyset$.
From $T'\subseteq V\setminus (\In\cup\Ex)$ and $\In\cap\Ex=\emptyset$, it follows that $T\cap\Ex=\emptyset$.
So, if $H\cap\Ex\neq\emptyset$, then $H\not\subseteq T$.
On the other hand, if $H\cap\Ex=\emptyset$, then there is an edge $\emptyset\neq H'\in\H(\sigma)$ such that $H'\subseteq H$.
Since $T'$ is an independent set of $\H(\sigma)$, there is a vertex $v\in(H'\setminus T')$.
From $H'\subseteq V\setminus(\In\cup\Ex)$ and $v\in H'$, it follows that $v\notin\In$.
Therefore, since $T=T'\cup\In$, $v\notin T$, and, because $v\in H'\subseteq H$, $H\not\subseteq T$.
\end{enumerate}
\end{proof}

Observe that, by the symmetrical nature of the \Dual problem, the roles of $\G$ and $\H$ can be swapped in an instance of \Dual.
By this reason, \cref{lemma:composition_property_transversals_expansion} can be easily adapted to state that transversals of $\H(\sigma)$ and independent sets of $\G(\sigma)$ can be extended, in this case by adding $\Ex$, to be transversals of $\H$ and independent sets of $\G$, respectively.

\begin{lemma}\label{lemma:composition_property_transversals_projection}
Let $\G$ and $\H$ be two hypergraphs.
\begin{enumerate}[label=$(\mathit{\alph*})$]
\item
If $T$ is a transversal of $\G$, then, for any assignment $\sigma=\assign{\In,\Ex}$ coherent with $T$, $T'=T\setminus\In$ is a transversal of $\G(\sigma)$.
\item
If $T$ is an independent set of $\H$, then, for any assignment $\sigma=\assign{\In,\Ex}$ coherent with $T$, $T'=T\setminus\In$ is an independent set of $\H(\sigma)$.
\end{enumerate}
Hence, if $T$ is a new transversal of $\G$ \Wrt $\H$, then, for any assignment $\sigma=\assign{\In,\Ex}$ coherent with $T$, $T'=T\setminus\In$ is a new transversal of $\G(\sigma)$ \Wrt $\H(\sigma)$.
\end{lemma}
\begin{proof}\hspace{0pt}
\begin{enumerate}[label=$(\mathit{\alph*})$]
\item
If $\emptyset \in \G$, then there are no transversals of $\G$ at all, and hence the property trivially holds because the antecedent is false.

If $\G = \emptyset$, then, by \cref{lemma:summary_properties} points $(c)$ and $(d)$, $\G(\sigma) = \emptyset$.
Since $\G(\sigma) = \emptyset$, any set of vertices $T'$ (even the empty one) is a transversal of $\G(\sigma)$, hence the property trivially holds because the consequent is true.

Assume that $\G$ contains only non\nbdash-empty edges.
There are two cases: either $\In$ is a transversal of $\G$, or it is not.
If $\In$ is a transversal of $\G$, then, by \cref{lemma:summary_properties} point $(e)$, $\G(\sigma)=\emptyset$.
Hence, trivially any set of vertices is a transversal of $\G(\sigma)$, and so is $T'$.
Consider now the case in which $\In$ is not a transversal of $\G$.
Since $T$ is a transversal of $\G$ coherent with $\sigma$, by (the contrapositive of) \cref{lemma:assignment_coherent_new_Tr_then_not_covering} point $(a)$, $\Ex$ is not covering.
Therefore, by \cref{lemma:summary_properties} point $(g)$, $\G(\sigma)$ contains only non\nbdash-empty edges.
Let $G'\in\G(\sigma)$ be an edge.
By definition of $\G(\sigma)$, there is an edge $G\in\G$ such that $G\cap\In=\emptyset$ (for otherwise $G'$ would not be in $\G(\sigma)$) and $G'=G\cap(V\setminus(\In\cup\Ex))$.
Since $T$ is a transversal of $\G$, $T\cap (G\setminus\In)\neq\emptyset$, because $G\cap\In=\emptyset$, and $T\cap (G\setminus\Ex)\neq\emptyset$, because $\sigma\dblsubeq T$ and so $T\cap\Ex=\emptyset$.
Therefore, $T\cap (G\setminus(\In\cup\Ex))\neq\emptyset$.
Since $G'=G\cap(V\setminus(\In\cup\Ex))$, $T'=T\setminus\In$ has a non\nbdash-empty intersection with $G'$, and hence $T'$ is a transversal of $\G(\sigma)$.

\item
If $\emptyset \in \H$, then there are no independent sets of $\H$ at all, and hence the property trivially holds because the antecedent is false.

If $\H = \emptyset$, then, by \cref{lemma:summary_properties} points $(c)$ and $(d)$, $\H(\sigma) = \emptyset$.
Since $\H(\sigma) = \emptyset$, any set of vertices $T'$ (even the empty one) is an independent set of $\H(\sigma)$, hence the property trivially holds because the consequent is true.

Assume that $\H$ contains only non\nbdash-empty edges.
There are two cases: either $\Ex$ is a transversal of $\H$, or it is not.
If $\Ex$ is a transversal of $\H$, then, by \cref{lemma:summary_properties} point $(e)$, $\H(\sigma)=\emptyset$.
Hence, trivially any set of vertices is an independent set of $\H(\sigma)$, and so is $T'$.
Consider now the case in which $\Ex$ is not a transversal of $\H$.
Since $T$ is an independent set of $\H$ coherent with $\sigma$, by (the contrapositive of) \cref{lemma:assignment_coherent_new_Tr_then_not_covering} point $(b)$, $\In$ is not covering.
Therefore, by \cref{lemma:summary_properties} point $(g)$, $\H(\sigma)$ contains only non\nbdash-empty edges.
Let $H'\in\H(\sigma)$ be an edge.
By definition of $\H(\sigma)$, there is an edge $H\in\H$ such that $H\cap\Ex=\emptyset$ (for otherwise $H'$ would not be in $\H(\sigma)$) and $H'=H\cap(V\setminus(\In\cup\Ex))$.
Since $T$ is an independent set of $\H$, there is a vertex $v\in (H\setminus T)$ such that $v\notin\Ex$, because $H\cap\Ex=\emptyset$, and $v\notin\In$, because $\sigma\dblsubeq T$ and so $\In\subseteq T$.
Therefore, from $v\notin\Ex$ and $v\notin\In$, it follows that $v\in H'$ because $H'=H\cap(V\setminus(\In\cup\Ex))$, and from $v\notin T$ (because $v\in (H\setminus T)$), it follows that $v\notin T'$.
Hence $H'\not\subseteq T'$, and thus $T'$ is an independent set of $\H(\sigma)$.
\end{enumerate}
\end{proof}

Observe again that, by the symmetrical nature of the \Dual problem, by swapping the roles of $\G$ and $\H$, and by considering the reversed assignment $\compl{\sigma}=\assign{\Ex,\In}$, \cref{lemma:composition_property_transversals_projection} can be easily adapted to state that transversals of $\H$ and independent sets of $\G$ coherent with $\compl{\sigma}$ can be shrunk, in this case by removing $\Ex$, to be transversals of $\H(\sigma)$ and independent sets of $\G(\sigma)$, respectively.

The following corollary descends directly from \cref{lemma:composition_property_transversals_expansion,lemma:composition_property_transversals_projection} and highlights an interesting property of decompositions.
\begin{corollary}
Let $\G$ and $\H$ be two hypergraphs.
Then, for all assignments $\sigma$, there is a new transversal of $\G(\sigma)$ \Wrt $\H(\sigma)$ if and only if there is a new transversal of $\G$ \Wrt $\H$ coherent with $\sigma$.
\end{corollary}

From the previous corollary, by noticing that the empty assignment is coherent with any set of vertices, we obtain the next property.

\begin{corollary}\label{corol:no_new_tr_iff_no_new_tr_in_subinstances}
Let $\G$ and $\H$ be two hypergraphs.
There is no new transversal of $\G$ \Wrt $\H$ if and only if, for all assignments $\sigma$, there is no new transversal of $\G(\sigma)$ \Wrt $\H(\sigma)$.
\end{corollary}

\begin{lemma}\label{lemma:if_dual_then_all_subinstaces_dual}
Let $\G$ and $\H$ be two dual hypergraphs.
Then, for all assignments $\sigma$, $\G(\sigma)$ and $\H(\sigma)$ are dual.
\end{lemma}
\begin{proof}
If $\G$ and $\H$ are trivially dual, then, by \cref{lemma:summary_properties} points $(a)$, $(b)$, $(c)$, and $(d)$, for any assignment $\sigma$, $\G(\sigma)$ and $\H(\sigma)$ are trivially dual.

Consider now the case in which $\G$ and $\H$ are non\nbdash-trivially dual.
Since $\G$ and $\H$ are dual, they are simple, satisfy the intersection property, and there is no new transversal of $\G$ \Wrt $\H$ (see \cref{lemma:no_new_transversal_iff_dual}).
Observe that, for any assignment $\sigma$, $\G(\sigma)$ and $\H(\sigma)$ are simple by definition and, by \cref{lemma:proiezione_hitting_property}, satisfy the intersection property.
Because there is no new transversal of $\G$ \Wrt $\H$, by \cref{corol:no_new_tr_iff_no_new_tr_in_subinstances}, there is no new transversal of $\G(\sigma)$ \Wrt $\H(\sigma)$.
Therefore, from \cref{lemma:no_new_transversal_iff_dual} it follows that $\G(\sigma)$ and $\H(\sigma)$ are dual.
\end{proof}

\begin{lemma}\label{lemma:if_all_subinstaces_dual_then_dual}
Let $\G$ and $\H$ be two simple hypergraphs.
If, for all assignments $\sigma$, $\G(\sigma)$ and $\H(\sigma)$ are dual, then $\G$ and $\H$ are dual.
\end{lemma}
\begin{proof}
Since, for all assignments $\sigma$, $\G(\sigma)$ and $\H(\sigma)$ are dual, by \cref{lemma:no_new_transversal_iff_dual}, for all assignments $\sigma$, $\G(\sigma)$ and $\H(\sigma)$ satisfy the intersection property and there is no new transversal of $\G(\sigma)$ \Wrt $\H(\sigma)$.
Hence, by \cref{lemma:proiezione_hitting_property}, $\G$ and $\H$ satisfy the intersection property.
Because, for all assignments $\sigma$, there is no new transversal of $\G(\sigma)$ \Wrt $\H(\sigma)$, by \cref{corol:no_new_tr_iff_no_new_tr_in_subinstances}, there is no new transversal of $\G$ \Wrt $\H$.
Given that $\G$ and $\H$ are assumed to be simple, by \cref{lemma:no_new_transversal_iff_dual} it follows that $\G$ and $\H$ are dual.
\end{proof}

It is interesting to observe that in the statement of the previous lemma it is necessary to assume that hypergraphs $\G$ and $\H$ are simple.
Indeed, for any assignment $\sigma$, $\G(\sigma)$ and $\H(\sigma)$ are simple by definition, because they undergo a minimization operation.
Hence, from $\G(\sigma)$ and $\H(\sigma)$ being dual (and hence also simple) we cannot derive $\G$ and $\H$ being simple.
To give an example, consider hypergraphs $\G = \emptyset$ and $\H$ being a hypergraph containing an empty edge and an other edge.
Hypergraph $\H$ is not simple, and hence $\G$ and $\H$ are not dual.
However, for any assignment $\sigma$, $\G(\sigma) = \emptyset$ and $\H(\sigma) = \{\emptyset\}$ (see \cref{lemma:summary_properties} points $(a)$, $(b)$, $(c)$, and $(d)$), and hence $\G(\sigma)$ and $\H(\sigma)$ are dual.

Given the properties above, the following lemma is a direct consequence of \cref{lemma:no_new_transversal_iff_dual,lemma:if_dual_then_all_subinstaces_dual,lemma:if_all_subinstaces_dual_then_dual}.

\medskip

\noindent\textbf{\Cref*{lemma:dual_iff_all_sub-instances_dual}.}
\textit{%
Two hypergraphs $\G$ and $\H$ are dual if and only if $\G$ and $\H$ are simple, satisfy the intersection property, and, for all assignments $\sigma$, $\G(\sigma)$ and $\H(\sigma)$ are dual (or, equivalently, there is no new transversal of $\G(\sigma)$ \Wrt $\H(\sigma)$).}

\section{A deterministic algorithm for \texorpdfstring{\DualProbHyp}{DUAL}}\label{sec:alg-gaur-det}
In this section we propose a deterministic duality algorithm \textsc{\Alg}, which is an extension of that proposed by \citet{Gaur1999} (see also \citet{Gaur2004}).
Algorithm \textsc{\Alg} here presented is in somewhat different from \citeauthor{Gaur1999}'s because the latter checks \emph{self}\nbdash-duality of a single DNF Boolean formula, while ours verifies duality between two hypergraphs.

Given two hypergraphs $\G$ and $\H$, our algorithm \textsc{\Alg}, like many others, to disprove that $\G$ and $\H$ are dual aims at finding, via sub\nbdash-procedure \textsc{\DetNewTrAlg}, a new transversal of $\G$ \Wrt $\H$.
To do so, the algorithm builds up, step after step, by including and excluding vertices, a set of vertices intersecting all edges of $\G$ that is different from all edges of $\H$ (i.e., a set of vertices that is a transversal of $\G$ and an independent set of $\H$, and hence a new transversal of $\G$ \Wrt $\H$).

As already discussed, choosing vertices to exclude allows us to decrease the number of edges of $\H$ that are not different (yet) from the candidate for a new transversal.
In particular, when \textsc{\DetNewTrAlg} excludes specific vertices, the number of edges of $\H$ still needed to be considered is halved (see \cref{lemma:riduzione_taglia_assignment_tree} and the discussion in \cref{sec:logarithmic_refuters}).

Let us now see the details of algorithm \textsc{\Alg}.
After exhibiting the algorithm, we will formally prove some of its properties and its correctness.
Algorithm \textsc{\Alg}, and more specifically sub\nbdash-procedure \textsc{\DetNewTrAlg} which checks the existence of new transversals, uses three sets to keep track of the included, excluded, and free vertices of the currently considered assignment, which are denoted by $\IncAlg$, $\ExcAlg$, and $\FreeAlg$, respectively.
We remind the reader that, to recognize an assignment as a non\nbdash-duality witness, we need to know what edges of $\G$ are separated from and what edges of $\H$ are still compatible with the assignment (and transversal) under construction.
To this purpose, in the algorithm we use the sets denoted by $\SepAlg$, and $\ComAlg$, respectively.
Observe that, for simplicity of the presentation and better readability of the algorithm, we here explicitly store sets $\IncAlg$, $\ExcAlg$, $\FreeAlg$, $\SepAlg$, and $\ComAlg$. This obviously requires more than quadratic logspace.
However, by storing instead the labels only of the extensions leading to the current assignment, and evaluating the aforementioned sets dynamically (see the proof of \cref{lemma:algoritmo_quadratic_logspace} and \cref{sec:logic}), the space complexity of the algorithm is bounded by \DSpace[$\log^2 \InputSize$].

The \textsc{\Alg} algorithm is listed below as \cref{alg:dual_guar_params_copia}.
The aim of the procedure \textsc{\CheckSimpleIP} is checking that hypergraphs $\G$ and $\H$ are simple and satisfy the intersection property.

\begin{algorithm}[!ht]
\caption{A deterministic duality algorithm based on \citeauthor{Gaur1999}'s.}\label{alg:dual_guar_params_copia}
\begin{algorithmic}[1]
\Procedure{\Alg}{$\G,\H$}
\If{$\lnot$\Call{\CheckSimpleIP}{$\G,\H$}}\label{line:det-alg_checkSimpleHP}
\textbf{return} \valfalse;
\EndIf
\State \textbf{return} $\lnot$\Call{\DetNewTrAlg}{$\G,\H,\emptyset,\emptyset, V$};\label{line:det-alg_callNewTr}
\EndProcedure
\Statex
\Procedure{\DetNewTrAlg}{$\G,\H,\IncAlg,\ExcAlg,\FreeAlg$}
    \If{$(\exists G)(G\in\G\land G\subseteq\ExcAlg)\lor(\exists H)(H\in\H\land H\subseteq\IncAlg)$}\label{line:test-defCom}\label{line:test-defSep}\label{line:inizio-test-defComSep}\label{line:inizio-test-completo-witness}
    \textbf{return} \valfalse;\label{line:ret-false-defCom}\label{line:ret-false-defSep}
    \EndIf\label{line:fine-test-defComSep}
    \State $\SepAlg\gets{}\{G\in\G\mid G\cap\IncAlg=\emptyset\}$;
    \State $\ComAlg\gets{}\{H\in\H\mid H\cap\ExcAlg=\emptyset\}$;
    \If{$\SepAlg=\emptyset \lor \ComAlg=\emptyset$}\label{line:inizio-check-witness}
        \textbf{return} \valtrue;\label{line:ret-true-witness}
    \EndIf\label{line:fine-test-completo-witness}\label{line:fine-check-witness}
    \State $U\gets \{v\in\FreeAlg \mid v\text{ belongs to at least half of the edges of }\ComAlg\}$;\label{line:inizio-calcolo-U}\label{line:fine-calcolo-U}\label{line:calcolo-U}
    \For{each $v : v\in U$}\label{line:inizio-test-firstType}
        \If{\Call{\DetNewTrAlg}{$\G,\H,\IncAlg,\ExcAlg\cup\{v\},\FreeAlg\setminus\{v\}$}}\label{line:call-ric-firstType}
            \textbf{return} \valtrue;
        \EndIf
    \EndFor\label{line:fine-test-firstType}
    \State $\FreeAlg\gets\FreeAlg\setminus U$;\label{line:inizio-inclusione-U}
    \State $\IncAlg\gets\IncAlg\cup U$;\label{line:inclusione-U}\label{line:fine-inclusione-U}
    \If{$(\exists H)(H\in\H\land H\subseteq\IncAlg)$}\label{line:test-defCom-dopoU}
        \textbf{return} \valfalse;\label{line:ret-false-defCom-dopoU}
    \EndIf
    \State $\SepAlg\gets{}\{G\in\G\mid G\cap\IncAlg=\emptyset\}$;
    \If{$\SepAlg=\emptyset$}\label{line:inizio-secondo-test-witness}
        \textbf{return} \valtrue;\label{line:ret-true-witness-dopoU}
    \EndIf\label{line:fine-secondo-test-witness}
    \For{each $v : v\in \FreeAlg$}\label{line:inizio-test-thirdType}
        \For{each $G : v\in G \land G\in \SepAlg$}
            \If{\Call{\DetNewTrAlg}{$\G,\H,\IncAlg\cup\{v\},\ExcAlg\cup (G\setminus\{v\}),\FreeAlg\setminus G$}}\label{line:call-ric-thirdType}
                \textbf{return} \valtrue;
            \EndIf
        \EndFor
    \EndFor\label{line:fine-test-thirdType}
    \State\textbf{return} \valfalse;\label{line:ret-false-esaurimento}
\EndProcedure
\end{algorithmic}
\end{algorithm}

Note that, for simplicity, procedure \textsc{\DetNewTrAlg} is meant to be called by value.
This means that the parameters passed to the procedure are local copies for each specific recursive call.
Therefore, any modification to those sets affects only the sets of the call currently executed.

We recall here that two hypergraphs $\G$ and $\H$ are dual if and only if they are simple, satisfy the intersection property, and there is no new transversal of $\G$ \Wrt $\H$ (see \cref{lemma:no_new_transversal_iff_dual}).
So, after checking that $\G$ and $\H$ are simple and satisfy the intersection property (line~\ref*{line:det-alg_checkSimpleHP}), it is checked that there is no new transversal of $\G$.
This is achieved by calling \textsc{\DetNewTrAlg}$(\G,\H,\emptyset,\emptyset,V)$ at line~\ref*{line:det-alg_callNewTr},
where the sets $\IncAlg$ and $\ExcAlg$ are initialized to $\emptyset$, and $\FreeAlg=V$.
This procedure is devised to answer \valtrue if and only if it finds a new transversal of $\G$ coherent with the assignment on which it is executed.
Given a pair of hypergraphs $\tuple{\G,\H}$, and an assignment $\pi=\assign{\In,\Ex}$, we say that procedure \textsc{\DetNewTrAlg} is executed on $\pi$ whenever it is called with the following parameters: \textsc{\DetNewTrAlg}$(\G,\H,\In,\Ex,V\setminus(\In\cup \Ex))$.
For notational convenience we denote it by \textsc{\DetNewTrAlg}$(\G,\H,\pi)$, or even more simply by \textsc{\DetNewTrAlg}$(\pi)$ when it is clear from the context what the two hypergraphs $\G$ and $\H$ are.

\medskip

Let us now analyze intuitively the execution of procedure \textsc{\DetNewTrAlg}.
Consider a generic call of the procedure executed on the pair of hypergraphs $\tuple{\G,\H}$ and on assignment $\pi$.
At line \ref*{line:inizio-test-defComSep} we check whether $\pi$ is a covering assignment, in which case clearly there is no new transversal of $\G$ coherent with $\pi$ (see \cref{lemma:assignment_coherent_new_Tr_then_not_covering}).
If this is not the case, then the procedure checks (at line \ref*{line:inizio-check-witness}) whether $\pi$ is (already) a witness.
Then the procedure computes a set $U$ (line~\ref*{line:calcolo-U}).
This set is locally computed in the call currently executed, and hence, for the following discussion, let us call it $U_\pi$.
We use the subscript $\pi$ because set $U$ depends on the history of the recursive calls having led to the one currently being executed, and hence depends on the currently considered assignment $\pi$ having been built so far (and encoded in sets $\IncAlg$ and $\ExcAlg$).

Set $U_\pi$ is the set of the free frequent vertices of $\pi$.
First, the procedure tries to exclude individually each of the vertices of $U_\pi$ (lines \ref*{line:inizio-test-firstType}\nbdash--\ref*{line:fine-test-firstType}).
If none of these attempts results in the construction of a witness (all the tests performed at line~\ref*{line:call-ric-firstType} return \valfalse), all vertices of $U_\pi$ are included (lines \ref*{line:inizio-inclusione-U}\nbdash--\ref*{line:fine-inclusione-U}).
Let us call $\sigma$ the assignment resulting after the inclusion of the vertices of $U_\pi$.
Then, the procedure checks again whether $\sigma$ is a covering assignment (line~\ref*{line:test-defCom-dopoU}).
Otherwise, the procedure tests whether $\sigma$ is a witness (line~\ref*{line:inizio-secondo-test-witness}).
If this is not the case, then the procedure tries to include each of the free vertices of $\sigma$ as a critical vertex with an edge from $\Sep(\sigma)$ witnessing its criticality (lines \ref*{line:inizio-test-thirdType}\nbdash--\ref*{line:fine-test-thirdType}).
If for none of these attempts it is possible to find a new transversal of $\G$ (all the tests performed at line~\ref*{line:call-ric-thirdType} return \valfalse), then the procedure answers \valfalse at line~\ref*{line:ret-false-esaurimento}, meaning that there is no new transversal of $\G$ coherent with $\pi$.

\medskip

Let us now make some observations on the procedure \textsc{\DetNewTrAlg}.
Throughout the whole execution of the procedure and its recursive calls, the sets $\IncAlg$, $\ExcAlg$, and $\FreeAlg$, are always a partition of the set of vertices $V$, and hence $\IncAlg$ and $\ExcAlg$ constitute a consistent assignment (i.e., $\IncAlg$ and $\ExcAlg$ do not overlap).
This descends from the fact that the extensions implemented in procedure \textsc{\DetNewTrAlg} are the same of those considered in tree $\Tree(\G,\H)$, and we have already discussed in \cref{sec:decomposition} that those extensions generate consistent assignments.

Besides this, every recursive call performed by \textsc{\DetNewTrAlg}$(\pi)$ is executed on an assignment $\pi'$ such that $\pi\dblsub\pi'$.
This implies that the set of free vertices becomes smaller and smaller from one recursive call to the next.
As a result, since the procedure tries to include or exclude vertices taken only from the set of the free ones, every recursion path traversed by \textsc{\DetNewTrAlg}$(\pi)$ is finite, because at some recursion level the set of free vertices is empty.

The tests performed at lines \ref*{line:inizio-test-completo-witness}\nbdash--\ref*{line:fine-test-completo-witness} and at lines \ref*{line:inizio-secondo-test-witness}\nbdash--\ref*{line:fine-secondo-test-witness} essentially check whether the assignment is a witness.
In particular, first it is ruled out that the current assignment is covering, and then is checked whether $\SepAlg$ or $\ComAlg$ is a transversal of $\G$ or $\H$, respectively.
Note that at lines \ref*{line:inizio-secondo-test-witness}\nbdash--\ref*{line:fine-secondo-test-witness} is tested only whether $\SepAlg$ is a new transversal of $\G$ \Wrt $\H$.
At that point of the execution is reasonable to do so because only $\SepAlg$ has changed after the previous check of the witnessing condition performed at lines \ref*{line:inizio-test-completo-witness}\nbdash--\ref*{line:fine-test-completo-witness}.

We can now formally prove the correctness of algorithm \textsc{\Alg}.
\begin{theorem}\label{theo:algoritmo_corretto}
Let $\G$ and $\H$ be two hypergraphs.
Then, the call \textsc{\Alg}$(\G,\H)$ outputs \valtrue if and only if $\G$ and $\H$ are dual.
\end{theorem}

To prove \cref{theo:algoritmo_corretto} we need some intermediate lemmas.
The following property is at the base of the decomposition used in Algorithm ``A'' of \citet{fred-khac-96}.
We state it in a form appropriate for our own subsequent discussion.

\begin{lemma}\label{lemma:transversal_include_o_esclude_vertice}
Let $\G$ and $\H$ be two hypergraphs, let $\sigma$ be an assignment, and let $v$ be a free vertex in $\sigma$.
Then, there is a new transversal of $\G$ \Wrt $\H$ coherent with $\sigma$ if and only if $\sigma\addassign\assign{\{v\},\emptyset}$ or $\sigma\addassign\assign{\emptyset,\{v\}}$ is coherent with a new transversal of $\G$ \Wrt $\H$.\footnote{This lemma essentially states that if there are new transversals of $\G$ \Wrt $\H$ coherent with $\sigma$, then they include or exclude the free vertex $v$. (If there is no new transversal either including or excluding $v$, then there is no new transversal coherent with $\sigma$ at all.)}
\end{lemma}
\begin{proof}
Let $T$ be a new transversal of $\G$ \Wrt $\H$ coherent with $\sigma$.
If $v\in T$, then $T$ is coherent with $\sigma\addassign\assign{\{v\},\emptyset}$, symmetrically if $v\notin T$, then $\sigma\addassign\assign{\emptyset,\{v\}}$ is coherent with $T$.
The other direction of the proof is obvious.
\end{proof}

On the other hand, the following property was used by \citeauthor{Gaur1999} in his algorithm~\citep{Gaur1999,Gaur2004}.
Again, we state it in a form appropriate for our own discussion.
(Note the difference with the statement of \cref{lemma:aggiungere_un_nodo_come_critico}.
In this case, the assignment $\sigma$ mentioned in the statement of the following lemma is coherent with a new transversal that is \emph{not} required to be minimal.)
\begin{lemma}\label{lemma:intersezione_minima}
Let $\G$ and $\H$ be two hypergraphs, and let $\sigma$ be a non\nbdash-witnessing assignment coherent with a new (not necessarily minimal) transversal $T$ of $\G$ \Wrt $\H$.
Then, there is an edge $G\in\Sep(\sigma)$ and a free vertex $v\in G$, such that $\sigma\addassign\assign{\{v\},G\setminus\{v\}}$ is coherent with a new transversal $T'$ of $\G$ \Wrt $\H$, and $T'\subseteq T$.
\end{lemma}
\begin{proof}
Let $T$ be a new transversal of $\G$ \Wrt $\H$, and let $\sigma=\assign{\In,\Ex}$ be a non\nbdash-witnessing assignment coherent with $T$.
Consider an edge $\wh{G}\in\arg\min_{G\in\Sep(\sigma)}\{|G\cap T|\}$, and let $v$ be any free vertex belonging to $\wh{G}\cap T$.
We will show that such an edge $\wh{G}$ and such a vertex $v$ are well defined.

Indeed, since $\sigma$ is a non\nbdash-witnessing assignment, $\Sep(\sigma)\neq\emptyset$ (for otherwise $\sigma$ would be a witness), and hence there is an edge $\wh{G}$ belonging to $\arg\min_{G\in\Sep(\sigma)}\{|G\cap T|\}$.
Moreover, by $\wh{G}\in\Sep(\sigma)$, $\wh{G}\cap\In = \emptyset$ and, because $\sigma$ is coherent with $T$, $\wh{G}\not\subseteq\Ex$, therefore there is a free vertex of $\sigma$ belonging to $\wh{G}$.
Observe that, by the fact that $\wh{G}\cap \In=\emptyset$, it follows that $\sigma\addassign\assign{\{v\},\wh{G}\setminus\{v\}}$ is actually a consistent assignment (i.e., the sets of included and excluded vertices do not overlap).

Now, let $T'=T\setminus (\wh{G}\setminus\{v\})$ (see \cref{fig:intersezione_minima}).
By definition, $T'$ is coherent with $\sigma\addassign\assign{\{v\},\wh{G}\setminus\{v\}}$.
We now claim that $T'$ is a transversal of $\G$.
Assume by contradiction that this is not the case.
Then, there is an edge $\wt{G}\in\G$ such that $T'\cap\wt{G}=\emptyset$ (see \cref{fig:intersezione_minima}).
\begin{figure}[!ht]
  \centering%
  \includegraphics[width=0.3\textwidth]{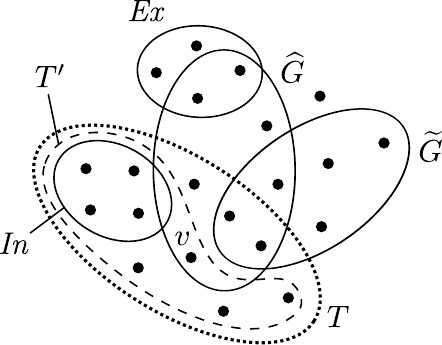}%
  \caption{Illustration for \cref{lemma:intersezione_minima}\label{fig:intersezione_minima}.}%
\end{figure}
From $\In\subset T$ (because $\sigma$ is coherent with $T$) and $\wh{G}\cap\In=\emptyset$ follows $\In\subset T'$ (because $T'=T\setminus (\wh{G}\setminus\{v\})$).
This, with $T'\cap\wt{G}=\emptyset$, implies that $\wt{G}\in\Sep(\sigma)$.
Since $T\cap\wt{G}\neq\emptyset$ (because $T$ is a transversal of $\G$) and $T'\cap\wt{G}=\emptyset$, by the definition of $T'$ it follows that $(\wt{G}\cap T)\subseteq ((\wh{G}\setminus\{v\})\cap T)$.
For this reason, $|\wt{G}\cap T|\leq|(\wh{G}\setminus\{v\})\cap T|<|\wh{G}\cap T|$ (because $v\in \wh{G}\cap T$): a contradiction, because $\wh{G}$ was chosen as one of the edges in $\Sep(\sigma)$ minimizing the size of its intersection with $T$.
Therefore, $T'$ is a transversal of $\G$.
To conclude, observe that $T'\subseteq T$ by definition of $T'$, and hence $T'$ is an independent set of $\H$ because $T$, being a new transversal of $\G$ \Wrt $\H$, is an independent set of $\H$.
Thus, $T'$ is a new transversal of $\G$ \Wrt $\H$.
\end{proof}

We now focus on the properties of procedure \textsc{\DetNewTrAlg}.
\begin{lemma}\label{lemma:esecuzione_singola_corretta}
Let $\G$ and $\H$ be two hypergraphs, and let $\pi$ be an assignment coherent with a new transversal of $\G$ \Wrt $\H$.
Then, either \textsc{\DetNewTrAlg}$(\pi)$ answers \valtrue at line~\ref*{line:ret-true-witness} or at line~\ref*{line:ret-true-witness-dopoU}, or among its recursive calls there is one executed on an assignment $\pi'$, with $\pi\dblsub\pi'$, coherent with a new transversal of $\G$ \Wrt $\H$.
\end{lemma}

\begin{proof}
Let $T$ be a new transversal of $\G$ \Wrt $\H$ coherent with $\pi$.
Since $\pi\dblsubeq T$, and $T$ is a new transversal of $\G$, $\pi$ cannot be a covering assignment (see \cref{lemma:assignment_coherent_new_Tr_then_not_covering}), and hence \textsc{\DetNewTrAlg}$(\pi)$ cannot return \valfalse at line~\ref*{line:ret-false-defCom}.
If $\pi$ is already a witness, then \textsc{\DetNewTrAlg}$(\pi)$ answers \valtrue at line~\ref*{line:ret-true-witness} (and the statement of the lemma would be proven).

If this is not the case, then $\pi$ is a non\nbdash-witnessing assignment coherent with $T$.
Let $U_\pi$ be the set $U$ computed by \textsc{\DetNewTrAlg}$(\pi)$ at line~\ref*{line:calcolo-U}.
First, let us assume that $U_\pi\neq\emptyset$.
At lines \ref*{line:inizio-test-firstType}\nbdash--\ref*{line:fine-test-firstType} recursive calls are performed on the various assignments $\pi'=\pi\addassign\assign{\emptyset,\{v\}}$, for each vertex $v\in U_\pi$ (note that $\pi\dblsub\pi'$).
If one of them is coherent with a new transversal of $\G$, then the statement of the lemma is proven.
If this is not the case, then, by \cref{lemma:transversal_include_o_esclude_vertice}, assignment $\sigma=\pi\addassign\assign{U_\pi,\emptyset}$ is coherent with $T$.

Again, since $\sigma\dblsubeq T$, $\sigma$ cannot be a covering assignment (see \cref{lemma:assignment_coherent_new_Tr_then_not_covering}), and hence the call \textsc{\DetNewTrAlg}$(\pi)$ cannot answer \valfalse at line~\ref*{line:ret-false-defCom-dopoU}.
If $\sigma$ is a witness, then \textsc{\DetNewTrAlg}$(\pi)$ returns \valtrue at line~\ref*{line:ret-true-witness-dopoU} (and the statement of the lemma would be proven).

If this is not the case, then also $\sigma$ is a non\nbdash-witnessing assignment coherent with $T$.
By \cref{lemma:intersezione_minima}, there is an edge $G\in\Sep(\sigma)$ and a free vertex $v\in G$ of $\sigma$ such that $\pi'=\sigma\addassign\assign{\{v\},G\setminus\{v\}}$ is coherent with a new transversal of $\G$ (again, with $\pi\dblsub\pi'$).
Note that such an assignment belongs exactly to those on which a recursive call is performed at lines \ref*{line:inizio-test-thirdType}\nbdash--\ref*{line:fine-test-thirdType}.
Hence the statement is proven.

To conclude, let us consider the case in which $U_\pi=\emptyset$.
In this case, \textsc{\DetNewTrAlg}$(\pi)$ does not execute the loop at lines \ref*{line:inizio-test-firstType}\nbdash--\ref*{line:fine-test-firstType}, and lines \ref*{line:inizio-inclusione-U}\nbdash--\ref*{line:fine-inclusione-U} do not alter in any way sets $\IncAlg$ and $\FreeAlg$.
This means that, after line~\ref*{line:fine-inclusione-U}, sets $\IncAlg$ and $\FreeAlg$ still reflect the original assignment $\pi$ on which the procedure was called.
Said so, the discussion is the very same as above, since \cref{lemma:intersezione_minima} guarantees that at least one of the recursive call performed at lines \ref*{line:inizio-test-thirdType}\nbdash--\ref*{line:fine-test-thirdType} is performed on an assignment coherent with a new transversal of $\G$.
\end{proof}

\begin{lemma}\label{lemma:esecuzione_generale_corretta}
Let $\G$ and $\H$ be two hypergraphs, and let $\pi$ be an assignment.
Then, the call \textsc{\DetNewTrAlg}$(\pi)$ answers \valtrue if and only if there is a new transversal of $\G$ \Wrt $\H$ coherent with $\pi$.
\end{lemma}
\begin{proof}
$(\Rightarrow)$
Let $\pi=\assign{\In,\Ex}$, and let us assume that \textsc{\DetNewTrAlg}$(\pi)$ answers \valtrue.
Therefore, \textsc{\DetNewTrAlg}$(\pi)$ itself or one of the recursive call spawned directly or indirectly by \textsc{\DetNewTrAlg}$(\pi)$ answers \valtrue either at line~\ref*{line:ret-true-witness} or at line~\ref*{line:ret-true-witness-dopoU} (see \cref{lemma:esecuzione_singola_corretta}).
We have already seen that those parts of the algorithm check whether the built assignment is a witness.
Moreover, observe that in the algorithm vertices of $\In$ are never removed from set $\IncAlg$, and vertices of $\Ex$ are never removed from set $\ExcAlg$.
Hence, if the algorithm answers \valtrue it is because it has actually built a witness $\sigma$ such that $\pi\dblsubeq\sigma$.
Hence, there is a new transversal of $\G$ \Wrt $\H$ coherent with $\pi$.

$(\Leftarrow)$
Assume now that there is a new transversal of $\G$ \Wrt $\H$ coherent with $\pi$.
By \cref{lemma:esecuzione_singola_corretta} we know that either \textsc{\DetNewTrAlg}$(\pi)$ answers \valtrue at line~\ref*{line:ret-true-witness} or at line~\ref*{line:ret-true-witness-dopoU} (and hence this direction of the proof would be proven), or among its recursive calls there is one executed on an assignment coherent with a new transversal of $\G$.

Let us consider the latter case, and let $\hat{\pi}^1$ be the first assignment coherent with a new transversal of $\G$ on which \textsc{\DetNewTrAlg}$(\pi)$ executes a recursive call.
Since $\hat{\pi}^1$ is coherent with a new transversal of $\G$, \cref{lemma:esecuzione_singola_corretta} applies to \textsc{\DetNewTrAlg}$(\hat{\pi}^1)$ as well.

By the recursive application of \cref{lemma:esecuzione_singola_corretta}, we conclude that among all the possible sequences of recursive calls rooted in \textsc{\DetNewTrAlg}$(\pi)$, there are sequences of calls successively executed on assignments coherent with a new transversal of $\G$.
Let us pose, for notational convenience, $\pi=\hat{\pi}^0$.

Let $\hat{p}=(\hat{\pi}^0,\hat{\pi}^1,\dots,\hat{\pi}^k)$ be the ``foremost'' and ``longest'' of those sequences such that, for all $1\leq i\leq k$, assignment $\hat{\pi}^i$ is coherent with a new transversal of $\G$ and is one of the assignments on which \textsc{\DetNewTrAlg}$(\hat{\pi}^{i-1})$ executes a recursive call.
With ``foremost'' we mean that, for every $1\leq i\leq k$, assignment $\hat{\pi}^i$ is the \emph{first} coherent with a new transversal of $\G$ on which \textsc{\DetNewTrAlg}$(\hat{\pi}^{i-1})$ executes a recursive call.
With ``longest'' we mean that \textsc{\DetNewTrAlg}$(\hat{\pi}^k)$ does not perform further recursive calls on assignments coherent with new transversals of $\G$.
Observe that such a sequence is finite because every recursive path traversed by the algorithm is finite (we have already discussed this).

Since $\hat{\pi}^k$ is coherent with a new transversal of $\G$ and $\hat{p}$ is the foremost and longest sequence of calls successively executed on assignments coherent with a new transversal of $\G$, from \cref{lemma:esecuzione_singola_corretta} it follows that \textsc{\DetNewTrAlg}$(\hat{\pi}^k)$ returns \valtrue either at line~\ref*{line:ret-true-witness} or at line~\ref*{line:ret-true-witness-dopoU}.
This positive answer propagates throughout all the recursion path and thus \textsc{\DetNewTrAlg}$(\pi)$ answers \valtrue{} as well.
\end{proof}

\begin{proof}[Proof of \Cref*{theo:algoritmo_corretto}]
By \cref{lemma:no_new_transversal_iff_dual}, $\G$ and $\H$ are dual if and only if they are simple, satisfy the intersection property, and there is no new transversal of $\G$ \Wrt $\H$.
Therefore, since $\G$ and $\H$ are explicitly checked to be simple hypergraphs satisfying the intersection property (line~\ref*{line:det-alg_checkSimpleHP}), the correctness of the algorithm \textsc{\Alg} directly follows from the correctness of the procedure \textsc{\DetNewTrAlg} (see \cref{lemma:esecuzione_generale_corretta}): in fact, any new transversal of $\G$, if exists, is coherent with $\emptyassign$ which is the assignment on which \textsc{\DetNewTrAlg} is invoked by the procedure \textsc{\Alg} (line~\ref*{line:det-alg_callNewTr}).
\end{proof}

Let us now consider the time complexity of the algorithm.
\begin{lemma}\label{lemma:dimezzamento_ricorsivo}
Let $\G$ and $\H$ be two hypergraphs satisfying the intersection property, and let $\pi$ be an assignment.
If \textsc{\DetNewTrAlg}$(\pi')$ is any of the recursive calls performed by \textsc{\DetNewTrAlg}$(\pi)$, then $|\Com(\pi')|\leq\frac{1}{2}|\Com(\pi)|$.
\end{lemma}
\begin{proof}
If \textsc{\DetNewTrAlg}$(\pi)$ does not perform any recursive call at all, then the statement of the lemma holds (because the antecedent is false).
Let us assume that \textsc{\DetNewTrAlg}$(\pi)$ performs recursive calls, and let $U_\pi$ be the set $U$ computed by \textsc{\DetNewTrAlg}$(\pi)$ at line~\ref*{line:calcolo-U}.
Set $U_\pi$ can be either empty or not, and we first consider the case in which it is non\nbdash-empty.

Let us focus on the recursive calls performed at line~\ref*{line:call-ric-firstType}.
By the definition of $U_\pi$, if $v\in U_\pi$, then $\varepsilon_v^{\Com(\pi)}\geq\frac{1}{2}$, and thus $|\Com(\pi')|\leq\frac{1}{2}|\Com(\pi)|$, for the assignment $\pi'=\pi\addassign\assign{\emptyset,\{v\}}$ (see \cref{lemma:riduzione_taglia_assignment_tree}).

Now, let us consider the calls at line~\ref*{line:call-ric-thirdType}.
If \textsc{\DetNewTrAlg}$(\pi)$ arrives at that stage of its execution, it means that all vertices belonging to $U_\pi$ have been already included (at line~\ref*{line:inclusione-U}).
Let us denote by $\sigma=\pi\addassign\assign{U_\pi,\emptyset}$ the assignment encoded by sets $\IncAlg$ and $\ExcAlg$ at that stage of the execution flow.
Observe that the set of the excluded vertices is the very same of that at the beginning of the execution of the call \textsc{\DetNewTrAlg}$(\pi)$, implying that $\Com(\sigma)=\Com(\pi)$.

Since the algorithm is in the loop at lines \ref*{line:inizio-test-thirdType}\nbdash--\ref*{line:fine-test-thirdType}, it is taking into consideration recursive calls on assignments $\pi'=\sigma\addassign\assign{\{v\},G\setminus\{v\}}$, where $G\in\Sep(\sigma)$ and $v\in G$ is a free vertex in $\sigma$.
Note that $\varepsilon_v^{\Com(\sigma)}<\frac{1}{2}$ because $\Com(\sigma) = \Com(\pi)$ and $v\notin U_\pi$.
Hence, by the intersection property and \cref{lemma:riduzione_taglia_assignment_tree}, it follows that $|\Com(\pi')|\leq\frac{1}{2} |\Com(\pi)|$ for every assignment on which a recursive call is invoked at line~\ref*{line:call-ric-thirdType}.

To conclude, if $U_\pi$ is empty, then only few things change.
Indeed, in this case, there is no recursive call at all at line~\ref*{line:call-ric-firstType}, and for the recursive calls at line~\ref*{line:call-ric-thirdType} a similar discussion to the that above proves that $|\Com(\pi')|\leq \frac{1}{2}|\Com(\pi)|$ (simply note that, being $U_\pi$ empty, then $\sigma=\pi$, $\Com(\sigma) = \Com(\pi)$, $v\notin U_\pi$, and hence, again, $\varepsilon_v^{\Com(\sigma)}<\frac{1}{2}$, and \cref{lemma:riduzione_taglia_assignment_tree} applies).
\end{proof}

\begin{lemma}\label{lemma:ramo_ricorsivo_logaritmico}
Let $\G$ and $\H$ be two hypergraphs satisfying the intersection property, and let $\pi$ be an assignment.
Then, the maximum depth of every recursive path traversed by \textsc{\DetNewTrAlg}$(\pi)$ is $O(\log |\Com(\pi)|)$.
\end{lemma}
\begin{proof}
By \cref{lemma:dimezzamento_ricorsivo}, it follows that the size of set $\ComAlg$ halves at every recursion step, and this happens for all the recursive calls.
Since the procedure \textsc{\DetNewTrAlg} is correct (\cref{lemma:esecuzione_generale_corretta}), every recursive path leading to a \valfalse answer has to return its (\valfalse) answer before $\ComAlg$ becomes empty, for otherwise an incorrect answer would be returned by \textsc{\DetNewTrAlg}$(\pi)$.
Therefore, the depth of every recursive path leading to a \valfalse answer is logarithmic in $|\Com(\pi)|$.

Clearly, also the depth of a recursive path leading to a \valtrue answer is logarithmic in the size of $\Com(\pi)$ (again, due to the halving size of $\ComAlg$).
\end{proof}

\begin{theorem}\label{theo:tempo_esecuzione_algoritmo_hitting-property}
Let $\G$ and $\H$ be two hypergraphs satisfying the intersection property.
Then, the time complexity of the algorithm \textsc{\Alg} is $O(\InputSize^{O(\log \InputSize)})$.
\end{theorem}
\begin{proof}
Checking whether hypergraphs $\G$ and $\H$ are simple and satisfy the intersection property (line~\ref*{line:det-alg_checkSimpleHP}) is feasible in $O(\InputSize^2)$.

Regarding the procedure \textsc{\DetNewTrAlg}, the size of set $U$, computed at line~\ref*{line:calcolo-U}, is bounded by $|V|$.
Therefore there are no more than $O(\InputSize)$ recursive calls performed at line~\ref*{line:call-ric-firstType}.
Moreover, there are no more than $O(\InputSize^2)$ recursive calls performed at line~\ref*{line:call-ric-thirdType}, because, given any edge $G\in\G$, there are at most $O(\InputSize)$ different vertices belonging to $G$ and $|\G|$ is $O(\InputSize)$.

Hence, every call of the procedure \textsc{\DetNewTrAlg} performs $O(\InputSize^2)$ recursive calls.
By \cref{lemma:ramo_ricorsivo_logaritmico}, every recursive path has a depth $O(\log |\H|)$ that is also $O(\log \InputSize)$, hence at most $O(\InputSize^{2 O(\log \InputSize)})$ calls of \textsc{\DetNewTrAlg} are executed.
Since each call of \textsc{\DetNewTrAlg} executes in time $O(\InputSize^2)$
to perform its computations, the overall time complexity of the algorithm \textsc{\Alg} is $O(\InputSize^{2(O(\log \InputSize)+1)})$, which is $O(\InputSize^{O(\log \InputSize)})$.
\end{proof}

If we assume that the input hypergraphs $\G$ and $\H$ are guaranteed to fulfill $\G\subseteq\Tr(\H)$ and $\H\subseteq\Tr(\G)$ (i.e., we assume the stricter condition imposed on the input hypergraphs by \citet{Boros2009}), and are such that each vertex belongs to at least an edge in both $\G$ and $\H$, then we can carry out a finer time complexity analysis.
Note that, if a vertex $v$ is not included in any edge of $\G$ and $\H$, then $v$ is not relevant for the purpose of checking the duality of $\G$ and $\H$, because $v$ can be completely ignored.
On the other hand, if a vertex $v$ belongs to an edge in one hypergraph and, at the same time, $v$ does not belong to any edge of the other hypergraph, then $\G$ and $\H$ are not dual.
Indeed, assume, for example, that vertex $v$ belongs to an edge $G$ of $\G$ and does not belong to any edge of $\H$.
The set of vertices $G\setminus\{v\}$ is clearly a transversal of $\H$, therefore $G$ is not a minimal transversal of $\H$, and hence $\G\neq\Tr(\H)$.
We need the following properties.
\begin{lemma}\label{lemma:witness_distinti}
Let $\G$ be a hypergraph, and $T=\{v_1,\dots,v_t\}\subseteq V$ be a minimal transversal of $\G$.
Then, there are (at least) $|T|$ distinct edges $G_1,\dots,G_t$ of $\G$ such that $T\cap G_i=\{v_i\}$.
\end{lemma}
\begin{proof}
Let $W=\{G_1,\dots,G_p\}$ be the set of all the edges of $\G$ witnessing the criticality of vertices in $T$.
Since $T$ is a minimal transversal of $\G$, by \cref{lemma:transMin-ogniVertexCritico}, for each vertex $v\in T$ there is at least one edge in $W$ witnessing $v$'s criticality.
Now, assume by contradiction that $p<t$.
Then, there must be two different vertices in $T$, say $v_i$ and $v_j$, and an edge $G_k\in W$ such that $\{v_i\}=T\cap G_k$ and $\{v_j\}=T\cap G_k$: a contradiction, because $v_i$ and $v_j$ are assumed to be different.
\end{proof}

\begin{corollary}\label{corol:bound_size_transversal}
Let $\G$ be a hypergraph, and let $T$ be a minimal transversal of $\G$.
Then, $|T|\leq |\G|$.
\end{corollary}
\begin{proof}
By \cref{lemma:witness_distinti}, for every minimal transversal $T$ of size $t$ there are at least $t$ distinct edges of $\G$ witnessing the criticality of the vertices of $T$.
So, a minimal transversal of $\G$ cannot have more vertices than the number of edges of $\G$, for otherwise there would not be enough criticality's witnesses.
\end{proof}

\begin{lemma}\label{lemma:bound_nodi}
Let $\G$ and $\H$ be two hypergraphs such that $\G\subseteq\Tr(\H)$ or $\H\subseteq\Tr(\G)$ and such that each vertex $v\in V$ belongs to an edge in $\G$ and to an edge in $\H$.
Then, $|V|\leq|\G|\cdot|\H|$.
\end{lemma}
\begin{proof}
If $\G\subseteq\Tr(\H)$, then from \cref{corol:bound_size_transversal} it follows that, for every edge $G\in\G$, $|G|\leq|\H|$, because $G$ is a minimal transversal of $\H$.
By summing these relations over all edges of $\G$, we obtain $\sum_{G\in\G}|G|\leq|\G|\cdot|\H|$, which, combined with $|V|\leq \sum_{G\in\G}|G|$ (because all vertices belong to at least an edge of $\G$), proves the statement.
If $\H\subseteq\Tr(\G)$, then the statement follows by symmetry.
\end{proof}

\begin{theorem}\label{theo:tempo_esecuzione_algoritmo_intersection-property}
Let $\G$ and $\H$ be two hypergraphs such that $\G\subseteq\Tr(\H)$ and $\H\subseteq\Tr(\G)$ and such that each vertex $v\in V$ belongs to an edge in $\G$ and to an edge in $\H$.
Then, the time complexity of algorithm \textsc{\Alg}, in which the condition of $\G$ and $\H$ being such that $\G\subseteq\Tr(\H)$ and $\H\subseteq\Tr(\G)$ is checked, instead of them being simple and satisfying the intersection property, is $O((|\G|\cdot|\H|)^{O(\log |\H|)})$.
\end{theorem}
\begin{proof}
To prove the statement we need a similar proof to that of \cref{theo:tempo_esecuzione_algoritmo_hitting-property}.
Just observe that verifying that $\G$ and $\H$ are such that $\G\subseteq\Tr(\H)$ and $\H\subseteq\Tr(\G)$ is feasible in $O(\InputSize^2)$ and verifying that each vertex is included in an edge in $\G$ and in an edge in $\H$ is feasible in $O(\InputSize)$.
Moreover, by \cref{lemma:bound_nodi}, the size of the set $U$, computed at line~\ref*{line:calcolo-U}, is bounded by $|\G|\cdot|\H|$.
In addition, there are no more than $|\G|\cdot|\H|$ recursive calls performed at line~\ref*{line:call-ric-thirdType}, because given any edge $G\in\G$, from $\G\subseteq\Tr(\H)$ and \cref{corol:bound_size_transversal} it follows that there are at most $|\H|$ different vertices belonging to $G$.

Therefore, every call of procedure \textsc{\DetNewTrAlg} performs $O(|\G|\cdot|\H|)$ recursive calls.
Since each call executes actually in time $O(|\G|\cdot|\H|)$
to perform its computations, the overall time complexity of the algorithm \textsc{\Alg} (with the modifications mentioned in the statement of the lemma) is $O((|\G|\cdot|\H|)^{O(\log |\H|)+1})$, which is $O((|\G|\cdot|\H|)^{O(\log |\H|)})$.
\end{proof}

It is now interesting to focus on the following fact.
By the proof of \cref{theo:tempo_esecuzione_algoritmo_hitting-property} (see the extended technical report~\cite{Gottlob_Malizia:DUAL_arXiv}), the leaves of the recursion tree of procedure \textsc{\DetNewTrAlg}, which are candidate to certify the existence of a new transversal of $\G$, are $O(\InputSize^{2 O(\log \InputSize)})$, but in principle there could be an exponential number of new minimal transversals of $\G$ \Wrt $\H$.
For example, let us consider the class of pairs of hypergraphs ${\{\tuple{\G_i,\H_i}\}}_{i\geq 1}$, defined as follows: $V_i=\{x_1,y_1,\dots,x_i,y_i\}$, $\G_i=\{\{x_j,y_j\}\mid 1\leq j\leq i\}$, and $\H_i=\{\{x_1,\dots,x_i\},\{y_1,\dots,y_i\}\}$.
For every $i\geq 1$, hypergraphs $\G_i$ and $\H_i$ satisfy the intersection property, $|\G_i|=i$, $|\H_i|=2$, and the number of minimal transversals of $\G_i$ missing in $\H_i$ is $\Theta(2^i)$.
So, it is not possible that each leaf of the recursion tree of \textsc{\DetNewTrAlg} identifies a unique new minimal transversal of $\G$.
For this reason we want to know what new transversals of $\G$ are identified by \textsc{\DetNewTrAlg} when it answers \valtrue.

\begin{theorem}\label{theo:algoritmo_trova_minimali}
Let $\G$ and $\H$ be two hypergraphs.
If \textsc{\DetNewTrAlg}$(\emptyassign)$ answers \valtrue, then the witness $\sigma$ on which the procedure has answered \valtrue at the end of its recursion is coherent with a new \emph{minimal} transversal of $\G$ \Wrt $\H$.
\end{theorem}
\begin{proof}
Since $\sigma$ is a witness, $\sigma$ must be coherent with a new transversal $\wt{T}$ of $\G$.
Let us assume by contradiction that $\sigma$ is not coherent with any minimal one, and let $\wh{T}$ be any new minimal transversal \emph{strictly} contained in $\wt{T}$.

Let $(\pi^0,\pi^1,\dots,\pi^k)$ be the sequence of the assignments successively considered by the stack of the recursive calls of \textsc{\DetNewTrAlg} resulting in the construction of the witness $\sigma$, where $\pi^0=\emptyassign=\assign{\emptyset,\emptyset}$, $\pi^k=\sigma$, and $\pi^i=\assign{\In^i,\Ex^i}$ for all $0\leq i\leq k$.

Remember that assignments $\pi^i$ are such that $\pi^\ell\dblsub\pi^{\ell+1}$, for every $0\leq \ell\leq k-1$ (see \cref{lemma:esecuzione_singola_corretta}).
This implies that, for every $0\leq i\leq k$, assignment $\pi^i$ is such that $\pi^i\dblsubeq\pi^k=\sigma\dblsubeq \wt{T}$, and hence $\Ex^i\cap \wh{T}=\emptyset$ (because $\Ex^i\cap \wt{T}=\emptyset$ and $\wh{T}\subset \wt{T}$).
For this reason, in order for $\pi^k=\sigma$ not to be coherent with $\wh{T}$, it must be the case that $\In^k\not\subseteq \wh{T}$.
Since $\emptyset=\In^0\subseteq \wh{T}$, $\In^k\not\subseteq \wh{T}$, and $\In^\ell\subseteq \In^{\ell+1}$, for every $0\leq \ell \leq k-1$, there must be an index $j$, with $0\leq j\leq k-1$, such that $\In^j\subseteq \wh{T}$ and $\In^{j+1}\not\subseteq \wh{T}$.

Let $v$ be a vertex belonging to $\In^{j+1}\setminus \wh{T}$.
This means that $v\notin \wh{T}$, and from $\In^{j}\subseteq \wh{T}$ it follows that $v\notin \In^{j}$.
Therefore, vertex $v$ is included during the call \textsc{\DetNewTrAlg}$(\pi^j)$ either at lines \ref*{line:inizio-inclusione-U}\nbdash--\ref*{line:fine-inclusione-U}, or during a recursive call performed at line~\ref*{line:call-ric-thirdType}.

First consider the latter case.
If $v$ were included at line~\ref*{line:call-ric-thirdType}, then it would be included as a critical vertex along with an edge of $\G$, say $G$, chosen at that moment, witnessing the criticality of $v$ in $\In^{j+1}$.
This means that $v\in \In^{j+1}$, and $(G\setminus\{v\})\subseteq \Ex^{j+1}$.
From $\pi^{j+1}\dblsubeq\pi^k\dblsubeq \wt{T}$ it follows that $v\in \wt{T}$ and that $\wt{T}\cap (G\setminus\{v\})=\emptyset$ (because $\wt{T}\cap \Ex^{j+1}=\emptyset$), and hence that $\wt{T}\cap G=\{v\}$.
Therefore, $v$ is critical in $\wt{T}$.
However, $v\notin \wh{T}$ and we are assuming that $\wh{T}$ is a minimal transversal of $\G$ such that $\wh{T}\subset \wt{T}$: a contradiction, because $v$ is at the same time critical and non\nbdash-critical in $\wt{T}$.
As a consequence, $v$ has to be included at lines \ref*{line:inizio-inclusione-U}\nbdash--\ref*{line:fine-inclusione-U}, as part of the set $U_{\pi^j}$ during the call \textsc{\DetNewTrAlg}$(\pi^j)$.

Since the execution flow of \textsc{\DetNewTrAlg}$(\pi^j)$ goes beyond line~\ref*{line:fine-test-firstType}, all the recursive calls performed at line~\ref*{line:call-ric-firstType} return \valfalse.
Among those calls, also \textsc{\DetNewTrAlg}$(\pi^j\addassign\assign{\emptyset,\{v\}})$ is performed (because $v\in U_{\pi^j}$), and the fact that it returns \valfalse implies, by \cref{lemma:esecuzione_generale_corretta}, that there is no new transversal of $\G$ \Wrt $\H$ coherent with $\pi^j\addassign\assign{\emptyset,\{v\}}$.
However, observe that $\In^j\subseteq \wh{T}$, and that $(\Ex^j\cup\{v\})\cap \wh{T}=\emptyset$ (because $\Ex^j\cap \wh{T}=\emptyset$ and $v\notin \wh{T}$).
Hence, $\wh{T}$ is a new transversal of $\G$ coherent with $\pi^j\addassign\assign{\emptyset,\{v\}}$: a contradiction, because we are assuming that \textsc{\DetNewTrAlg}$(\pi^j\addassign\assign{\emptyset,\{v\}})$ returns \valfalse.
Thus $\sigma$ is coherent with a new minimal transversal of $\G$ \Wrt $\H$.
\end{proof}

Let $\G$ and $\H$ be two hypergraphs.
Assume that \textsc{\DetNewTrAlg}$(\emptyassign)$ answers \valtrue, and let $\sigma=\assign{\In,\Ex}$ be the assignment on which the procedure answers \valtrue at the end of its recursion.
There are two cases: (1) $\Sep(\sigma)=\emptyset$, or (2) $\Sep(\sigma)\neq\emptyset$.

In Case (1), by \cref{theo:algoritmo_trova_minimali}, $\sigma$ is coherent with a new minimal transversal of $\G$, and $\sigma$ is such that $\Sep(\sigma)=\emptyset$, hence it follows that $\In$ is a new minimal transversal of $\G$ \Wrt $\H$.

Consider Case (2).
Since $\sigma$ is a witness and $\Sep(\sigma)\neq\emptyset$, it must be the case that $\Com(\sigma) = \emptyset$, and hence $\Ex$ is a transversal of $\H$.
By this, $\compl{\Ex}$ is an independent set of $\H$.
Let us denote by $\mathit{Free}(\sigma) = V\setminus(\In\cup\Ex)$ the set of free vertices in $\sigma$.
Consider the set $T_\sigma=\{\In\cup S\mid S\in\Tr(\Sep(\sigma)^{\FreeSet(\sigma)})\}$.\footnote{With a slight abuse of notation, we here regard $\Sep(\sigma)$ as if it were a hypergraph. In this case, observe that $\Sep(\sigma)^{\mathit{Free}(\sigma)}$ corresponds to ${(\G_{V\setminus\In})}^{V\setminus(\In\cup\Ex)}$.}
Transversals $S$ of $\Sep(\sigma)^{\FreeSet(\sigma)}$, by definition, are such that $S\cap\Ex = \emptyset$, and hence $S\subseteq\compl{\Ex}$ which means that $S$ is an independent set of $\H$.
Therefore, from an adaptation of \cref{lemma:composition_property_transversals_expansion} to $\Sep(\sigma)$, it follows that $T_\sigma$ is a set of new transversals of $\G$.
Note that some of the elements of $T_\sigma$ are new minimal transversals of $\G$, while others are not minimal.
However, there is no guarantee that there are no more minimal transversals, i.e., there are new minimal transversals of $\G$ not belonging to $T_\sigma$.

\end{appendices}

\end{document}